\documentclass{llncs}

\usepackage[T1]{fontenc}
\usepackage[utf8]{inputenc}
\usepackage{epsfig,xcolor}
\usepackage{tabularx}
\usepackage{amsmath,amssymb}
\usepackage{latexsym,multirow,xspace}
\usepackage{multirow}
\usepackage{soul}
\usepackage{tipa}

\usepackage[text={15.4cm,19.5cm},centering]{geometry}  

\usepackage{algorithm}
\usepackage[noend]{algorithmic} % ----- our version
\begin{document}

\title{A methodology to design distributed algorithms for mobile entities:  the pattern formation problem as case study%
\thanks{The work has been supported in part by the Italian National Group for Scientific Computation (GNCS-INdAM).}
}

\author{%
Serafino Cicerone\inst{1},  
Gabriele Di Stefano\inst{1}, 
Alfredo Navarra\inst{2}
}

\institute{Dipartimento di Ingegneria e Scienze dell'Informazione e Matematica,
        Università  degli Studi dell'Aquila, I-67100 
        L'Aquila, Italy.
\email{serafino.cicerone@univaq.it},
\email{gabriele.distefano@univaq.it}
\and 
Dipartimento di Matematica e Informatica,
        Università degli Studi di Perugia I-06123 
        Perugia, Italy.
\email{alfredo.navarra@unipg.it} 
}

\maketitle

\begin{abstract}
Following the wide investigation in distributed computing issues by mobile entities of the last two decades, we consider the need of a structured methodology to tackle the arisen problems. The aim is to simplify both the design of the resolution algorithms and the writing of the required correctness proofs. We would encourage the usage of a common framework in order to help both algorithm designer and reviewers in the intricate work of analyzing the proposed resolution strategies. In order to better understand the potentials of our methodology, we consider the Pattern Formation (PF) problem approached in [Fujinaga et al. \emph{SIAM J. Comput.}, 2015] as case study. Since the proposed resolution algorithm has turned out to be inaccurate and also of difficult fixing, we design a new algorithm guided by the proposed methodology, hence fully characterizing the problem.
\end{abstract}


% ==================================================================
\newcommand{\U}{\mathcal{U}} % unsolvable conf
\newcommand{\A}{\mathbb{A}}  % generic algorithm
\newcommand{\I}{\mathcal{I}} % all the initial configurations
\newcommand{\IA}{\mathcal{I}_{\A}} % all the initial configurations

\newcommand{\Ann}{\mathit{Ann}} % annulus
\newcommand{\Rob}{\mathit{Rob}} % robots (in a given environment: line, circle, annulus, ...)

\newcommand{\Sector}{\mathit{Sector}}

\newcommand{\CT}{C^T} % parking circle at Top
\newcommand{\CB}{C^B} % parking circle at Bottom

\newcommand{\GoToC}{\mathtt{GoToC^T}} 
\newcommand{\Distmin}{\mathtt{Distmin}} 
\newcommand{\Circle}{\mathtt{CircleForm}} 
\newcommand{\Gathering}{\mathtt{Gathering}} 
\newcommand{\Leader}{\mathtt{Leader}} 

\newcommand{\wait}{\texttt{Wait}\xspace}
\newcommand{\look}{\texttt{Look}\xspace}
\newcommand{\compute}{\texttt{Compute}\xspace}
\newcommand{\move}{\texttt{Move}\xspace}
\newcommand{\mult}{\mathit{mult}\xspace}

\newcommand{\AS}{\mathit{AS}\xspace}
\newcommand{\dist}{dist\xspace}

\newcommand{\safe}{\mathit{safe}\xspace}

\newcommand*{\myqed}{\hfill\ensuremath{\blacktriangleright}}%

\newcommand{\arc}[1]{{%
  \setbox9=\hbox{#1}%
  \ooalign{\resizebox{\wd9}{\height}{\texttoptiebar{\phantom{A}}}\cr#1}}
}

% =================== dal papero su Methodology

\newcommand{\F}{\mathcal{F}}
\newcommand{\Unew}{\mathcal{U}}
\newcommand{\C}{\mathcal{C}}
\newcommand{\R}{\mathcal{R}}
\newcommand{\M}{\mathcal{M}}
\newcommand{\T}{\mathcal{T}}

\newcommand{\Aut}{\mathit{Aut}}

\newcommand{\Proc}{\mathit{Proc}}
\newcommand{\nil}{\mathit{nil}}

\newcommand{\pre}{\mathtt{pre}}
\newcommand{\post}{\mathtt{post}}

\newcommand{\false}{\mathtt{false}}
\newcommand{\true}{\mathtt{true}}

% =================== dal papero su Rigorous Approach ...

%\newcommand{\pf}{\ensuremath{\textsc{PF}}\xspace}
%\newcommand{\gath}{\ensuremath{\textsc{Gath}}\xspace}

%\newcommand{\apf}{\ensuremath{APF}\xspace}
%\newcommand{\epf}{\ensuremath{EPF}\xspace}

% ------------ modelli sincronia
\newcommand{\fsync}{{\sc FSync}\xspace}
\newcommand{\ssync}{{\sc SSync}\xspace}
\newcommand{\sasync}{{\sc SAsync}\xspace}
\newcommand{\async}{{\sc Async}\xspace}

% ------------ matematica spicciola
\newcommand{\angolo}{\sphericalangle}
\newcommand{\argmin}{\arg\!\min}
\newcommand{\argmax}{\arg\!\max}
\newcommand{\ang}[2]{\langle #1, #2 \rangle}
\newcommand{\mathsc}[1]{\text{\textsc{#1}}}
\newcommand{\Reali}{ \mathbb{R} } % num reali
\newcommand{\Ex}{\mathbb{E}}  % esecuzione algo
\newcommand{\Int}{\mathit{int}}  % interior of a region
\newcommand{\minview}{\textit{min\_view}}
\newcommand{\h}{\mathit{H}}

\newcommand{\halfline}{\mathit{hline}}
\newcommand{\Line}{\mathit{line}}

% ------------ nomi problemi

\newcommand{\pf}{\mathsc{PF}\xspace}
\newcommand{\gath}{\mathsc{Gath}\xspace}
\newcommand{\SB}{\mathsc{SB}\xspace}
\newcommand{\RS}{\mathsc{RS}\xspace}
\newcommand{\PPF}{\mathsc{PPF}\xspace}
\newcommand{\Fin}{\mathsc{Fin}\xspace}
\newcommand{\SC}{\mathsc{SC}\xspace}
\newcommand{\Term}{\mathsc{Term}\xspace}
%\newcommand{\Member}{\mathsc{Member}\xspace}

% ------------ blocchi e commenti
\newcommand{\blocco}[1]{\medskip\noindent\textit{#1. }}
\newcommand{\com}[1]{{\color{gray}{\small #1}}}
\newcommand{\blue}[1]{{\color{blue}{\small #1}}}

\newcommand{\linecomment}[2]{ { \noindent \color{red}{\small [$\bullet$ \textsc{#1}: #2]}}}
%\newcommand{\linecomment}[2]{{}}

% ------------ altro

\newcommand{\DecProb}[3]{%
\
\par\smallskip\noindent
\begin{tabularx}{\columnwidth}{@{\hspace{1mm}}l@{\hspace{1mm}}X@{\hspace{1mm}}}
\hline & \\[-4mm]
\multicolumn{2}{c}{\hspace{-2mm}\textsc{\small #1}\hspace*{-2mm}}\\[0.5mm]
\hline \ & \\[-3mm]
{\textsc{given}:} & #2 \\
{\textsc{problem}:~} & #3 \\
\hline
\end{tabularx}
\par\smallskip\noindent
\
}

% ------------ nomi di predicati: tutti iniziano con 'x' per evitare conflitti con altri comandi latex 
\newcommand{\xg}{\mathtt{g}}
\newcommand{\xm}{\mathtt{m}}
\newcommand{\xa}{\mathtt{a}}
\newcommand{\xd}{\mathtt{d}}
\newcommand{\xf}{\mathtt{f}}
\newcommand{\xt}{\mathtt{t}}
\newcommand{\xu}{\mathtt{u}}
\newcommand{\xp}{\mathtt{p}}
\newcommand{\xw}{\mathtt{w}}
\newcommand{\xc}{\mathtt{c}}
\newcommand{\xduno}{\mathtt{d}_1}
\newcommand{\xddue}{\mathtt{d}_2}

\newcommand{\predue}{ \neg \xc }%\newcommand{\preuno}{ \neg \xc }
\newcommand{\pretre}{ \xa \wedge \neg \xc }%\newcommand{\predueuno}{ \xa \wedge \neg \xc }
\newcommand{\prequattro}{ \xa \wedge \neg \xc \wedge \xm }%\newcommand{\preduedue}{ \xa \wedge \neg \xc \wedge \xm }
\newcommand{\precinque}{ \neg \xc \wedge \xf }%\newcommand{\pretre}{ \neg \xc \wedge \xf }
\newcommand{\presei}{ \xa \wedge \neg \xc \wedge \xm \wedge \xt }%\newcommand{\prequattro}{ \xa \wedge \neg \xc \wedge \xm \wedge \xt }
\newcommand{\presette}{ \xa \wedge \neg \xddue \wedge \neg \xu }%\newcommand{\precinque}{ \xa \wedge \neg \xddue \wedge \neg \xu }
\newcommand{\preotto}{ \xa \wedge \neg \xduno \wedge \xu }%\newcommand{\preB}{ \xa \wedge \neg \xduno \wedge \xu }
\newcommand{\prenove}{ \neg \xm \wedge \xp }%\newcommand{\preC}{ \neg \xm \wedge \xp }
\newcommand{\predieci}{ \xg }%\newcommand{\preD}{ \xg }
\newcommand{\preundici}{ \xw }%\newcommand{\preF}{ \xw }

% ==================================================================

% ==================================================================
% Introduction
% ==================================================================
\section{Introduction}\label{sec:introduction}
In the last two decades there has been a rapid growth and development in the field of distributed computing by mobile entities. The aim is to study the computational and complexity issues arising in systems of decentralized entities required to accomplish global tasks. 
Depending on the entities' capabilities and the environment where the entities operate, one may ask which tasks can be performed, if not always under which conditions, and perhaps at what cost. However, one of the central questions and certainly the most investigated one, is to determine what are the minimal hypotheses that allow a given problem to be solved. 

Here we are interested in what is known in the literature as the \emph{Look}-\emph{Compute}-\emph{Move} model. In this model, entities from now on referred to as \emph{robots} operate in Look-Compute-Move (LCM) cycles.
In one cycle a robot takes a snapshot of the surrounding 
(Look). Accordingly, in the Compute phase it decides 
whether to move toward a specific target or not, and in the positive case 
it moves (Move). The accuracy or the information a robot acquires during the Look phase as well as its computing and moving skills depend on the assumed capabilities. 

A comprehensive survey about the state-of-art in this research area until 2012 can be found in~\cite{FPS12}. Very recently, a new book surveying on the advances under different settings has been released~\cite{FPS19}.

Although many high qualified researchers are involved and more and more sophisticated resolution strategies have been devised to face the arisen problems, still a structured methodology that could help in designing resolution algorithms is missing. 
The need of a methodology comes from three main observations: 1) the distributed environment might be very `hostile' in the sense that sometimes it is difficult to be sure one is considering all possible situations/events that may occur; 2) it is certainly desirable to have a list of bullets that guides and helps the design of a resolution algorithm along with the corresponding correctness proof; 3) actually there are in the literature several cases of claimed results that turned out to be only partially true or basically incorrect.
 
 To support 3) it is worth citing the detailed analysis reported in~\cite{CDN19} where well-established results like [Fujinaga et al., \emph{SIAM J. Comp.} 44(3), 2015]~\cite{FYOKY15}, more recent approaches like [Bramas et al., \emph{SSS} and \emph{PODC}, 2016]~\cite{BT16b,BT16} and ‘unofficial' results like [Dieudonn{\'{e}} et al., \emph{arXiv:0902.2851}]~\cite{DPVarx09}, that is the extended version of [Dieudonn{\'{e}} et al., \emph{DISC}, 2010]~\cite{DPV10}, revealed to require major technical revisions. Further examples can be found in [Doan et al., \emph{OPODIS}, 2017]~\cite{DBO17} where by means of model checking approaches  some imperfections or missing cases in [D'Angelo et al., \emph{Dist. Comp.} 27(4), 2014]~\cite{DDN14} have been shown; whereas [D'Emidio et al., \emph{Inf. Comput.} 263, 2018]~\cite{DDFN18} highlights some flaws arising from [Das et al. \emph{Theor. Comput. Sci.} 609, 2016]~\cite{DFPSY16}. Very recently, we also came across [Pattanayak et al. \emph{J. Parallel Distrib. Comput.} 123, 2019]~\cite{PMRM19} where the authors completely neglect to handle possible families of input symmetric configurations.\footnote{To provide some evidence of our assertion, we point out the reader to the discussion in~\cite{PMRM19} right after Theorem 3. It comes out that for instance configurations admitting more than one axis of symmetry are not considered, as well as configurations where the center of the Smallest Enclosing Circle of the robots is occupied by one robot are said to allow the election of a leader only if the total number of robots is even.}

The three motivations exposed above encouraged us to investigate on and to recommend the usage of a common framework in order to help both algorithm designer and reviewers in the intricate work of analyzing the devised resolution strategies. 
To this respect we propose a new methodology 
that highlights fundamental properties required to approach problems arising in distributed environments, helping in the design of new algorithms as well as on proving their correctness. It might also be useful to revise previous algorithms in order to better check their validity. 

This paper comes after a couple of attempts~\cite{CDN19,CDN18c} to provide formal and structured arguments to support the proposed resolution algorithms and their proofs, designed for specific problems. 
For instance, based on some arguments that here we revisit and extend, we could fix in~\cite{CDN18c} the algorithm first sketched in~\cite{CDN16}. 
We believe our investigation on a generalized and formal methodology is now mature to be proposed. In order to fully understand the potentials of our new methodology, we consider the Pattern Formation problem approached in~\cite{FYOKY15} as case study. 
As already outlined, the algorithm proposed in~\cite{FYOKY15} has turned out to be inaccurate. An attempt to provide a patch by the same authors can be found in~\cite{FYOKY17}.
However, by personnel communication the authors confirmed us that the algorithm cannot be easily fixed and that they give up with further attempts. 

Guided by the proposed methodology, we design a new algorithm that fully characterizes the considered Pattern Formation problem.
%\linecomment{ser}{da qualche parte bisognerebbe rimettere lo statement formale del teorema che ri-dimostriamo: \pf con chiralità sse $\rho(R)$ divide $\rho(F)$ ...}

%\linecomment{Alf}{Non lo metterei qui il teorema, anche perchè non abbiamo gli elementi per scriverlo, vedi summetricità. Inoltre l'enfasi è sulla metodologia.}

\subsection{Outline}
In the next section, we start by introducing the basic notions required to approach the distributed computing environment of  mobile robots we refer to. Section~\ref{sec:methodology} is the core of the paper as we present our detailed methodology to approach problems within the specified environment. Section~\ref{sec:notation} considers the Pattern Formation problem approached in~\cite{FYOKY15} as case study for our methodology. It introduces all additional assumptions and notation required by the definition of the specific problem and by our resolution algorithm. The formal definition of the resolution algorithm is then provided in Section~\ref{sec:algorithm}. According to the methodology, the algorithm is designed to solve various sub-problems whose composition leads to the resolution of PF. Actually we fully characterize the approached PF problem, and this is a main result on its own. An explanatory and extended example about the application of the algorithm in order to better highlight all the peculiarities of our methodology and of our new strategy is given in Section~\ref{sec:example}. Section~\ref{sec:correctness} contains the correctness proof of the proposed algorithm obtained by following the guidelines dictated by the methodology. Finally, Section~\ref{sec:concl} provides conclusive remarks, posing ideas for future investigation.

% ==================================================================
% Preliminaries
% ==================================================================
\section{Preliminaries}
Before starting presenting the methodology, we need to formalize some of the concepts already introduced and to specify some of the robots capabilities.
For instance, from now on we focus on robots moving in the Euclidean plane. Clearly all arguments we present to define our methodology can be easily extended to higher dimensions or to the case of robots moving in graphs.
Other assumptions are instead dictated by the request of the weakest hypothesis under which problems remain solvable.
As first set of weak assumptions, we consider robots to be:
\begin{itemize}
\item{Autonomous}: no centralized control;
\item{Dimensionless}: modeled as geometric points in the plane;
\item{Anonymous}: no unique identifiers;
\item{Oblivious}: no memory of past events;
\item{Homogeneous}: they all execute the same \emph{deterministic} algorithm;
\item{Silent}: no means of direct communication;
\item{Disoriented}: no common knowledge of any orientation (coordinate system, handedness, etc.);
\item{Non-rigid}: robots are not guaranteed to reach a destination within one move;
\end{itemize}

Further assumptions will be specified soon.

\subsection{LCM model}\label{sec:LCMmodel}
Each robot in the system has sensory capabilities allowing it to determine the location of other robots in the plane, relative to its own location.
 Each robot refers in fact to a \emph{Local Coordinate System} (LCS) that might be different from robot to robot. %For the sake of completeness, in the following we formalize the adopted model.
 The robots also have computational capabilities which allow them to compute the location where to move along with the whole trajectory to trace. Each robot follows an identical algorithm that is preprogrammed into the robot. 
This
algorithm may also provide some additional data that can be exploited during the computations. 
The behavior of each robot can be described according to the sequence of four states: \wait, \look, \compute, and \move. Such states form a computational cycle (or briefly a cycle) of a robot.
%Note that the LCS of a robot may change within different computational cycles.
 The operations performed by each robot $r$ in each state will be now described in more details.
\begin{enumerate}
\item \wait. The robot is idle. A robot cannot stay indefinitely idle. %Initially, all robots are in \wait. 
\item  \look. The robot observes the world by activating its sensors which will return a snapshot of the positions of all other robots with respect to its LCS. Each robot is viewed as a point. Hence, the result of the snapshot (i.e., of the observation) is just a set of coordinates in its LCS.
\item  \compute. The robot performs a local computation according to a deterministic algorithm $\A$ (we also say that the robot executes $\A$). The algorithm is the same for all robots, and the result of the \compute phase is a destination point along with a trajectory to reach it.
\item  \move. If the destination point is the current location of $r$, $r$ performs a $\nil$ movement (i.e., it does not move); otherwise it moves toward the computed destination along the computed trajectory. 
\end{enumerate}

When a robot is in \wait we say it is \emph{inactive}, otherwise it is \emph{active}. In the literature, the computational cycle is simply referred to as the \look-\compute-\move (LCM) cycle, as during the \wait phase a robot is inactive. 

Initially robots are inactive, but once the execution of an algorithm $\A$ starts - unless differently specified - there is no instruction to stop it, i.e., to prevent robots to enter their LCM cycles. Then, the \emph{termination} property for $\A$ can be stated as follows: once robots have reached the required goal by means of $\A$, from there on robots can perform only the $\nil$ movement. %the only movement that can occur or that algorithm $\A$ can allow is the $\nil$ movement. 
Sometimes termination is not even required as robots might be asked to execute infinite computations, e.g., perpetual exploration~\cite{BDP17,GKMNZ08}, patrolling~\cite{CDGJNRS19,CGKKKT17,KK15}.

Note that the LCS of a robot may change within different LCM cycles.

During the \look phase, robots can perceive \emph{multiplicities}, that is whether a same point
is occupied by more than one robot. The multiplicity detection capability might be \emph{local} or \emph{global}, depending whether the multiplicity is detected only by robots composing the multiplicity or by any robot performing the \look phase, respectively. Moreover, the multiplicity detection can be \emph{weak} or \emph{strong}, depending whether a robot can detect only the presence of a multiplicity or if it perceives the exact number of robots composing the multiplicity, respectively.

About movements, a strong assumption is about the so-called \emph{rigid} movements where robots are always guaranteed to reach the destination within one LCM cycle.
A weaker assumption is what we consider, that is about \emph{non-rigid} movements: the distance traveled within a move is neither infinite nor infinitesimally small. More precisely, we can assume an adversary that has the power to stop a moving robot before it  reaches its destination. However, there exists an unknown constant $\nu > 0$ such that if the destination point is closer than $\nu$, the robot will reach it, otherwise the robot will be closer to it of at least $\nu$. Note that, without this restriction on $\nu$, an adversary would make it impossible for any robot to ever reach its destination. 

We assume that cycles are performed according to the weakest Asynchronous scheduler (\async): the robots are activated independently, and the duration of each phase is finite but unpredictable (the activation of each robot can be thought as decided by the adversary). As a result, robots do not have a common notion of time. Moreover, according to the definition of the \look phase, a robot does not perceive whether other robots are moving or not. Hence, robots may move based on outdated perceptions. In fact, due to asynchrony, by the time a robot takes a snapshot of the configuration, this might have drastically changed once the robot starts moving. The scheduler determining the cycles timing is assumed to be fair, that is, each robot becomes active and performs its cycle within finite time and infinitely often. Figure~\ref{fig:models-a} compares the \async scheduler with the other scheduler proposed in the literature. In the figure, the \wait state is implicitly represented by the time while a robot is inactive. In particular, it shows that in the Fully-synchronous (\fsync) scheduler all robots are always active, and the activation phase can be logically divided into global rounds:
 for all $i\geq 1$, all robots start the $i$-th LCM cycle simultaneously and synchronously execute each phase. %\look, \compute, and \move.

The Semi-synchronous (\ssync) scheduler coincides with the \fsync model, with the only difference that 
some robots may not start the $i$-th LCM cycle for some $i$ (some of the robots might be in the \wait state), but all of those who have started the $i$-th cycle synchronously execute each phase.% \look, \compute, and \move.

The Semi-asynchronous (\sasync) still maintains a sort of synchronous behavior as each phase lasts the same amount of time, but robots can start their LCM cycles at different times. It follows that while a robot is performing a \look phase, other active robots might be performing the \compute or the \move phases. 

Clearly, the four synchronization schedulers induce the following hierarchy (see, e.g.~\cite{CDN18a,DFPSY16,DDFN18}): \fsync robots are more powerful (i.e. they can solve more tasks) than \ssync robots, that in turn are more powerful than \sasync  robots, that in turn are more powerful than \async robots. This simply follows by observing that  the adversary can control more parameters in \async than in \sasync, and it controls more parameters in \sasync than in \ssync and \fsync. In other words, protocols designed for \async robots also work for \sasync, \ssync and \fsync robots. Contrary, any impossibility result stated for \fsync robots also holds for \ssync, \sasync and \async robots.  

\begin{figure}
\begin{center}
\scalebox{0.6}{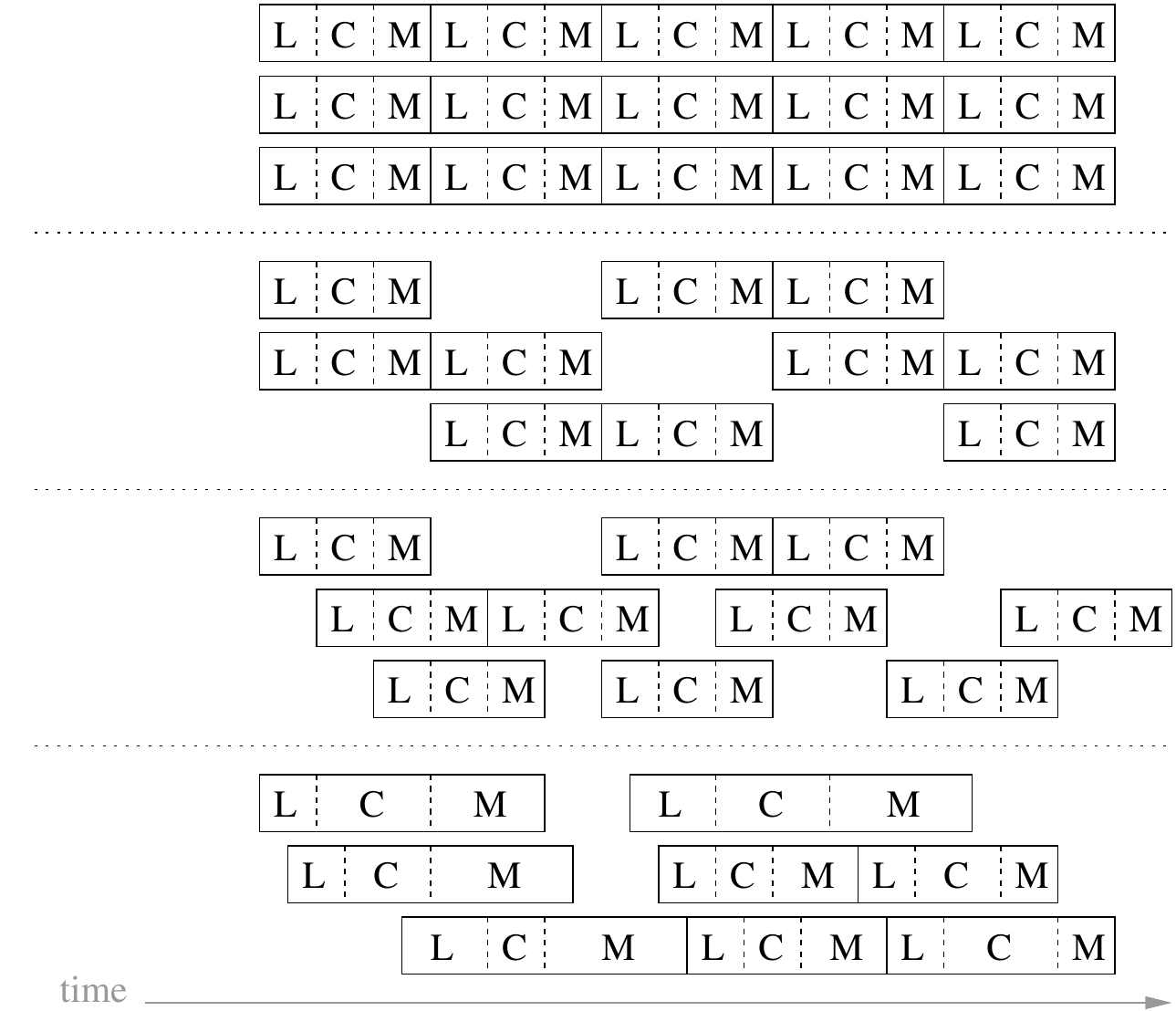 }
\end{center}
\caption{The execution model of computational cycles for each of \fsync, \ssync, \sasync and \async\ robots. The \wait state is implicitly represented by empty time periods.}\label{fig:models-a}
\end{figure}

\subsection{Robot Model}\label{sec:intro-model}
The robot model is mainly borrowed from~\cite{CDN19,CFPS12,FYOKY15}. We consider a system composed by a set of $n$ mobile \textit{robots}. 
Let $\Reali$ be the set of real numbers, at any time the multiset $R=\{r_1,r_2,\ldots,r_n\}$, with $r_i\in \mathbb{R}^2$,  contains the \textit{positions} of all the robots. By abusing notation, we often refer to $r\in R$ as a \emph{robot} instead of a robot position.

We arbitrarily fix an $x$-$y$ coordinate system $Z_0$ and call it the \emph{global coordinate system}. A robot, however, does not have access to it. It is used only for the purpose of description, including for specifying the input. All actions taken by a robot are done in terms of its local (and current) $x$-$y$ coordinate system, whose origin always indicates its current position. Let $r_i(t)\in \Reali^2$ be the location of robot $r_i$ (in $Z_0$) at time $t$. Then a multiset $R(t) = \{r_1(t), r_2(t), \ldots, r_n(t)\}$ is called the \emph{configuration} of $R$ at time $t$ (and we simply write $R$ instead of $R(t)$ when we are not interested in any specific time). 

Each robot $r_i$ has a LCS $Z_i$, where the origin always coincides with its current location.  
Let $Z_i(p)$ be the coordinates of a point $p\in \Reali^2$ in $Z_i$. If $r_i$ takes a time interval $[t_0,t_1]$ for performing the \look phase, then it obtains a multiset $Z_i(R(t)) = \{Z_i(r_1(t)),Z_i(r_2(t)),...,Z_i(r_n(t))\}$ for some $t\in [t_0, t_1]$, where $Z_i(r_i(t)) = (0, 0)$. That is, $r_i$ has the global-strong multiplicity detection ability.% and can count the number of robots sharing a location. 
\footnote{Although our methodology might be easily extended to  weaker capabilities with respect to the multiplicity detection and most importantly concerning the visibility of the robots,  we prefer to maintain such assumptions for the easy of the discussion and because of the chosen case study.}

%Let $\{t_i : i = 0,1,\ldots\}$ be the set of time instants at which any robot takes the snapshot $R(t_i)$ during its \look phase. %Without loss of generality, we assume $t_i = i$ for all $i = 0,1,\ldots$. Then, a
%An infinite sequence $\Ex : R(t_0),R(t_1),\ldots$ is called an execution with $R(t_0)$ as initial configuration.

\subsubsection{Symmetric configurations}\label{ssec:symm-conf}
In the Euclidean plane, a map $ \varphi :\Reali^2 \rightarrow  \Reali^2$ is called \emph{isometry} or distance preserving if for any $a,b \in \Reali^2$ one has $d(\varphi(a),\varphi(b))=d(a,b)$, where $d()$ denotes the standard Euclidean distance function. 
Examples of isometries in the plane are \emph{translations}, \emph{rotations} and \emph{reflections}. An isometry $\varphi$ is a translation if there exists no point $x$ such that $\varphi(x)=x$;
it is a rotation if there exists a unique point $x$ such that $\varphi(x)=x$ (and $x$ is called \emph{center of rotation}); it is a reflection if there exists a line $\ell$ such that $\varphi(x)=x$ for each point $x\in \ell$ (and $\ell$ is called \emph{axis of symmetry}). 

Given an isometry $\varphi$ different from the identity, the \emph{cyclic subgroup} of order $p$ generated by $\varphi$ is given by $\{\varphi^0, \varphi^1=\varphi\circ\varphi^0, \varphi^2=\varphi \circ \varphi^1, \ldots, \varphi^{p-1} = \varphi\circ \varphi^{p-2}\}$, where $\varphi^0$ is the identity automorphism, $\varphi^i \neq \varphi^0$ for each $0<i<p$, and $\varphi^p = \varphi^0$.
A reflection always generates a cyclic subgroup  of order $p=2$. Whereas, the cyclic subgroup generated by a rotation can be of any finite order $p>1$.

An \emph{automorphism} of a configuration $R$ is an isometry in the plane that maps robots into robots (i.e., points of $R$ into $R$). The set of all automorphisms of $R$ forms a group with respect to the composition denoted by $\Aut(R)$ and called \emph{automorphism group} of $R$. In general (i.e., for robots completely disoriented), the isometries in $\Aut(R)$ are the identity, rotations, reflections and their compositions (translations are not possible as $R$ contains a finite number of elements). If $|\Aut(R)|=1$, that is $R$ admits only the identity automorphism, then $R$ is said to be \emph{asymmetric}, otherwise it is said to be \emph{symmetric} (i.e., $R$ admits rotations or reflections). 
%\linecomment{gab}{dire qua che nel nostro lavoro non sono presenti le riflessioni poiché abbiamo la chiralità?}
%\linecomment{Alf}{Qui ancora non la introduciamo la chiralità dato che è un'assunzione del caso di studio e non della metodologia.}

If a configuration $R$ is symmetric due to an automorphism $\varphi$, two robots $r$, $r'\in R$ are \emph{equivalent} if $r'=\varphi(r)$. As a consequence, no algorithm can distinguish between two equivalent robots, and then it cannot avoid that the two \async robots start the computational cycle simultaneously. In such a case, there might be a so called \emph{pending move}, that is one of the two robots performs its entire computational cycle while the other has not started or not yet finished its \move phase, i.e. its move is pending. Clearly, any other robot is not aware whether there is a pending move, that is it cannot deduce such an information from the snapshot acquired in the \look phase. This fact greatly increases the difficulty to devise algorithms for symmetric configurations.

\subsubsection{Robots' view}\label{ssec:robot-view}
According to the capabilities of the robots, by opportunely elaborating the configuration perceived with respect to its own LCS, a robot obtains what will be later called the \emph{view} of a robot. Actually, sometimes a robot is asked to evaluate what would be the view of other robots, hence it is convenient that the view does not depend on the current LCS, as this might be completely different from cycle to cycle and from robot to robot. Hence, unless further knowledge is provided to the robots, the view should exploit only the information that all robots can equally perceive, like those concerning relative distances and angles among robots' positions.
It follows that in general, in a symmetric configuration there are robots with the same view. 
For instance, by considering a configuration with a multiplicity, then the view cannot discriminate among the robots composing the multiplicity, i.e. a configuration with a multiplicity is always perceived as symmetric.
Instead, in a symmetric configuration $R$ without multiplicities, in the stronger model with robots aware of $Z_0$, $R$ can be perceived as asymmetric by the robots as the view may exploit the coordinates of the robots to discriminate among all of them (as if they had unique identifiers).

% ==================================================================
% Methodology
% ==================================================================
\section{Methodology}\label{sec:methodology}
We now have all the ingredients necessary to present our new methodology. The main advantages will be to assist for
(1) designing a distributed algorithm $\A$ for solving a problem $\Pi$, and (2) proving that $\A$ is correct.

For the ease of discussion here we focus on the so-called formation problems where the goal to achieve is that of reaching a  disposal of the robots that satisfies a specified property.

Let $\R$ be the set of all the possible configurations and consider the following general robot-based computing problem: 

\begin{itemize}
\item
Let $\Pi$ be a problem that takes as input a configuration $R$ belonging to the set $\I\subseteq \R$ (the set of all initial configurations) and some static data $D$ (a description of the goal to be achieved along with other possible input data) and asks to transform $R$ into any configuration $F\in\F(D)\subseteq \R$, where $\F(D)$ is the set of final configurations for $\Pi$ with respect to $D$.%
%\footnote{We do not assume $\Pi$ to accomplish infinite computations (e.g., patrolling). Hence, $\F(D)$ is always defined.
%}
\end{itemize}
%
%\begin{example}\label{ex:instances-of-Pi}
%If $\Pi$ corresponds to the Pattern Formation problem considered in~\cite{FYOKY15}, $\I$ is the set of all configurations without multiplicities,  $D$ is a set of points in $\Reali^2$ describing the pattern to be formed and $\F(D)$ is the set of all the configurations similar \linecomment{ser}{qui similar ancora non è stato definito -- fixare} to the pattern. When $\Pi$ corresponds to the Gathering problem considered in~\cite{FPS12} (a specific case of the Pattern Formation problem), $D$ is any point in $\Reali^2$.
%

\smallskip
As examples, consider the cases in which $\Pi$ corresponds to the Pattern Formation (\pf) problem~\cite{SY99}, or to the Gathering (\gath) problem~\cite{FPS12}. For both problems, the set $\I$ contains all configurations whose elements are distinct (i.e., no multiplicity occurs). %,
%\footnote{Throughout this paper, we assume that any initial configuration contains no multiplicity. This is a typical assumption since it is impossible to break up multiple robots on a single position as we assume all robots execute the same algorithm}
%%and composed of robots in the \wait state. 

The \pf problem can be defined as follows:
\begin{itemize}
\item
The set $D$ just contains a representation of the final configuration $F$ to be obtained. In particular, given a multiset $F$ of $n$ points in $\Reali^2$ expressed as $Z_0(F)$, we say that an algorithm $\A$ \emph{forms} $F$ from an initial configuration $R\in \I$ composed of $n$ robots if for each possible execution %$\Ex : R=R(t_0),R(t_1),\ldots$, 
there exists a time instant $t>0$ in which $R(t)$ is similar%
\footnote{Let $P_1$ and $P_2$ be two multisets of points: if $P_2$ can be obtained from $P_1$ by uniform scaling, possibly with additional translation, rotation and reflection, then $P_2$ is \emph{similar} to $P_1$.}
  to $F$ % and no robots move after $t$, i.e., $R(t') = R(t)$ hold for each integer $t'\ge t$
 and $\A$ terminates (i.e., $R(t') = R(t)$ hold for each integer $t'\ge t$). The set $F$ is called the \emph{pattern.}
It follows that for the \pf problem $\F(D)\equiv \F(F)$ is the set containing all the configurations similar to the pattern $F$. 
\end{itemize}
The \gath problem is a special case of the \pf problem: it is characterized by a multiset $F$ containing one element with multiplicity $n$, thus consisting in making the robots to form a single point. Other possible formation problems might require to reach a configuration where some property holds, like for instance that no three robots are aligned, hence $\F(D)$ would be composed by the set containing all such configurations.

\smallskip
Depending on $\Pi$, there could exist a set of configurations $\Unew(D)$ whose elements represent unsolvable configurations (i.e.,  $\Pi$ is unsolvable when  $R \in \Unew(D)$). In such a case, any algorithm $\A$ able to solve $\Pi$ must transform any element of $\I\setminus \Unew(D)$ into $F \in \F(D)$. 
%
%\begin{example}\label{ex:unsolvable}
%%%%%%%%%%%%%%%%%%%%%%%%%%%% OLD %%%%%%%%%%%%%%%%%%%%%%%%%%
%With respect to the \pf problem, $\Unew(D)$ is known to contain the set of all configurations $R$ such that the \emph{symmetricity} of the initial configuration (i.e., $\rho(R)$) does not divide the symmetricity of the pattern to be formed (i.e., $\rho(D)$); 
%%%%%%%%%%%%%%%%%%%%%%%%%%%%%%%%%%%%%%%%%%%%%%%%%%%%%%%%%%%
With respect to the \pf problem, the entire set $\Unew(D)$ has not been characterized so far. However, from~\cite{SY99} it is known that any initial configuration cannot admit symmetries that do not appear also in the final configuration, unless such symmetries can be broken;
in \gath, $\Unew(D)$ is any configuration with just two \async robots occupying different positions. %or two multiplicites. %A formal definition of the symmetricity $\rho(\cdot)$ will be provided in Section~\ref{ssec:symmetricity}.

%\linecomment{Alf}{Modificato, si può parlare di simmetricità soltanto quando sussiste chiralità ma la introdurrei soltanto nel case study, visto anche che è stata rimossa dal titolo.}
%\myqed
%\end{example}

% -----------------------------------------------
\subsection{Problem decomposition into tasks}\label{ssec:methodology:decomposition}
A single robot has rather  weak capabilities with respect to the general problem it is asked to solve along with other robots (we recall that robots have no direct means of communication). For this reason, any resolution algorithm $\A$ for a problem $\Pi$ should be based on a preliminary decompositional approach: $\Pi$ should be divided into a set of sub-problems so that each sub-problem is enough simple to be thought as a ``task'' to be performed by (a subset of) robots. This subdivision could require several steps before obtaining the definition of such simple tasks, thus generating a sort of hierarchical structure. Our methodology recommends the following preliminary steps:

\begin{itemize}
\item
Define a (hierarchical) decomposition of $\Pi$ into sub-problems. Each sub-problem should be easy enough to be solved by assigning a \emph{task} $T$ to robots;
\item In order to define a task $T$ in a rigorous way, $T$ should correspond to a well-defined movement for (a subset of) robots. In particular, $T$ should be defined according to:
\begin{itemize}
\item a subset $R’\subseteq R$ of moving robots, 
\item a trajectory $\tau$ for each robot in $R'$ defined as a curve having as starting point the position of a robot in $R’$, and as final point a target position defined according to the strategy.
\end{itemize}
\end{itemize}
These preliminary steps imply the following additional considerations:

\begin{itemize}
\item According to the LCM model, during the \compute phase each robot should be able to recognize the task to be performed just according to the configuration perceived during the \look phase and the input data $D$. This recognition could be performed by providing $\A$ with a predicate $P_i$ for each task $T_i$: such a predicate, if evaluated true with respect to the perceived configuration $R$ and the input data $D$, reveals to robots that $T_i$ is the task to be performed.

\item At each time instant $t$, exactly one task $T_i$ must be recognized; hence, predicates must be mutually exclusive.

\item In order to accomplish the designed tasks, it is possible that a resolution algorithm $\A$ generates configurations that are not in $\I$. The set containing all configurations taken as input or generated by $\A$ is denoted as $\IA$. Note that by definition $\I \setminus \Unew(D) \subseteq \IA$. Moreover, for sake of correctness, $\IA\cap \Unew(D) = \emptyset$ must hold (i.e., no unsolvable configurations are generated by $\A$). 
%We call the set of all configurations handled and solvable \linecomment{ser}{penso che 'solvable' vada tolto perché è condizione che deve essere rispettata dall'algortmo - penso convenga dire che $\I'$ contiene tutte le conf 'prese in input' e 'generate' dall'algoritmo durante la sua esecuzione o alla fine dell'esecuzione } by the algorithm as $\I'$. Clearly $\I \setminus \Unew(D) \subseteq \I'$. Moreover, $\I'\cap \Unew(D) = \emptyset$ must hold. % must not contain configurations in $\Unew(D)$. 
\end{itemize}

% --------------- Da recuperare nelle subsection successive ...
% \begin{itemize}
% \item There might be relationships between tasks (eg. a task $T_2$ can be performed only when another task $T_1$ is completed).
% \item Due to the adversary, it is possible that more than one computational cycle is necessary to the moving robots to complete a task $T$.
% \end{itemize}
% --------------- End: Da recuperare nelle subsection successive ...

Most of the concepts introduced in the above observations can be formalized according to the general computational schema reported in Algorithm~\ref{alg:compute}; this schema describes how the generic algorithm $\A$ works according to the proposed methodology.

\begin{algorithm}[h]
\caption{Compute}\label{alg:compute}
\begin{algorithmic}[1]
\begin{small}
\REQUIRE Configuration $R\in \IA$, Data $D$
\ENSURE  A trajectory $\tau$ for each moving robot
%\textbf{Input}: configuration $R\in \IA$, data $D$\\
\IF{ $P_1(R,D)$} 
   \STATE Call $\Proc_1$ and return $m_1: R\to \tau_1$ 
\ENDIF
\IF{ $P_2(R,D)$} 
   \STATE Call $\Proc_2$ and return $m_2: R\to \tau_2$ 
\ENDIF
%\COMMENT 
\STATE $\cdots$%\\[2mm] 
\IF{ $P_k(R,D)$} 
   \STATE Call $\Proc_k$ and return $m_k: R\to \tau_k$ 
\ENDIF
\IF{ $P_F(R,D)$} 
   \STATE return $m_{k+1}:R \to\nil$ 
\ENDIF
\end{small}
\end{algorithmic} 
\end{algorithm}

%%%%%%%%%%%%%%%%%%%%%%%%%%%%%%%%%%%%%%%%%%%%%%%%%%%%%%%%%%%%%%%%%%%%
% MODIFICARE A SECONDA DELL'IMPAGINAZIONE

Concerning the computational schema reported in Algorithm~\ref{alg:compute}%
, the following conditions apply:

%%%%%%%%%%%%%%%%%%%%%%%%%%%%%%%%%%%%%%%%%%%%%%%%%%%%%%%%%%%%%%%%%%%%
\begin{itemize}
\item
every $P_i$ is a predicate computable on the input that identifies the corresponding procedure $\Proc_i$ to be computed;

\item
$P_F$ (also identified as $P_{k+1}$) is the predicate characterizing configurations in $\F(D)$;

\item
for every possible input pair $R$, $D$, with $R\in \IA$, there exists a true predicate $P_i(R,D)$; 

\item  
in order to allow robots to exactly recognize the task to be performed, it must hold $P_i(R,D)\wedge P_j(R,D) =\false$, for $i\neq j$.

\item
$m_i$ is the move computed by procedure $\Proc_i$: it associates to each robot a trajectory belonging to the set $\tau_i$. Notice that, in general, only a subset $R_i$ of $R$ is involved in the task $T_i$, hence $m_i(r)=\nil$ for each $r\in R\setminus R_i$; moreover, $m_F(r)$ (also identified as $m_{k+1}(r)$) is always the $nil$ movement, for each $r\in R$;

%$m_i$ is the move computed by Procedure $\Proc_i$: it associates to each robot a trajectory belonging to the set $\tau_i$. Notice that, for $i\le k$, only a subset $R_i\neq \emptyset$ of $R$ is involved in the task $T_i$, 
% and there must exist at least one robot $r$ such that $m_i(r)\neq null$, whereas $m_i(r')=\nil$ for each $r'\in R\setminus R_i$.
%Concerning $m_F(r)$ (also identified as $m_{k+1}(r)$), it is always the $null$ movement, for each $r\in R$;

\item 
%in particular, 
for each $r\in R_i$, $m_i(r)$ denotes the trajectory $\tau \in \tau_i$ that $r$ must trace; in general, $\tau$ is a curve in the plane starting from the position of $r$ and ending at a target position defined by the strategy.
\end{itemize}
%
%The computational schema reported in 
\smallskip\noindent
In practice, Algorithm~\ref{alg:compute} can be used in a distributed algorithm as follows: 
\begin{quote}
%\emph{if a robot $r$ executing algorithm $\A$ detects that predicate $P_i$ holds, then $r$ traces $m_i(r)$.}
\emph{if a robot $r$ executing algorithm $\A$ detects that predicate $P_i$ holds, then $r$ first computes the next move $m_i$ obtained by executing $\Proc_i$, and then traces $m_i(r)$.}
\end{quote}

\smallskip\noindent
By abusing notation, $T_i$ can be used not only to identify a task but also to identify the set of all configurations that satisfy predicate $P_i$:
\begin{itemize}
\item
we call such a set \emph{class} $T_i$;
\item
class $T_F$ contains all the \emph{final configurations} in $\F(D)$;
\item
since predicates are mutually exclusive, classes $T_i$, $1\leq i \leq k+1$, form a partition of $\IA$. \\
\end{itemize}
We assume that the definition of $\Pi$ implies a characterization of the initial and final configurations. As already observed, it is possible that there exist unsolvable configurations for $\Pi$. Consider now any algorithm $\A$ for solving $\Pi$. By definition, $\A$ must transform any element of $\I\setminus \Unew(D)$ into an element in $\F(D)$. We recall that in such a transformation it is possible that $\A$ generates intermediate configurations in $\R\setminus \I$.  

As an example, consider again the \gath problem. In such a case $\I$ contains all the configurations with $n$ distinct points, $\F(D)$ contains any configuration with one point having multiplicity $n$, and the algorithm may generate intermediate configurations which are both non-initial and non-final, like for instance those configurations occurring as soon as two or more robots compose a multiplicity.

This example implies that, in general, the set of configurations associated to any task $T_i$ does not contain only elements of $\I$ but also some elements of $\R\setminus \I$ (cf. Figure~\ref{fig:partition2}).

\begin{figure}
\begin{center}
\scalebox{0.90}{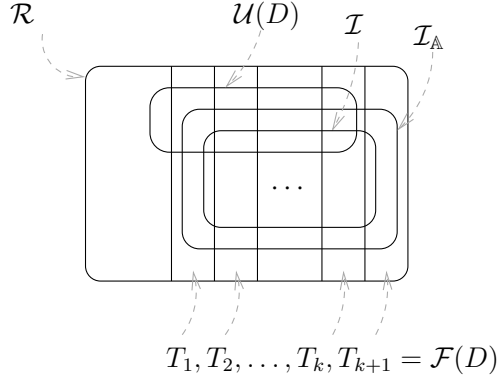 }
\end{center}
\caption{A visualization of the possible relationships among the classes $T_1,T_2,\ldots,T_{k+1}$ and the sets $\I$, $\Unew(D)$, $\F(D)$ and $\IA$, for a generic algorithm $\A$. For sake of correctness, algorithm $\A$ must guarantee $\IA\cap \Unew(D) = \emptyset$. 
}
\label{fig:partition2}
\end{figure}

%----------------------------------
\subsection{On the definition of predicates $P_i$}\label{ssec:methodology:predicates}

We have assumed that any algorithm $\A$ solving $\Pi$ is based on a strategy that decomposes the problem into tasks $T_1,T_2,\ldots,T_k,T_{k+1}$. In most cases, each task can be accomplished only when some \emph{pre-conditions} are fulfilled, and such conditions must be verified by $\A$ according to the current input configuration. Hence, in order to define the predicates, we need:
\begin{itemize}
\item
\emph{basic variables} that capture metric/topological/numerical/ordinal aspects of the input configuration which are relevant for the used strategy and that can be evaluated by each robot on the basis of its view;
\item
\emph{composed variables} that express the pre-conditions of each task $T_i$.
\end{itemize}
If we assume that $\pre_i$ is the composed variable that represents the pre-conditions of $P_i$, for each $1\le i\le k+1$, then predicate $P_i$ can be defined as follow:
\begin{equation}\label{eq:predicates}
	P_i = \pre_i \wedge \neg ( \pre_{i+1} \vee \pre_{i+2} \vee \ldots \vee \pre_{k+1} )
\end{equation} 
where $\pre_{k+1}$ is the pre-condition of $P_F$ and, in particular, $P_F=\pre_{k+1}$.

\begin{remark}\label{remark:precondition}
	This way to define the predicates implies a linearization/ordering of the tasks that must be accomplished. In fact, with respect to a given input configuration, the first predicate to be checked is $P_F=\pre_{k+1}$; if it is false, then $P_k=\pre_{k} \wedge \neg \pre_{k+1} = \pre_{k}$ is checked. If all predicates until $P_3$ returned false, then $P_2=\pre_2 \wedge \neg ( \pre_{3} \vee \pre_{4} \vee \ldots \vee \pre_{k+1} ) = \pre_2$ is checked. If even $P_2$ is false then we need, according to the designed task ordering, that $T_1$ must be performed on the input configuration. By choosing $\pre_1$ being the $\true$ tautology, then  we are sure each configuration in $\IA$ is processed by the algorithm, and that differently from what is shown in Figure~\ref{fig:partition2}, $T_1$, $\ldots$, $T_{k+1}$ would make a partition of the whole set of configurations $\R$. Moreover, it easily follows that $P_i\wedge P_j = \false$, for each $i\neq j$. In fact, if we assume $P_i=\true$ (which implies $\pre_i=\true$) and w.l.o.g. $j>i$, then by definition $\pre_j =\false$ as it appears in the negative form in $P_i$ in conjunction with $\pre_i$. 
	% Similarly, any $P_j$, $j<i$, is false as in it $\pre_i$ appears in its negative form. 
We conclude this remark by observing that different orderings may be defined and an ordering can be always decided according to the designed strategy.
\end{remark}

%---------------------------------
\subsection{On the concepts related to the execution of $\A$}\label{ssec:methodology:execution}

We start by introducing the concept of \emph{evolution} of an algorithm $\A$ expressed as an infinite sequence of configurations produced by $\A$ starting from an initial configuration $R\in \I$. By assuming \fsync, \ssync, or \sasync models, an evolution is a \emph{discrete} sequence of configurations, each one associated with a specific time instant generated by the common clock.

In the \async model, instead, an evolution becomes a \emph{continuous} sequence of configurations,
that is a curve $R(t)$ in $\Reali^{2n}$ as each configuration can be represented at each time $t$ by a vector of $n$ elements, each of them in $\Reali^2$. Basically each vector represents the $n$ positions of the $n$ robots in the plane.  
\begin{figure}[ht]
\begin{center}
\scalebox{0.6}{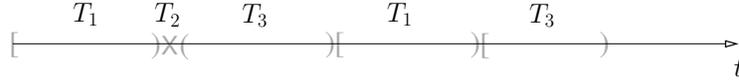 }
\end{center}
\caption{A fragment of an evolution for \async. %with five stages. 
The first interval corresponds to the sequence of configurations where predicate $P_1$ holds (i.e., configurations belonging to class $T_1$). The second interval is composed of just one point where $P_2$ holds, and so forth.}
\label{fig:continuos-execution}
\end{figure}
As shown in Figure~\ref{fig:continuos-execution}, each segment representing the time during which a task is performed can be closed, open or half-open depending on the  definition of the predicates. E.g., if a predicate exclusively depends on the presence of a robot $r$  on a specific point $p$ then the corresponding segment is in fact a single point (representing $r$ passing on $p$) or a closed interval (representing $r$ lying on $p$ for a while - the time represented by the segment). %Referring to Figure~\ref{fig:continuos-execution}, the second interval is composed of just one point.

\begin{definition}\cite{FYOKY15}\label{def:execution}
An execution of an algorithm $\A$ with respect to an initial configuration $R$ is an infinite discrete sequence $\Ex : R=R(t_0),R(t_1),\ldots$, where $\T= \{t_i : t_i < t_{i+1},~ i = 0,1,\ldots\}$ is the set of time instants at which at least one robot takes the snapshot $R(t_i)$ during its \look phase.%
%{\color{red}{, and $R(t_i)\neq R(t_{i+1})$ for each $i\in \T$}}. 
%\linecomment{ser}{mettendo questa condizione non dovrebbero servire le modifiche in rosso aggiunte da Alf nella def di transizione e nel remark successivo.}
\end{definition}

 If $R$ is composed of \fsync robots, then $\T$ contains all the time instants generated by the common clock and by definition at each $R(t_i)$, each robot is performing the \look phase. Basically, the execution of algorithm $\A$ coincides with its evolution.
 If $R$ is composed of \ssync robots, then $\T$ contains all the time instants generated by the common clock and by definition at each $R(t_i)$, each robot is either idle or performing the \look phase. 
 If $R$ is composed of \sasync robots, then $\T$ contains all the time instants generated by the common clock and by definition at each $R(t_i)$, each robot is either idle or starting any of its LCM phases. 
Finally, if $R$ is composed of \async robots, each $R(t)$ corresponds to a time instant picked along the continuous time line that represents the evolution of algorithm $\A$, in which there is at least one active robot taking the snapshot during its \look phase whereas any other robot might be idle or in any phase and time of its LCM cycle. In particular, robots can be seen while moving. 

A possible execution arising from the evolution of Figure~\ref{fig:continuos-execution} is shown in Figure~\ref{fig:discrete-execution}. By definition of task, each configuration of the execution satisfies one of the specified predicates. %Each maximal sub-sequence of the execution where the involved configurations satisfy the same predicate is called a \emph{stage}.

\begin{figure}[ht]
\begin{center}
\scalebox{0.6}{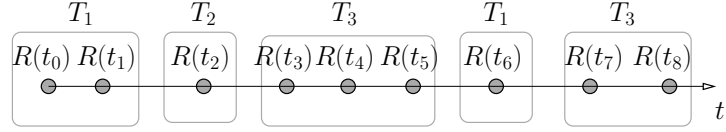 }
\end{center}
\caption{ A fragment of an execution with nine configurations and their membership to classes.  %and five stages. The first stage, for instance, consists of two configurations $R(0)$ and $R(1)$ both satisfying predicate $P_1$ and hence belonging to class $T_1$.
}
\label{fig:discrete-execution}
\end{figure}

%Within an execution, we define the concept of \emph{transition}: 1) two consecutive stages define a transition between the two corresponding classes; 2) more than one configuration within a stage define a transition within the same class. The set of all transitions among all the possible executions of $\A$ and for all the possible input configurations determine a directed graph $G$ called \emph{transition graph}. 

\begin{definition}\label{def:transition}
Let $\Ex: R=R(t_0),R(t_1),\ldots$ be an execution of an algorithm $\A$. Two consecutive configurations $R(t_i)$ and $R(t_{i+1})$ in $\Ex$, with $R(t_i) \neq R(t_{i+1})$, give rise to a \emph{transition} $T_j\rightarrow T_k$ from the class $T_j$ of $R(t_i)$ to the class $T_k$ of $R(t_{i+1})$.   
%
%Let $\Ex: R=R(0),R(1),\ldots$ be an execution of an algorithm $\A$. A \emph{transition} generated by $\Ex$ is defined as 
%follows: 1) two consecutive stages in $\Ex$ define a transition between the two corresponding classes; 2) more than one configuration within a stage in $\Ex$ define a transition within the same class. 
Generalizing, we say that \emph{$\A$ generates a transition} between two classes $T_j$, $T_k$ (or within the same class if $j=k$) if there exists an input configuration $R$ and an execution $\Ex$ of $\A$ that generates such a transition. 
\end{definition}

The set of all transitions of $\A$ determines a directed graph $G$ called \emph{transition graph}. \\

\begin{definition}\label{def:transition-graph}
A transition graph $G=(V,E)$ for an algorithm $\A$ is defined as follows:
\begin{itemize}
\item $V=\{T_1,T_2,\cdots, T_{k+1}\}$;
\item there exists a directed edge $(T_i,T_j)$ in $E$ 
      if and only if $\A$ generates a transition from class $T_i$ 
      to class $T_j$, possibly $i=j$.
%\item vertex $T_i$ admits a self-loop if there exists a stage made of 
%      more than one configuration that satisfies predicate $T_i$.
%      \linecomment{ser}{questo item mi sembra ridondante}
\end{itemize}
\end{definition}

\smallskip\noindent
 
\begin{remark}\label{remark:nosink}
According to Definitions~\ref{def:transition} and~\ref{def:transition-graph}, and since the move $m_F$ is $\nil$, then $T_F$ is a sink node in the transition graph $G$. 
Let $T_i$ be any node in $G$ different from $T_F$: if $T_i$ does not admit a self-loop, then each pair of consecutive configurations appearing in any execution and both belonging to $T_i$ must represent the same configuration. Moreover, since $T_i$ occurs in $G$ then it must admit at least a transition toward another node as the move designed for $T_i$ must lead to a new task. It follows that in $G$ there cannot exist \emph{sink} nodes except $T_F$.
\end{remark}
 
Figure~\ref{fig:partition} shows the transition graph $G$ corresponding to the fragment of execution of Figure~\ref{fig:discrete-execution}.
 
\begin{figure}[ht]
\begin{center}
\scalebox{0.6}{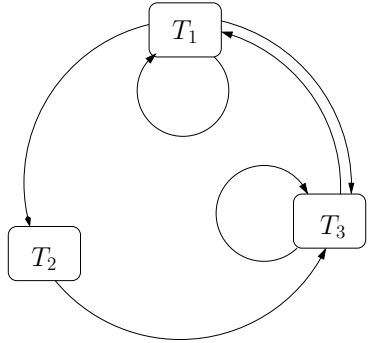 }
\end{center}
\caption{The transition graph corresponding to the fragment of execution of \ref{fig:discrete-execution}.}
\label{fig:partition}
\end{figure}

%-----------------------------------------------------
\subsection{On the correctness of $\A$}\label{sec:methodology:correctness}
The next theorem summarizes the properties that must hold to guarantee the correctness of an algorithm according to the proposed methodology.

\begin{theorem}\label{claim:correctness}
According to the proposed methodology, the correctness of any algorithm $\A$ can be obtained by proving that all the following properties hold:
\begin{itemize}
\item[$\h_1$:]
for each configuration in $\IA$ at least one predicate $P_i$ is true and,
for each $i\neq j$, $T_i \cap T_j= \emptyset$;
\item[$\h_2$:]
configurations in $\Unew(D)$ are not generated by $\A$, i.e. $\IA\cap \Unew(D)=\emptyset$;
\item[$\h_3$:]\label{it:transitions}
for each class $T_i$, the classes reachable from $T_i$ by means of a transition are exactly those represented in the transition graph $G$ (i.e., the transition graph is correct);
\item[$\h_4$:]
    possible cycles in the transition graph $G$ (including self-loops) must be performed a finite number of times.
%\item[$\h_5$:]
%    there is a direct path from each vertex $T_i$ of $G$ to $T_F$.
\end{itemize}
\end{theorem}

\begin{proof}
What we are going to show is that if all properties $\h_1,\ldots,\h_4$ hold, then there exists a time $t$ such that $R(t)$ is in $T_F$ and $R(t')=R(t)$ for any time $t'\ge t$, that is $\A$ is correct. 

Assume that a non-final configuration $R\in \I\setminus \U(D)$ is provided as input to $\A$. According to $\h_1$ there exists a single task (say $T_i$) to be assigned to robots with respect to $R$. According to $\h_2$, any configuration generated from $T_i$ (say $R'$) is solvable. Moreover, by $\h_3$ and $\h_4$, we can consider $R'$ belonging to some class (say $T_j$) different from $T_i$. The transition from $T_i$ to $T_j$ is represented in the transition graph $G$ by an edge from $T_i$ to $T_j$. According to this analysis, we can say that $R'$ will evolve during the time by changing its membership from class to class according to $G$. Although by Definition~\ref{def:execution} the execution is infinite, it will certainly reach $T_F$ (i.e., ends up in a task from where no new configurations are generated) since property $\h_4$ assures that cycles (including self-loops of tasks different from $T_F$) are performed a finite number of times. This, along with Remark~\ref{remark:nosink} that excludes the occurrence of sink nodes in $G$ different from $T_F$, implies that the execution of $\A$ eventually produces a final configuration.
Moreover, as the only movements allowed in $T_F$ by the methodology are the $\nil$ ones, then the reached configuration will not change anymore.
\end{proof}

The following remark restricts the number of properties to be proven in order to guarantee the correctness of any algorithm $\A$.

\begin{remark}\label{rem:pre-i}
If each predicate is defined as shown in Equation~\ref{eq:predicates} then, according to Remark~\ref{remark:precondition}, property $\h_1$ holds. 
\end{remark}

% ----------------------------------------------------
\subsubsection{On detecting transitions among classes}\label{ssec:transitions}
Notice that the most difficult property to prove among $\h_1,\ldots,\h_4$ is certainly $\h_3$, that is, it is difficult to correctly detect the transitions among classes generated by $\A$. 
%
% ---------------- new 
%Reminding that an execution of an algorithm $\A$ with respect to an initial configuration $R$ is an infinite sequence of configurations $\Ex : R=R(0),R(1),\ldots$ %, where $\T= \{t_i : i = 0,1,\ldots\}$ is the set of time instants at which any robot takes the snapshot $R(t_i)$ during the \look phase. 
%If $R$ is composed of \fsync or \ssync robots, then $\T$ contains all the time instances generated by the common clock and at each $R(t)$, $t\ge 0$, each robot in the configurations is either idle or is performing the \look phase. Conversely, in configurations composed of \async robots, at $R(t)$ for some $t>0$ it is possible that there are robots that have already decided to move at time $t'<t$ according to some move $m$ but they have not yet completed or started their movement. 
In particular, as already observed, in a generic configuration $R(t)$, $t>0$, of an execution of algorithm $\A$ designed for \async robots, it is possible there are robots that have already decided to move at time $t'<t$ according to some move $m$ but they have not yet completed or started their movement.
Move $m$ is then a pending move and its presence may heavily affect the correctness analysis of algorithm $\A$. In particular, if the creation of configuration $R(t)$ determines a transition from a task $T_j$ to another task $T_i$, then it results to be hard to analyze the behavior of $\A$ when $T_i$ is performed and hence difficult to correctly detect the transitions from $T_i$.
%
% ---------------- old (07.09.19)
%For instance, consider a transition $T_j \to T_i$; in the \async case it is possible that there are robots that already decided to move according to $T_j$ but they have not yet completed their movement in $T_i$. Move $m_j$ is then a pending moves and it must be carefully handled. This implies that in general it is hard to analyze the behavior of robots during the stage associated to $T_i$ and, in turn, it is difficult to detect all the possible transitions from $T_i$. 

The following definitions of \emph{stationary/almost-stationary/robust} configurations and transitions allow us to face such difficulties.

\begin{definition}[Stationary robot]\label{def:stationary}
A robot is said to be \emph{stationary} in a configuration $R(t)$ if at time $t$ it is:
\begin{itemize}
\item inactive, or
\item active, and during its current LCM cycle:
    \begin{itemize}
    \item it has not taken the snapshot yet;
    \item it has taken snapshot $R(t')=R(t)$, $t' \leq t$;
    \item it has taken snapshot $R(t')$, $t' \leq t$, which leads to a nil movement.
\end{itemize}
\end{itemize}
\end{definition}

\noindent
It is worth to remark that Definition~\ref{def:stationary} is a refinement of the one provided in~\cite{FYOKY15}. 

\begin{definition}[Stationary configuration]
A configuration $R$ is said to be \emph{stationary} if all robots are stationary in $R$.
\end{definition}
A simplification of the definition of stationary robot can be obtained if assuming that the snapshot is always taken at the beginning of the \look phase. According to~\cite{FPSW08}, this is always possible. The rationale behind it is that the \look phase can be potentially thought as composed of three sub-phases: (i) activation of the sensors; (ii) instantaneous snapshot acquisition; (iii) processing data. Hence, by considering sub-phase (i) as part of the preceding inactivity phase, the assumption stands.
If assumed, then the case of an active robot that has not yet taken the snapshot can be removed from the definition of stationary robot.

Note that, according to Definition~\ref{def:stationary}, a robot $r$ is \emph{non-stationary} in a configuration $R(t)$, if at time $t$ robot $r$ is active, has taken a snapshot $R(t')\neq R(t)$, $t'<t$, and is planning to move or is moving with a non-nil trajectory (i.e., $r$ gives rise to a pending move).

%\begin{definition}[Almost-stationary configuration]
%Let $R\in T_i$ be a configuration obtained from $R'\in T_j$ by applying move $m_j$. 
%$R$ is said to be \emph{almost-stationary} if the following condition holds:
%\begin{itemize}
%\item
%for each robot $r\in R$, if $r$ has computed a trajectory $\tau_j$ in $R'$ but has not yet terminated the corresponding LCM cycle, then the remaining part of $\tau_j$ is included into $\tau_i$, where $\tau_i$ is the trajectory that $r$ would compute from $R$.
%\end{itemize}
%\end{definition}

\begin{definition}[Almost-stationary configuration]
A configuration $R$ is said to be \emph{almost-stationary} if each robot in $R$ is either stationary or non-stationary but in such a case the remaining part of the trajectory it has not yet traced is included into $\tau$, where $\tau$ is the trajectory that $r$ would compute from $R$.
\end{definition}
%
%%%%%%%%%%%%%%%%%%%%%%%% DEF ALTERNATIVA
%\comment{
%{\color{blue}{ 
%   -- Definizione alternativa, ALF: -- \medskip 
%}}
%
%\begin{definition}[Almost-stationary configuration]
%Let $R\in T_i$ be a configuration obtained from one in $T_j$ by applying move $m_j$. 
%$R(t)$ is said to be \emph{almost-stationary} if the following condition holds:
%\begin{itemize}
%\item
%for each robot $r\in R(t)$, if $r$ started (but not yet terminated) its LCM cycle in some $R(t')$, $t'<t$, then $R(t'')\in T_j$ for any $t'\leq t'' <t$ and the remaining part of $\tau_j$ computed by $r$ from $R(t')$ is included into $\tau_i$, where $\tau_i$ is the trajectory that $r$ would compute from $R(t)$.
%\end{itemize}
%\end{definition}
%}
%%%%%%%%%%%%%%%%%%%%%%%% FINE DEF ALTERNATIVA

\begin{definition}[Robust configuration]
A configuration $R$ belonging to a task $T_i$ is said to be \emph{robust} if each robot $r$ in $R$ is either stationary or non-stationary but in such a case as long as $r$ has not terminated its current LCM cycle the configuration still belongs to $T_i$.
\end{definition}

From the above definitions it follows that each stationary configuration is also almost-stationary, and each almost-stationary configuration is also robust.

\begin{definition}[Types of transitions]
Let $T_j \to T_i$ be a transition. 
Then such a transition is stationary (almost-stationary, robust, resp.) 
if each $R\in T_i$ produced from any $R'\in T_j$ by applying move $m_j$ is 
stationary (almost-stationary, robust, resp.).
\end{definition}
Notice that the types of transition form a hierarchy: each stationary transition is also almost-stationary, and each almost-stationary transition is also robust.

\smallskip
Now, consider again the problem remarked above, namely the detection of transitions among classes generated by the algorithm. If we are able to show that all the  transitions leading to a class $T_i$ are stationary, then no pending moves must be considered during the analysis of the algorithm with respect to phase $T_i$, and this greatly simplifies the correctness proof. 
Similarly, if we prove that all the transitions leading to $T_i$ are almost-stationary there could be pending moves, but they can be analyzed as scheduled by the current move $m_i$. The robust case is more difficult to be detected and managed, but again, if proved, it simplifies the analysis of the algorithm since it ensures to resolve all pending moves due any task $T_j$ preceding the current task $T_i$ within $T_i$. Hence no pending moves generated in $T_j$ can propagate to any task succeeding $T_i$. 

%allows to {\color{red} easily validate} the trajectories of pending moves. 

\begin{remark}\label{rem:simplified-transitions}
Each time the creation of configuration $R(t)$, $t>0$, determines a transition from a task $T_j$ to task $T_i$ (possibly $i=j$) and such a transition is stationary, almost-stationary or robust, then the analysis of the behavior of the algorithm $\A$ during the execution of task $T_i$ is greatly simplified since possible movements due to past moves do not affect $\A$. In other words, when a transition is stationary/almost-stationary/robust, the complexity of the correctness analysis is somehow comparable to that occurring in case of \fsync/\ssync robots.
\end{remark}

According to this remark, our methodology suggests to reduce the complexity of proving property $\h_3$ by adding the following optional property:
\smallskip
\begin{itemize}
\item[$\h_{3'}$:]
\textit{each transition not leading to $T_F$ is stationary, almost-stationary, or robust, while each transition leading to $T_F$ is stationary. }
\end{itemize}

\medskip
It is worth to note that when designing an algorithm it is not so obvious that property $\h_{3'}$ can be ensured for all transitions. For the sake of completeness, we call any other possible type of transition as \emph{unclassified} transition.

Another phenomenon that could make it difficult to prove the correctness of any algorithm $\A$ is the presence of possible \emph{collisions} between robots. By collision we mean any kind of undesired multiplicity, such as those created by chance and not on purpose. To this respect, it is undesirable that the trajectories of two moving robots intersect. When $\A$ is not collision-free, not only $\A$ could fail to correctly terminate but also it may generate  more transitions and more configurations than those actually needed. By maintaining the  algorithm collision-free would then also confine the size of $\IA$. Accordingly, our methodology suggests to add the following optional property: 
\smallskip
\begin{itemize}
\item[$\h_{3''}$:]
\textit{the algorithm is collision-free.}
\end{itemize}

% ----------------------------------------------------
\subsubsection{Dealing with cycles of the transition graph $G$}\label{ssec:cycles}
Concerning property $\h_4$, here we provide a possible strategy by which it can be approached. 

A natural way to guarantee $\h_4$ is to consider each edge of $G$ 
and ensuring it can be traversed a finite number of times. This coincides with proving that each node of $G$ is entered a finite number of times.
This approach might be tedious and requiring a lot of effort. However it can be simplified as follows. Clearly, if a node does not admit the self-loop then it can be entered more than once only if it is part of a cycle in $G$. 
Then, consider all \emph{simple} cycles of $G$.\footnote{A simple cycle is any cycle where each node appears exactly once.} Referring to Figure~\ref{fig:partition}, simple cycles are the self-loops plus ($T_1,T_3$) and ($T_1,T_2,T_3$). One may detect a suitable subset of edges representing a hitting set of the edges involved in the simple cycles. Clearly self-loops must be all included. If one ensures such edges are traversed a finite number of times then it is guaranteed all cycles (the simple ones and compositions of them) are traversed a finite number of times as well. In order to prove that an edge is traversed (or a node is entered, resp.) a finite number of times, some property should be detected which provides a monotonic evolution of the execution with respect to such a property leading to the negation of the property itself within a finite number of edge traversals (node accesses, resp.). Once all self-loops are resolved this way, one may focus on nodes. Considering a suitable subset of nodes representing a hitting set of the nodes of the remaining simple cycles, it must be ensured that such nodes are entered a finite number of times. Clearly, the smaller is the size of the hitting set, the less is the number of nodes that must be considered. However, the minimality of the hitting set is not a requirement. One is free to choose among the nodes in favor of simple arguments required for the proofs.

% ==================================================================
% Notation
% ==================================================================
\section{Case study, detailed notation and definitions}\label{sec:notation}
As discussed in Introduction, we aim to show the potentials of the proposed methodology by an extended case study. In particular, we consider the \pf variant approached in~\cite{FYOKY15}. 
In addition to the definition of \pf provided in Section~\ref{sec:methodology}, in such a variant robots are endowed with global strong multiplicity detection and with \emph{chirality}, that is they share a common handedness. This of course changes their perception during the \look phase, as now the view can also exploit the chirality. For instance, by looking at the leftmost configuration in Figure~\ref{fig:notation}, it is evident the only disposal of the robots induces a vertical axis of reflection passing through the five aligned robots. However, when chirality is assumed, the specular robots at the two sides of the axis can be associated with different views, as chirality discriminate among left and right. In particular, robots share a common clockwise direction. As a consequence, from now one we restrict the set $\Aut(R)$ of all automorphisms for any configuration $R$ to contain only the identity and possible rotations, as reflections are resolved by chirality.

Generalizing~\cite{FYOKY15},
we relax the requirement that the LCS specific of a single robot remains the same among different LCM cycles.

We now provide all the notation, definitions and properties that will be exploited later for designing our new resolution algorithm for \pf with chirality. 

% ------ Definizione di PF tolta perché presente già nella sezione sulla metodologia ...
% -----------------------------------------
%
%A configuration $R$ is said \emph{initial} if it is stationary and all elements in $R$ are distinct, that is, no multiplicity occurs.%
%\footnote{Throughout this paper, we assume that any initial configuration contains no multiplicity. This is a typical assumption since it is impossible to break
%up multiple robots on a single position as we assume all robots execute the same algorithm}
%
%
%Let $P_1$ and $P_2$ be two multisets of points: if $P_2$ can be obtained from $P_1$ by uniform scaling, possibly with additional translation, rotation and reflection, then $P_2$ is \emph{similar} to $P_1$.  
%%
%Given a pattern $F$ expressed as a multiset $Z_0(F)$, we say that an algorithm $\A$ \emph{forms} $F$ from an initial configuration $I$ if for each possible execution $\Ex : R(0)(= I),R(1),R(2),\ldots$, there exists a time instant $i>0$ such that $R(i)$ is similar to $F$ and no robots move after $i$, i.e., $R(t) = R(i)$ hold for each integer $t\ge i$. 
%%
%The Pattern Formation (\pf) problem can be formulated as follows: 
%\begin{itemize}
%\item \emph{Given any initial configuration $R$ formed by $n\ge 3$ robots and any pattern $F$ (i.e., a multiset of $n$ elements) devise an algorithm $\A$, if any, able to form $F$ from $R$.}
%\end{itemize}

%------------------------------------------------
\subsection{Notation}\label{ssec:notation}
Given two distinct points $u$ and $v$ in the Euclidean plane, 
% let $d(u,v)$ denote their  distance, 
let $\Line(u,v)$ denote the straight line passing through these points, 
and let $(u,v)$ ($[u,v]$, resp.) denote the open (closed, resp.) segment containing 
all points in $\Line(u,v)$ that lie between $u$ and $v$. The half-line starting at 
point $u$ (but excluding the point $u$) and passing through $v$ is denoted by $
\halfline(u,v)$. We denote by $\angolo(u,c,v)$ the angle centered in $c$ obtained by rotating clockwise $\halfline(c,u)$ until overlapping $\halfline(c,v)$. The angle $\angolo(u,c,v)$ is measured 
from $u$ to $v$ in clockwise direction and the measure is 
always meant as positive.

Given an arbitrary multiset $P$ of points in $\Reali^2$, $\mult(p,P)$ denotes the number of occurrences of $p$ in $P$, while $C(P)$ and $c(P)$ denote the smallest enclosing circle of $P$ and its center, respectively. 
Let $C$ be any circle concentric to $C(P)$. 
We say that a point $p\in P$ is \emph{on} $C$ if and only if $p$ is on the circumference of $C$; $\partial C$ denotes all the points of $P$ that are on $C$. We say that a point $p\in P$ is \emph{inside} $C$ if and only if $p$ is in the area enclosed by $C$ but not in $\partial C$; $\Int(C)$ denotes all the points inside $C$. The radius of $C$ is denoted by $\delta(C)$.  
The smallest enclosing circle $C(P)$ is unique and can be computed in linear time~\cite{M83}. A useful characterization of $C(P)$ is expressed by the following property.

\begin{property}\cite{Welz91}\label{prop1} $C(P)$ passes either through two of the points of $P$ that are on the same diameter (antipodal points), or through at least three points. $C(P)$ does not change by eliminating or adding points to $\Int(P)$. $C(P)$ does not change by adding points to $\partial C(P)$. However, it may be possible that $C(P)$ changes by either eliminating or changing positions of points in $\partial C(P)$.
\end{property}

Given a multiset $P$, we say that a point $p\in P$ is \emph{critical} if  $C(P) \neq C(P\setminus \{p\}$).%
\footnote{Note that in this work we use operations on multisets.}
It easily follows that if $p\in P$ is a critical point, then $p\in \partial C(P)$.

\begin{property}\cite{CieliebakP02}\label{prop2} If $|\partial C(P)|\ge 4$ then there exists at least one point in $\partial C(P)$ which is not critical.
\end{property}

Given a multiset $P$, consider all the concentric circles that are centered in $c(P)$ and with at least one point of $P$ on them: $C_{\uparrow}^{i}(P)$ denotes the $i$-th of such circles, and they are ordered so that by definition $C_{\uparrow}^{1}(P)$ is the first one (which coincides with $c(P)$ when $c(P)\in P$), $C(P)$ is the last one, and the radius of $C_{\uparrow}^{i}(P)$ is greater than the radius of $C_{\uparrow}^{j}(P)$ if and only if $i>j$.
Additionally, $C_{\downarrow}^{i}(P)$ denotes one of the same concentric circles, but now they are ordered in the opposite direction: $C_{\downarrow}^{1}(P) = C(P)$ is the first one, $c(P)$ is the last one when $c(P)\in P$, and the radius of $C_{\downarrow}^{i}(P)$ is greater than the radius of $C_{\downarrow}^{j}(P)$ if and only if $i<j$.

Finally, we provide some additional notation and terminology referred to a given configuration $R$ and a given pattern $F$. The following definitions assume that $C(R)\equiv C(F)$ (cf. Figure~\ref{fig:notation}):
\begin{figure}[ht]
\begin{center}
\scalebox{0.90}{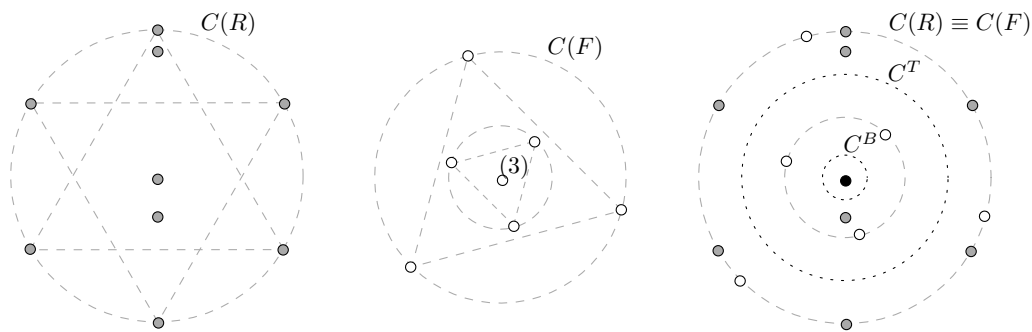 }
\end{center}
%\caption{An example of input for the \pf\ problem perceived by a generic robot according to its LCS, and related notation: on the left, an initial configuration $R$ composed of 9 robots with symmetricity $\rho(R)=1$; on the middle, the pattern $F$ with symmetricity $\rho(F)=3$  (numbers close to points refer to multiplicities); on the right, the embedding of $C(F)$ into $C(R)$ and the parking circles $\CT$ and $\CB$ (robots located in points of $F$ are represented as black points). Notice that in this example set $\M(C(R))$ contains two elements and just one robot is located inside $\Ann$.
%}
\caption{An example of input for the \pf\ problem perceived by a generic robot according to its LCS, and related notation: on the left, an initial configuration $R$ composed of 9 robots; on the middle, the pattern $F$, numbers close to points refer to multiplicities; on the right, the embedding of $C(F)$ on $C(R)$ and the parking circles $\CT$ and $\CB$ (robots located in points of $F$ are represented as black points). Notice that in this example just one robot is located inside $\Ann$.
}
\label{fig:notation}
\end{figure}
\begin{itemize}
\item
$\CT$ the \emph{parking circle at top level}, that is the median circle between $C(F)$ and $C_{\downarrow}^{2}(F)$ if $\Int(C(F))\neq \emptyset$, otherwise the median circle between $C(F)$ and $c(F)$;
\item
$\CB$ the \emph{parking circle at bottom level}; it corresponds to the median circle between $c(R)$ and $\min \{ \delta(C_{\uparrow}^{2}(R)), \delta(C_{\uparrow}^{1}(F)) \}$ when $c(F)\not\in F$, or the median circle between $c(R)$ and $\min \{ \delta(C_{\uparrow}^{2}(R)), \delta(C_{\uparrow}^{2}(F)) \}$ when $c(F)\in F$;
\item
$\Ann$ denotes the interior of the \emph{annulus} comprised by $C(R)$ and $\CT$ (hence, both the boundary circles $C(R)$ and $\CT$ are excluded from $\Ann$); 
\item
given a robot $r\in \partial C(R)$, $\ell_r$ denotes the line segment $[c(R),r]$; $\ell_r$ is called \emph{robot-ray};
\item
given a point $f\in \partial C(F)$, the line segment $\ell_f=[c(R),f]$  is called \emph{pattern-ray};
\item
$\Rob(\cdot)$ is a function that takes a region of the plane (e.g., annulus, sector, ray, ...) as input and returns all robots lying in the given region (e.g., $\Rob(\Ann)$ contains all robots in the annulus).
%\item
%%Let $C$ be any circle concentric to $C(R)$. $\M(C)$ denotes the set containing all the maximum cardinality subsets $M\subseteq \partial C$ such that $M$ forms a regular $|M|$-gon and $|M|$ divides $\rho(F)$. %Notice that all the elements in $\M$ are pairwise disjoint;
%Let $C$ be any circle concentric to $C(R)$. $\M(C)$ denotes the set containing all the maximum cardinality subsets $M\subseteq \partial C$ such that all the following conditions hold: 
%\begin{enumerate}
%\item $M$ forms a regular $|M|$-gon;
%\item $|M|$ divides $\rho(F)$;
%\item $|M|>1$.
%\end{enumerate} 
%\item 
%$\M'(C) = \bigcup_{M\in \M(C)} M$.
\end{itemize}

%--------------------------------------------------------------
\subsection{Symmetricity}\label{ssec:symmetricity}
The \pf with chirality problem was first introduced by Suzuki and Yamashita for the robots moving in the  Euclidean plane~\cite{SY99}. They characterized the class of formable patterns for \fsync robots endowed with chirality by using the following notion of \emph{symmetricity}. 
 
Consider a partition of $P$ into $k$ regular $m$-gons with common center $c(P)$, where $k = n/m$. Such a partition is called \emph{regular}. The symmetricity $\rho(P)$ of $P$ is the maximum $m$ such that there is a regular partition of $P$ into $k$ regular $m$-gons. Notice that $m$ points at $c(P)$ forms a regular $m$-gon,%
\footnote{A multiplicity of $m$ points, all at $c(P)$, is considered as a regular $m$-gon with radius zero.} %
any pair $\{p,q\}$ of points is a regular 2-gon with center the median point of the line segment $[p,q]$, and any point is a regular 1-gon with an arbitrary center.
Since any $P$ can be always partitioned into $n$ regular 1-gons, the symmetricity $\rho(P)$ is well defined. Examples of $\rho$ are depicted in Figure~\ref{fig:rho}.$(a)$-$(d)$. 
%
%We say that multiset $P$ of $n$ points in $\Reali^2$ (and hence a configuration $R$  or a pattern $F$) is symmetric if $\rho(P) = 1$, asymmetric otherwise. 
%
To this respect, notice the case in Figure~\ref{fig:rho}.$(c)$, where $\rho(P) = 1$ while $P$ appears to be symmetric. This particular case means that whenever $c(P)\in P$, the robot on $c(P)$ can transform $P$ into an asymmetric configuration $P'$ with $\rho(P') = 1$ by leaving $c(P)$.

\begin{figure}[t]
	\begin{center}
	 \resizebox{0.88\textwidth}{!}{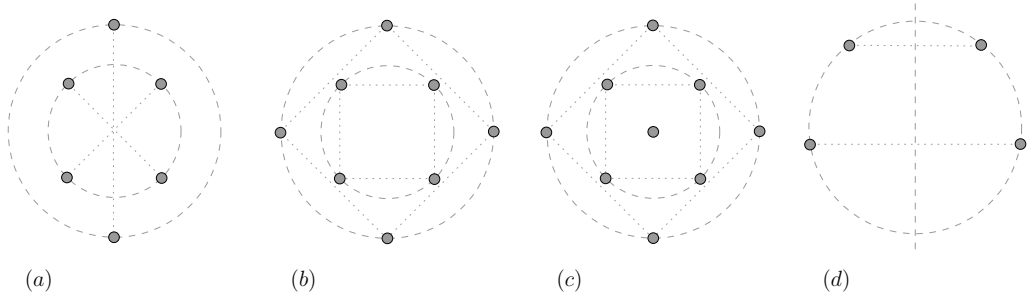}
	\caption{\small Examples of symmetricity of a set of points $P$. 
			In $(a)$, $\rho(P)=2$; 
			in $(b)$, $\rho(P)=4$; 
			in $(c)$, $\rho(P)=1$; 
			in $(d)$, $\rho(P)=1$.}
	\label{fig:rho}
	\end{center}
\end{figure}

In Section~\ref{sec:methodology} we have formalized the Pattern Formation (\pf) problem. We can now recall the characterization about formable patterns according to the notion of symmetricity.

\begin{theorem} \emph{\cite{SY99}}\label{th:rho}
Let $R$ be an initial configuration and $F$ be a pattern. 
$F$ is formable from $R$ by \fsync robots with chirality if and only if $\rho(R)$ divides $\rho(F)$.
\end{theorem}

This result states that the pattern formation problem, even for \fsync robots, highly depends on the symmetricity of both $R$ and $F$; moreover, when robots have chirality the symmetricity is entirely represented by the parameter $\rho$. On the contrary, Figure~\ref{fig:rho}.$(d)$ shows that $\rho$ is not useful when robots have no chirality since it does not take into consideration reflection symmetries. An interesting characterization about the symmetricity of points in the 3-dimensional space can be found in~\cite{YUKY17}.

Notice that the above theorem implies that, for the \pf problem, the set of unsolvable configuration $\U(D)$, with $D=F$, contains at least all configurations $R$ such that $\rho(R)$ does not divide $\rho(F)$. Formally, $\U(F) \supseteq \{R: \rho(R)$ \mbox{ does not divide } $\rho(F)\}$. Actually, as we will prove in Section~\ref{sec:correctness} by means of Theorem~\ref{th:correctness}, $\U(F)\cap \I = \{R: \rho(R)$ \mbox{ does not divide } $\rho(F)\}$. %, hence from now on we assume such equality. 
Concerning unsolvable configurations that are not initial, $\U(F)$ certainly contains those with a multiplicity composed by a number of robots greater than the number of robots composing the biggest multiplicity of $F$, as the adversary can always prevent to break multiplicities (i.e. to break such kind of symmetries). 
%\linecomment{ser}{possono esserci altre unsolvable non iniziali, come ad esempio qundo in $F$ ci sono $n_1$ biggest multiplicity e nella conf corrente ci sono $n_2>n_1$ biggest multiplicities - potremmo allora generalizzare dicendo che se $\mu_1,\mu_2,\ldots, \mu_{t'}$ sono le molteplicità in decreasing order del pattern e $\ell_1,\ell_2,\ldots, \ell_{t''}$ sono le molteplicità in decreasing order della conf corrente, con $t'' \ge t'$, allora deve valere che $\mu_i \ge \ell_i$ for each $1\le i\le t'$}

%%%%%%%%%%%%%%%%%%%%%%%%%%%%%%%%%%%%%% NEW %%%%%%%%%%%%%%%%%%%%%%%%%%%%%%%%%%%%%%%%%

\smallskip
Related to the symmetricity, we need to introduce one further parameter that will be exploited by our resolution algorithm. 
Let $C$ be any circle concentric to $C(R)$. $\M(C)$ denotes the set containing all the maximum cardinality subsets $M\subseteq \partial C$ such that all the following conditions hold: 
\begin{enumerate}
\item robots in $M$ form a regular $|M|$-gon;
\item $|M|$ divides $\rho(F)$;
\item $|M|>1$.
\end{enumerate} 
Then, let $\M'(C)=\bigcup_{M\in \M(C)} M$, i.e., $\M'(C)$ is set of robots belonging to elements of $\M(C)$.
By referring to Figure~\ref{fig:notation}, the initial configuration $R$ (on the left) has symmetricity $\rho(R)=1$ and the set $\M(C(R))$ contains two elements of three robots each, since the pattern $F$ (on the middle) has symmetricity $\rho(F)=3$.

\smallskip

The next lemma makes a relationship between $\rho(R)$ and the size of any element of $\M(C)$, being $C$ any circle centered in $c(R)$ and with robots in $\partial C$.

\begin{lemma}\label{lem:divM}
 Let $F$ be a pattern, $R$ be a configuration such that $\rho(R)$ divides $\rho(F)$, and $M \in \M(C)$ with  %a given circle $C$ in $\Ann$. 
$C=C_{\downarrow}^{i}(R)$, $i\ge 1$. Then $\rho(R)$ divides $|M|$.
\end{lemma}
\begin{proof}
 Let $r$ be a robot in $M$ and let $\varphi \in \Aut(R)$ such that $\varphi^i(r)$ are distinct robots belonging to $C$, for each $i=0$, $1$, $\ldots$, $\rho(R)-1$. If $\varphi(r)=r'$ belongs to $M$, then all the robots $\varphi^i(r)$ belong to $M$ and this implies the claim.
 
 We show by contradiction that the above case is the only possible one. In fact, if $r' \not \in M$, by the equivalence of $r'$ with $r$, also $r'$ and any other robot in $\{\varphi(r) ~|~ r \in M\}$ must be part of a regular $|M|$-gon $M'$, different from $M$. It comes out, in general, that $\{\varphi^i(r) ~|~ r \in M,~ i=0,1,\ldots, \rho(R)-1\}$ form a regular $\mbox{lcm}(\rho(R),|M|)$-gon, where $\mbox{lcm}(a,b)$ denotes the least common multiple of $a$ and $b$. Since by hypothesis $\rho(R)$ divides $\rho(F)$ and, by definition of $\M(C)$, also $|M|$ divides $\rho(F)$, then $\mbox{lcm}(\rho(R),|M|)$ divides $\rho(F)$ as well. If $\rho(R)$ does not divide $|M|$, then $\mbox{lcm}(\rho(R),|M|)>|M|$ but this contradicts the maximality of $|M|$.
\end{proof}

% ------ 
%It is worth to remark that in~\cite{YUKY17}, authors used cyclic groups to model symmetricity of robots with chirality. In particular, the cyclic group $G_k$ consists of the rotation operations around $c(P)$ by $2\pi/k, 2 \cdot 2\pi/k,\ldots, k \cdot 2\pi/k$. Of course, its order is $k$. 
%As particular cases we have $G_1$, the trivial group containing only the identity operation (which occurs when there is no symmetry at all) and $G_2$, the  group containing only the identity operation and the rotation by $\pi$ around the center (for instance, all the isometries of the letter `Z'). 
%Hence, if a configuration does not contain any robot on $c(R)$, then its symmetricity is essentially the maximum order of the cyclic group that acts on it; otherwise, its symmetricity is 1.
%-------
%\linecomment{Alf}{Cambiare la lettera $C$ già usata per i cerchi.}

%------------------------------------------------
\subsection{View of robots}\label{ssec:view}
We now formalize the concept of \emph{view} of a point in the Euclidean plane according to our needs (cf. Section~\ref{ssec:robot-view}). 
%It can be used by robots to determine whether a configuration $R$ and/or a pattern $F$ is symmetric or not.
%
Let $P$ be a generic multiset of points not including $c=c(P)$. For $p\in P$, we denote by $V(p)$ the view of $P$ computed from $p$. This is a sequence of couples (angle, distance) defined as follows: first $(0,d(c,p))$ then, in order from the farthest to the closest point to $c$, all couples $(0,d(c,p'))$ for any $p'\neq p$ in $\halfline(c,p)$, and successively all couples $(\angolo(p,c,p'), d(c,p'))$ arising from all other rays processed in clockwise order and points $p'$ from the farthest to the closest ones to $c$, foreach ray.
If $p=c(P)$ then $p$ is said the point in $P$ of minimum view, otherwise any $p = argmin\{V(p') : p'\in P\}$ is said of minimum view in $P$.

These definitions naturally extend to any configuration $R$ of robots and to a pattern $F$ as well. In particular, as we are dealing with robots endowed with chirality, the clockwise direction used in the definition of the view is  well-defined. 

As already observed in Section~\ref{ssec:robot-view}, if each robot can be associated with a unique view, then the configuration is perceived as asymmetric. For instance, in Figure~\ref{fig:rho}, configurations (a), (b) and (c) are all perceived as symmetric, whereas (d) is not as the clockwise direction produces different views to the potentially specular robots. In practice, the effect of assuming chirality results in breaking all reflection axes by means of the view. It comes out that if a robot views a configurations as symmetric, the only type of symmetry it can perceive is the rotation. In an asymmetric configuration, instead, each robot is associated with a different view and in particular there is only one robot associated with the minimum view. However, when there is a single robot $r$ occupying $c(R)$ as in Figure~\ref{fig:rho}.(c), then $r$ is the only robot of minimum view by definition. This property can be exploited to break a possible rotation, if required. It follows that when $\rho(R)=1$ then $R$ is either perceived as asymmetric or there is a single robot in $c(R)$. 
%%%%%%%%%%%%%%%%%%%%%%%%%%%%%%%%%%%%%%%%%%%%%%%%%%%%%%%%%%%%%%%
%Basically, when $\rho(R)=1$ it is always possible to elect a leader (for instance the robot of minimum view), and $R$ is called a \emph{Leader configuration}. \linecomment{Gab}{Leader configuration non lo usiamo ma il concetto può essere utile quando descriviamo la procedura ``leader'' }
%
%\linecomment{ser}{mattere solita relazione tra vista minima e asimmetria; aggiungere nota sul fatto che nel caso di conf asimmetriche la vista minima permette di eleggere un leader ...}

% ===================================================
% Algorithm
% ===================================================
\section{The algorithm for \pf}\label{sec:algorithm}
In this section we present our algorithm for solving the \pf problem for \async robots endowed with chirality. This algorithm is designed according to the methodology provided in Section~\ref{sec:methodology}.  

Before presenting the algorithm, we recall that any input configuration $R$ does not contain multiplicities. Concerning the number of robots $n$, we assume $n\ge 3$, since for $n=1$ the \pf\ problem is trivial and for $n=2$, either \pf\ is trivial or unsolvable depending whether $F$ is composed of two or one point~\cite{CFPS12}, respectively. Concerning the pattern $F$ to form, it might contain multiplicities. Moreover, according to Theorem~\ref{th:rho}, we assume that $\rho(R)$ is a divisor of $\rho(F)$ (otherwise $R\in \U(F)$, that is $R$ is unsolvable).

In the remainder, we first provide a high-level description of our strategy for the decomposition of the \pf problem into tasks (cf. Section~\ref{ssec:algorithm:subdivision}), then we summarize all the defined tasks (cf. Section~\ref{ssec:summarized-tasks}), and finally we present all the details of our algorithm concerning tasks' predicates, moves, and transition graph (cf.  Section~\ref{sec:algorithm:tasks}). Notice that in  Section~\ref{sec:example} we provide an explanatory example about the behavior of the proposed algorithm, and there we provide some missing details about moves. 

% -------------------------------------------------------
\subsection{Subdivision into tasks}\label{ssec:algorithm:subdivision}

As suggested in Section~\ref{ssec:methodology:decomposition}, here we describe a hierarchical decomposition of \pf into sub-problems so that each sub-problem is simple enough to be formalized as a task realizable by (a subset of) robots.

The problem is initially divided into six sub-problems denoted as \emph{Symmetry Breaking} (\SB), \emph{Reference System} (\RS), \emph{Partial Pattern Formation} (\PPF), \emph{Finalization} (\Fin),  \emph{Special Cases} (\SC), and \emph{Termination} (\Term). Some of these sub-problems are further refined until the corresponding tasks can be easily formalized. These initial six sub-problems are described by assuming an initial configuration $R$ to be transformed into a pattern $F$. %According to Theorem~\ref{th:rho}, it is also assumed that $\rho(R)$ divides $\rho(F)$.
 
 \blocco{Symmetry Breaking \emph{(\SB)}} Consider the case in which the initial configuration admits a rotation due to an automorphism $\varphi$ whose order $p$ is  not a divisor of $\rho(F)$. In this situation, by~\cite{SY99}, $\rho(R)$ must be necessarily equal to one as otherwise the problem would be unsolvable. It follows that by the definition of symmetricity, there must be a robot occupying $c(R)$. It is mandatory for each solving algorithm 
 to break this symmetry. In fact, without breaking the symmetry, any pair of symmetric robots may perform the same kind of movements and this may prevent the formation of the desired pattern. 
 
In our strategy, a single task $T_1$ is used to address the problem \SB. This task requires to carefully move the robot away from the center until to obtain a stationary asymmetric configuration. 
The main difficulties for \SB are: (1) to avoid the formation of other symmetries that could prevent the pattern formation and (2) to correctly face the situation in which multiple steps are necessary to reach the target. In the latter case, the algorithm must detect whether there is a possible robot moving that has not yet reached a designed target.

Notice that we consider \SB as a task of the Reference System sub-problem that we are going to describe in the next paragraphs.  
 
\blocco{Reference System \emph{(\RS)} - (How to embed $F$ on $R$)} This sub-problem concerns one of the main difficulties arising when the pattern formation problem is addressed: the lack of a unique embedding of $F$ on $R$ that allows each robot to uniquely identify its target (the final destination point to form the pattern). In particular, $\RS$ can be described as the problem of moving or matching some (minimal number of) robots into specific positions such that they can be used by any other robot as a common reference system. Such a reference system should imply a unique mapping from robots to targets, and should be maintained along all the movements of robots. 

As preliminary embedding of $F$ on $R$, it is assumed $C(F)$ matches with $C(R)$.
Then, $\RS$ is solved by leaving on (or moving to) $C(R)$ a number 
$m\ge 2$ of robots so that $m$ divides $\rho(F)$.\footnote{
Our strategy requires to solve $\RS$ only when $\rho(F)>1$ and $\delta(F)>0$.  This will be explained at the end of Section~\ref{ssec:summarized-tasks}.} 
Successively, if required, the $m$ robots left on $C(R)$ are rotated so as to form a regular $m$-gon. In doing so, the full embedding of $F$ on $R$ can be easily determined by matching the $m$ robots on $C(R)$ with $m$ points on $C(F)$: if there are exactly $m$ points in $\partial C(F)$ the embedding is unique, if there are $k\cdot m$ points, with $k\ge 2$, the $m$ robots on $C(R)$ are matched with the $m$ points in  $\partial C(F)$ having minimum view. As long as no further robots are moved to $C(R)$ and the $m$ robots on $C(R)$ are not moved, the embedding of $F$ on $R$ remains well-defined. Finally, in order to guarantee stationarity before changing task, we require not only the formation of the regular $m$-gon but also that $\Ann$  - i.e., the annulus between $C(R)$ and $\CT$ - does not contain robots.

Since $\RS$ is a complex problem, it is further divided into six sub-problems. As already pointed out, the first sub-problem is \SB, then we need to specify $\RS_1$,$\RS_2\ldots$, $\RS_5$. They are detailed as follows:

\begin{itemize}
\item $\RS_1$ is responsible for opportunely moving toward $\CT$ all robots in $\Ann$, that is robots residing in the area between $C(R)$ and $\CT$ - this problem is associated to task $T_2$.

\item $\RS_2$ is responsible for removing robots from $C(R)$ when too many robots reside there. Since such a removal can be performed in two different ways, this problem is further subdivided: 
\begin{itemize}
\item 
$\RS_{2.1}$ considers configurations where $\M(C(R)) \neq \emptyset$, that is configurations having regular $m$-gons on $C(R)$ such that $m>1$ and $m$ divides $\rho(F)$. This task removes robots from $C(R)$ until exactly one maximal regular $m$-gon of $\M(C(R))$ remains - this problem is associated to task $T_3$; 
\item
$\RS_{2.2}$ considers configurations where $\M(C(R)) = \emptyset$, that is configurations without regular $m$-gons on $C(R)$ such that $m>1$ and $m$ divides $\rho(F)$. Since such configurations are asymmetric, this task removes one non-critical robot at a time from $C(R)$ until exactly $m$ robots remain, with $m$ being the minimal prime factor of $\rho(F)$ or $m=3$ (and subsequently two antipodal robots must be created by task $T_6$ in order to remove a non-critical robot from $C(R)$) - this problem is associated to task $T_4$.
\end{itemize} 
             
\item 
$\RS_3$ is responsible for moving robots to $C(R)$ when there are too few robots on $C(R)$ with respect to $\rho(F)$. In particular, this task is responsible for moving robots from the interior toward $C(R)$ so as to obtain on $C(R)$ a number $m$ of robots equal to the minimal prime factor of $\rho(F)$ - this problem is associated to task $T_5$.

\item $\RS_4$ is responsible for creating two antipodal robots on $C(R)$; it could be necessary as a next task of $T_4$ when three robots are on $C(R)$ but three is not a divisor of $\rho(F)$  - this problem is associated to task $T_6$.

\item $\RS_5$ is responsible for forming a uniform circle on $C(R)$ when the number $m$ of robots on it is equal to the minimal prime factor of $\rho(F)$ - this problem is associated to task $T_7$.
\end{itemize}

\blocco{Partial Pattern Formation \emph{(\PPF)}} The main difficulties in this task are to preserve the reference system and to avoid collisions during the movements.
The task concerns moving all robots inside $C(R)$ so as to form a preliminary pattern $F'$ defined from $F$ as follows. Pattern $F'$ differs from $F$ only for those possible points on $C(F)$ different from the $m$ ones already matched by the resolution of problem $\RS$ - notice that $\PPF$ is addressed only once $\RS$ is solved. Such points, if any, are instead radially projected to $\CT$ in $F'$. In our strategy, task $T_8$ is designed to solve this problem. 
For addressing this task we consider the area delimited by $C(R)$ as divided into $m$ sectors. Within each sector we can guarantee that at most one robot per time is chosen to be moved toward its target: it is the one not on a target, closest to an unoccupied target, and of minimum view in case of tie. We are ensured that always one single robot $r$ per sector will be selected since the maximum symmetricity that the configuration can assume is $m$ (we recall that, due to the solution provided for the $\RS$ problem, the robots on $\partial C(R)$ form a regular $m$-gon). 
For each sector, the selected robot is then moved toward one of the closest targets until it reaches such a point if it resides inside the same sector, or it reaches the successive (clockwise) sector. 
All moves must be performed so as to avoid the occurrence of collisions; hence, it follows that sometimes the movements are not straightforward toward the target point. To this end we exploit a kind of Manhattan distance (called here \emph{Sectorial distance}) where moving between two points in the area delimited by $C(R)$ is constrained by rotating along concentric circles centered at $c(R)$ and moving along rays starting from $c(R)$. 

In order to solve $\PPF$, we make use of a procedure called $\Distmin()$ designed ad-hoc for computing the required trajectories according to the Sectorial distance. Once $F'$ is formed, either $F'$ coincides with $F$ or it only remains to radially move robots from $\CT$ to $C(R)$. To this aim problem $\Fin$ is addressed.

\blocco{Finalization \emph{(\Fin)}} It refers to the so-called finalization task. It occurs when the only robots not well positioned according to $F$ are those on $\CT$. By guaranteeing radial movements of such robots toward $C(R)$, the formation of pattern $F$ is completed. In our strategy, task $T_9$ is designed to solve this problem. It is worth to mention that while moving robots from $\CT$ to $C(R)$, the common reference system might be loss. However, we are able to guarantee that robots can always detect they are solving $\Fin$.  

\blocco{Special Cases \emph{(\SC)}} 
This concerns the resolution of some easily identifiable sub-cases that have been already solved in the literature and hence can be treated apart by known algorithms.
For the sake of convenience, in our strategy the resolution of the special case in which $F$ is composed of one point with multiplicity $|R|$ (a.k.a. \gath) is delegated to~\cite{CFPS12}. Similarly, when $\rho(F)=1$ then~\cite{CDN19} is applied as a subroutine. In both cases, the identification of the sub-problem is determined simply by looking at $F$, that is it does not depend on the robot movements. For such cases, our strategy considers a specific task $T_{10}$.
 
\blocco{Termination \emph{(\Term)}} It refers to the requirement of letting robots recognize the pattern has been formed, hence no more movements are required. In our strategy, a task $T_{11}$ is designed to address this problem. Clearly, only $\nil$ movements are allowed, hence if the task is started from a stationary configuration, then it won't be possible to switch to any other task.

%\smallskip
%Actually, in order to solve \pf\ with chirality, we remark that an additional problem must be faced. This crosses different tasks and concerns how robots can be able to recognize which task they are solving currently without jumping to different tasks in an uncontrolled way.

% --------------------------------------------
\subsection{The designed tasks}\label{ssec:summarized-tasks}
By summarizing the above analysis and according to the proposed methodology,  we can say that our strategy partitions the \pf problem into the following eleven tasks $T_1$, $T_2$, $\ldots$, $T_{11}$:
\begin{itemize}
\item[-] \RS: Create a common reference system. 
              General sub-problem further divided into \SB, $\RS_1$, $\RS_2$, $\ldots$, $\RS_5$:
   \begin{itemize}
   \item[-] \SB\ - Ensure $c(R)$ empty: \textit{task $T_1$}.
   \item[-] $\RS_1$ - Make $\Ann$ empty to ensure stationarity: \textit{task $T_2$}.
   \item[-] $\RS_2$: Sub-problem concerning the removal of robots from $C(R)$ 
         until $|\partial C(R)|$ divides $\rho(F)$. It is further divided 
         into two tasks according to the cardinality of $\M(C(R))$:
       \begin{itemize}
       \item[-] $\RS_{2.1}$ - Case $\M(C(R) \neq \emptyset$:  
              remove robots from $C(R)$ until exactly one maximal regular 
              $m$-gon of $\M(C(R))$ remains: \textit{task  $T_3$};
       \item[-] $\RS_{2.2}$ -  Case $\M(C(R)) = \emptyset$: remove robots 
                from $C(R)$ until exactly $m$ robots remain, with $m$ 
                being either the minimal prime factor of $\rho(F)$, 
                or $m=3$: \textit{task $T_4$}.
       \end{itemize}
   \item[-] $\RS_3$ - Bring robots to $C(R)$ until $|\partial C(R)|$ 
         divides $\rho(F)$:  \textit{task $T_5$}.
   \item[-] $\RS_4$ - Create two antipodal robots on $C(R)$: \textit{task $T_6$}.
   \item[-]  $\RS_5$ - Create a regular $m$-gon on $C(R)$: \textit{task $T_7$}.
   \end{itemize}
   \item[-] \PPF\ - Make a partial pattern formation: \textit{task $T_8$}.
   \item[-] \Fin\ - Finalize the pattern formation: \textit{task $T_9$}.
   \item[-] \SC\ -  Solve \pf by means of other algorithms when $F$ is composed 
            of one point with multiplicity $|R|$ or
            $\rho(F)=1$:  \textit{task $T_{10}$}.
   \item[-] \Term\ - Identify that $F$ is formed and hence maintain 
         each robot without moving: \textit{task $T_{11}$}.
\end{itemize}

\smallskip\noindent
We remark that task $T_{10}$ uses known algorithms to address the cases in which (1) $\rho(F) = 1$ or (2) $F$ is composed of one point with multiplicity $|R|$ (that is, $\delta(C(F)) = 0$). As a consequence, in each task different from $T_{10}$ our strategy can assume the following conditions:  $\rho(F) > 1$ and $\delta(C(F)) > 0$. 

Summarizing, our strategy will be based on the next properties maintained valid in each task different from $T_{10}$: 
\begin{itemize}
\item
points in $\partial C(F)$ form regular $m$-gons with $m\ge 2$;
\item
$C(F)\equiv C(R)$;
\item
robots movements never change the radius and the center of $C(R)$.
\end{itemize}
% ------------------------------------------------------------------
\subsection{Tasks' predicates and moves}\label{sec:algorithm:tasks}

Here we provide all the details about tasks' predicates, moves, and the transition graph for our algorithm, as suggested by the methodology in Sections~\ref{ssec:methodology:predicates} and~\ref{ssec:methodology:execution}. In particular, Table~\ref{tab:basic-variables} shows
the basic variables that capture all the metric/topological/numerical/ordinal aspects that are relevant for our strategy. Notice that most of them capture the relationships between the number of robots on $C(R)$ and $\rho(F)$, as required by the tasks associated to the sub-problem $\RS$.

% ------------------------------ inizio tabella variabili base
\begin{table*}
\small
\caption{ The basic Boolean variables used to define all the tasks' preconditions. }
\label{tab:basic-variables}
\bgroup
\def\arraystretch{1.4}%  1 is the default, change whatever you need
\setlength{\tabcolsep}{6pt}
\begin{center}
  \begin{tabular}{ | c | p{0.75 \textwidth} | }
    \hline
    \textit{var}  &  \textit{definition} \\ \hline \hline
    $\xduno$ & $|\partial C(R)|$ is not a divisor of $\rho(F)$  
%\linecomment{ser}{deve essere rinominato in $\xduno$ per evitare lo stesso simbolo usati per le mosse}    
    \\ \hline
    $\xddue$ & $|\partial C(R)|$ is not the minimal prime factor of $\rho(F)$ 
%\linecomment{ser}{deve essere rinominato in $\xddue$ per evitare lo stesso simbolo usati per le mosse}                
               \\ \hline
    $\xf$ & $|\partial C(R)|$ is smaller than the minimal prime factor of $\rho(F)$         
            \\ \hline
    $\xt$ & $|\partial C(R)|=3$ and 2 is a divisor of $\rho(F)$  \\ \hline
    $\xu$ & Robots in $\partial C(R)$ form a regular $m$-gon \\ \hline
    $\xc$ & $\partial C_{\uparrow}^{1}(R) =\{r\}$ and 
            $d(r,c(R)) < \delta(\CB)$ \\ \hline 
    $\xa$ & $\Rob(\Ann)$ is empty \\ \hline
    $\xm$ & $\M(C(R))$ is empty  
%\linecomment{ser}{deve essere rinominato in $\xm$}     
    \\ \hline
    $\xp$ & $F$ can be obtained by projecting radially on $C(R)$ all 
            robots in $\Ann \cup \CT$ \\ \hline
    $\xg$ & $\rho(F)=1$ or $F$ contains only one element with multiplicity 
            $|R|$ \\ \hline
    $\xw$ & $R$ is similar to $F$ \\ \hline
  \end{tabular}
\end{center}
\egroup
\end{table*}
% ------------------------------ fine tabella variabili base

Table~\ref{tab:tasks-bis} summarizes all the ingredients determined by the proposed methodology: the first two (general) columns recall the hierarchical decomposition described in the previous section, the third column associates tasks names to sub-problems, and the fourth column defines precondition $\pre_i$ for each task $T_i$ (cf. Section~\ref{ssec:methodology:predicates}). These preconditions must be considered according to Equation~\ref{eq:predicates}, that is the predicate $P_i$ associated to task $T_i$ is defined as 
\begin{equation}\label{eq:2}
P_i = \pre_i \wedge \neg ( \pre_{i+1} \vee \pre_{i+2} \vee \ldots \vee \pre_{11}), \text{ for each } i=1,2,\ldots,11  
\end{equation} 	
As a consequence, such predicates are intended to be used in the \compute phase of each robot as presented in Algorithm~\ref{alg:compute}.

The fifth column of Table~\ref{tab:tasks-bis} contains the name of  the move used in each task (we simply denote as $m_i$ the move used in task $T_i$), and the specification of each move is provided in Table~\ref{tab:moves}. Notice that in Table~\ref{tab:moves} some moves are directly specified, while a few of them are defined by means of specific procedures (namely, $\GoToC$, $\Distmin$, $\Circle$, $\Gathering$, and $\Leader$ - formally defined in the next section). Moreover, all the trajectories defined in the moves are always straight lines, or arcs of circles centered in $c(R)$, or compositions of both in order to guarantee stationarity and to avoid collisions. More details that specify all target points and trajectories will be provided in Section~\ref{sec:correctness}.

The last column of Table~\ref{tab:tasks-bis} reports the possible transitions for each task. For instance, while performing task $T_1$ our algorithm may generate configurations belonging to the classes associated to tasks $T_1,T_2,\ldots,T_6$, and during task $T_9$ only configurations belonging to the classes $T_9$ and $T_{11}$ may be generated.
According to the proposed methodology, all such transitions are summarized in the transition graph (cf. Section~\ref{ssec:methodology:execution}) shown in Figure~\ref{fig:transitions2}.

\begin{table*}[ht]
\small
\caption{ Algorithm for \pf. }
\label{tab:tasks-bis}
\bgroup
\def\arraystretch{1.4}%  1 is the default, change whatever you need
\setlength{\tabcolsep}{6pt}
\begin{center}
  \begin{tabular}{ | c | l | l | l | c | r | l | l | }
    \hline
    \textit{problem} & \multicolumn{3}{l|}{\textit{sub-problem}} & \textit{task} &  \textit{precondition} &  \textit{move} &  \textit{transitions} \\ \hline \hline

  \multirow{11}{*}{\raggedleft \pf } & \multirow{7}{*}{\raggedleft \RS }  & \multicolumn{2}{l|}{ \SB }  & \textit{ $T_1$} & \textit{true}  &  $m_1$ & $T_1,T_2,T_3,T_4,T_5,T_6$ \\ \cline{3-8}

 &  & \multicolumn{2}{l|}{ $\RS_1$ }  & \textit{ $T_2$} & $\predue$  & $m_2$ & $T_2,T_3,T_4,T_6,T_7,T_8$ \\ \cline{3-8}

&  & \multirow{2}{*}{\raggedleft $\RS_2$ }  & $\RS_{2.1}$ & \textit{ $T_3$} & $\pretre$  & $m_{3}$ & $T_2,T_3,T_8$ \\ \cline{4-8}

&  &  & $\RS_{2.2}$ & \textit{ $T_4$} & $\prequattro$  & $m_{4}$ & $T_2,T_4,T_6,T_7$ \\ \cline{3-8}
                          
&  & \multicolumn{2}{l|}{ $\RS_3$ }  & \textit{ $T_5$} & $\precinque$  & $m_5$ & $T_2,T_5,T_7$ \\ \cline{3-8}                         
 
&  & \multicolumn{2}{l|}{ $\RS_4$ }  & \textit{ $T_6$} & $\presei$   & $m_6$ & $T_3,T_6,T_9$ \\ \cline{3-8} 

&  & \multicolumn{2}{l|}{ $\RS_5$ }  & \textit{ $T_7$} & $\presette$  & $m_7$ & $T_7,T_8,T_9,T_{11}$ \\ \cline{2-8}                                                                   

&  \multicolumn{3}{l|}{ \PPF }  &  \textit{ $T_8$} & $\preotto$ &  $m_8$  & $T_8,T_9,T_{11}$ \\ \cline{2-8}
&  \multicolumn{3}{l|}{ \Fin }  & \textit{ $T_9$} &  $\prenove$ &  $m_9$  & $T_9,T_{11}$ \\ \cline{2-8}
&  \multicolumn{3}{l|}{ \SC}  &  \textit{ $T_{10}$} & $\predieci$ &  $m_{10}$  & $T_{10},T_{11}$ \\ \cline{2-8}
&  \multicolumn{3}{l|}{ \Term }  & \textit{ $T_{11}$} & $\preundici$ &  $\nil$  & $T_{11}$\\ \hline%\hline

\hline
  \end{tabular}
\end{center}
\egroup
\end{table*}

% ------------------------------ inizio tabella mosse
\begin{table*}[ht]
\small
\caption{ Moves associated to tasks. }
\label{tab:moves}
\bgroup
\def\arraystretch{1.3}%  1 is the default, change whatever you need
\setlength{\tabcolsep}{5pt}
\begin{center}
  \begin{tabular}{ | c | p{0.87 \textwidth} | }
    \hline
    \textit{move}  &  \textit{definition} \\ \hline \hline
      
    $m_1$ & Robot $r \in \partial C_{\uparrow}^{1}(R)$ moves radially to $\CB$ \\ \hline
    $m_2$ & 
            Let $C = C_{\uparrow}^{i}(R)$ be the circle contained in  $\Ann$ and 
            with minimum index $i$. 
            If $\partial C\setminus \M'(C)\neq \emptyset$ then
            let 
               $R_2$ be the set of robots in $\partial C\setminus \M'(C)$ of minimal view else let $R_2$ be the set of 
            robots on $C$ of minimal view -- 
            call $\GoToC( R_2 )$ \\ \hline
    $m_{3}$ &  If $\partial C(R)\setminus \M'(C(R))\neq \emptyset$ then let 
               $R_3$ be the set of robots in $\partial C(R)\setminus \M'(C(R))$ of minimal view else let $R_3$ be the set of robots on $C(R)$ of minimal view --
               call $\GoToC( R_3 )$\\ \hline
    $m_{4}$ &  Let $r$ be the non-critical robot in $\partial C(R)$ of minimal 
               view and let $R_4= \{r\}$ --   
               call $\GoToC( R_4 )$ \\ \hline
    $m_5$ & 
    A point $p \in C(R)$ is said \emph{forbidden} for $C(R)$ if it forms an angle of $\frac {2\pi} n \cdot k$ degrees in $c(R)$ with any robot on $C(R)$, for $k = 0,1,\ldots, n$ (with $n$ being the number of robots);
        Let $r$ be the robot in $\partial C_{\downarrow}^{1}(R)$ having minimum view; 
            $r$ moves toward $C(R)$ avoiding forbidden points \\ \hline
    $m_6$ & The three robots on $C(R)$ form a triangle with angles 
            $\alpha_1 \ge \alpha_2 \ge \alpha_3$ and let $r_1$, $r_2$ and $r_3$ 
            be the three corresponding robots. For equal angles, the role of the 
            robot is selected according to the view, i.e. if $\alpha_1 = \alpha_2$ 
            then the view of $r_1$ is smaller than that of $r_2$. Robot $r_2$ rotates toward 
            the point $t$ such that $\alpha_1$ becomes of $90^{\circ}$ \\ \hline
    $m_7$ & Call $\Circle(\alpha)$, where $\alpha=2\pi / | \partial C(R)|$ \\ \hline
    $m_8$ & Call $\Distmin()$  \\ \hline
    $m_9$ & All robots in $\Ann \cup \CT$ radially move toward $C(R)$ \\ \hline
    $m_{10}$ & If $F$ is composed of one point with multiplicity $|R|$ then call 
            $\Gathering()$;\newline
            If $\rho(F)=1$ then call $\Leader()$
  \\ \hline

  \end{tabular}
\end{center}
\egroup
\end{table*}
% ------------------------------ fine tabella mosse

\begin{figure}[ht]
\begin{center}
\scalebox{0.50}{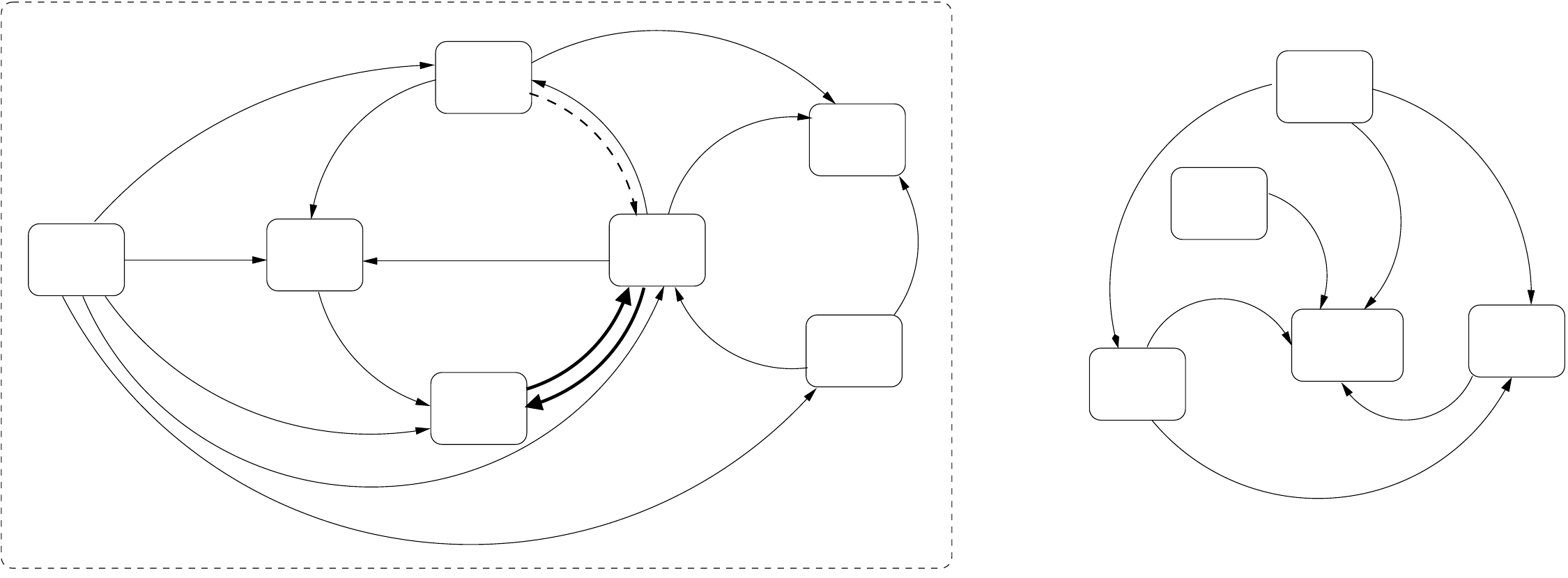 }
\end{center}
\caption{For sake of presentation, the transition graph is divided into two parts: $(a)$ transitions among tasks in $\RS$, and $(b)$ transitions among $\RS$ and the tasks associated to sub-problems $\PPF$, $\Fin$, $\SC$, $\Term$.  The transitions represented by bold arrows (from/to $T_2$ to/from $T_3$) are unclassified, the one represented by the dashed arrow (from $T_4$ to $T_2$) is robust, all the others are stationary. The types of the self-loops - omitted from each task - will be discussed in the correctness proof in Section~\ref{sec:correctness}. Notice that apart for the self-loops, the only simple cycles are: $(T_2,T_{3})$, $(T_2,T_{4})$, $(T_2,T_6,T_{3})$, $(T_2,T_{4},T_6,T_{3})$.
%\caption{For sake of presentation, the transition graph is divided into two parts: $(a)$ transitions among tasks in $\RS$, and $(b)$ transitions among $\RS$ and the tasks associated to sub-problems $\PPF$, $\Fin$, $\SC$, $\Term$. The transitions represented by a dashed arrow (the two between $T_2$ and $T_3$, and the one from $T_4$ to $T_2$) are robust, all the others are stationary. Notice that apart for the self-loops omitted from each task, the only simple cycles are: $(T_2,T_{3})$, $(T_2,T_{4})$, $(T_2,T_6,T_{3})$, $(T_2,T_{4},T_6,T_{3})$. 
}
\label{fig:transitions2}
\end{figure}

% ===================================================
% The example
% ===================================================
\section{Explanatory example and moves details}\label{sec:example}
In this section we provide an explanatory example about the behavior of the proposed algorithm for the \pf problem.  We take advantage of this example to provide the missing details about moves. In particular, we provide the pseudo-code of procedures $\GoToC$, $\Distmin$, and $\Circle$, along with their correctness. We also briefly discuss how algorithms $\Gathering$ from~\cite{CFPS12} and $\Leader$ from~\cite{CDN19} are exploited. We also formally prove some properties about these procedures.

\begin{figure}[ht]
\begin{center}
\scalebox{0.95}{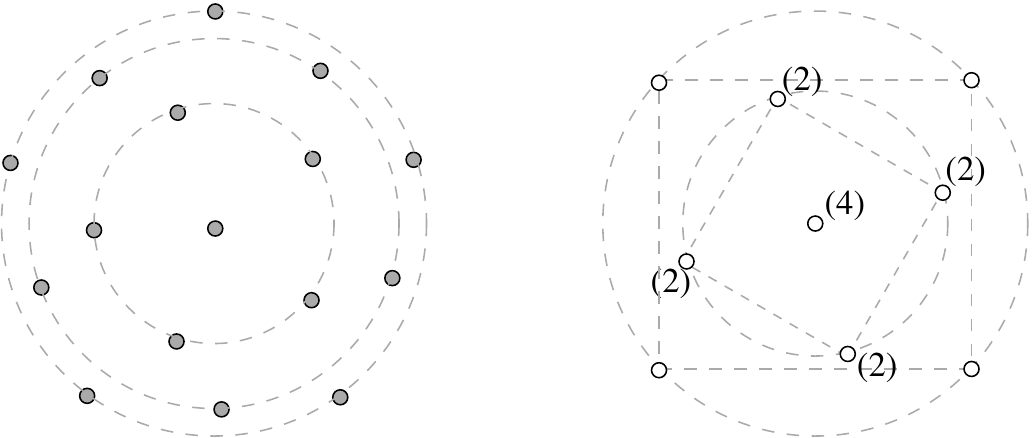 }
\end{center}
\caption{The input for the \pf\ problem that we use as running example throughout Section~\ref{sec:example}. Notice that the initial configuration $R$ is composed of 16 robots and $\rho(R)=1$, while the pattern $F$ has symmetricity $\rho(F)=4$ (numbers close to points refer to multiplicities).
}
\label{fig:input}
\end{figure}

The example is based on the input defined in Figure~\ref{fig:input}. Notice that both the configuration $R$ and the pattern $F$ defined in the example are symmetric but $\rho(R)=1$ and $\rho(F)=4$. In the next subsections, we analyze each task separately, according to the order dictated by a possible execution of the algorithm.

% --------------------------------------------------------
\subsection{Task $T_1$}
This task is associated to the sub-problem \SB. As already remarked, this sub-problem is thought for breaking possible symmetries by moving a robot $r$ from $c(R)$ (i.e., when $\rho(R)=1$). 

Concerning the current example, we now show that configuration $R$ in Figure~\ref{fig:input} belongs to task $T_1$. Each robot can detect this situation by evaluating the predicates characterizing each task. First, notice that variable $\xc$ holds in $R$, and this immediately implies that the configuration does not belong to any of tasks $T_2$, $\ldots$, $T_6$ (in fact, from Table~\ref{tab:tasks-bis} it follows that variable $\xc$ is negated in each precondition of these tasks). Since there are five robots on $C(R)$ and $\rho(F)=4$, then each robot deduces that both $\xduno$ and $\xddue$ are true in $R$: this implies that $R$ does not belong to $T_7$ nor to $T_8$.  
Variable $\xp$ is false in $R$ since $F$ cannot be obtained by radially projecting on $C(R)$ all robots in $\Ann \cup \CT$ (to observe $\Ann$ and $\CT$ refer to Figure~\ref{fig:T1}). According to the value of $\xp$, $R\not \in T_9$.
Variable $\xg$ is false as $\rho(F)=4$, hence $R\not \in T_{10}$. Finally, $\xw$ is false as $R$ is not similar to $F$ and hence $R\not \in T_{11}$. By concluding this analysis, it follows that $R$ does not belong to any of tasks $T_2$, $\ldots$, $T_{11}$ and according to precondition of $T_1$ and to definition of predicate $P_1$ -- cf. Equation~\ref{eq:2}, it follows that $R\in T_1$.

Since $R\in T_1$ then move $m_1$ is applied by the algorithm (cf. Figure~\ref{fig:T1}, left side). Robot $r$ located on $c(R)$ is moved radially along any direction to reach the parking circle $\CB$ in order to guarantee stationarity.\footnote{For the sake of completeness the exact direction toward which the robot moves will be specified in Section~\ref{sec:correctness}.}
It is worth to remark that even though the initial configuration does not admit symmetry, but there is a robot at a distance from $c(R)$ smaller than $\delta(\CB)$, then it is moved to the parking circle $\CB$ before starting any other task. 

Once the robot in $c(R)$ has reached the specified target (possibly within multiple LCM-cycles), configuration in Figure~\ref{fig:T1}, right side, is obtained. The obtained configuration is stationary and belongs to task $T_2$.

\begin{figure}[ht]
\begin{center}
\scalebox{1.0}{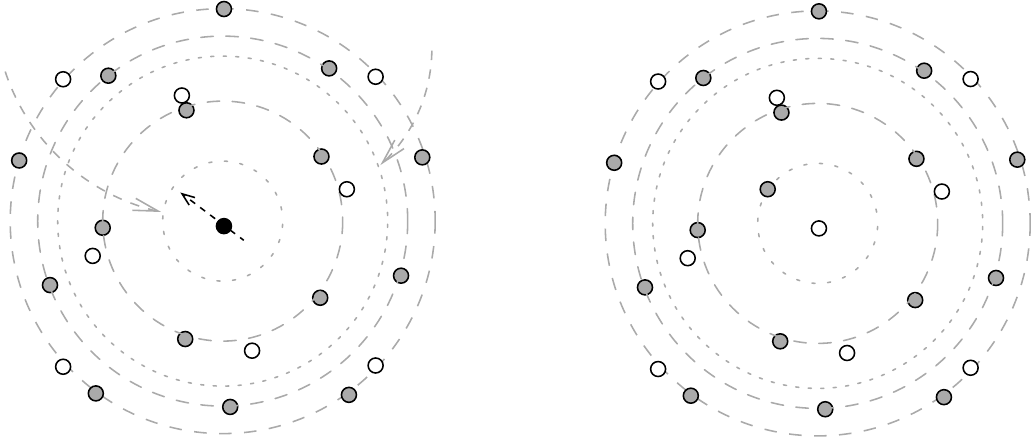 }
\end{center}
\caption{Task $T_1$: Ensure $c(R)$ empty. Notice the parking circles $\CT$ and $\CB$.
}
\label{fig:T1}
\end{figure}
%

% --------------------------------------------------------
\subsection{Task $T_2$}
This task is responsible for the correct removal of the robots from $\Ann$, and their movement toward the parking circle $\CT$ without generating unsolvable configurations. This removal is done in order to guarantee stationarity when later the algorithm starts removing robots from $C(R)$, when needed. Notice that there might be a number of robots equal to $\rho(R)$ that can move concurrently according to $m_2$ (this occurs when the processed configuration is symmetric). 

To perform this task, all robots in $\Ann$ eventually move according to the trajectory computed by Procedure $\GoToC$ specified in Algorithm~\ref{alg:gotoc} and used by move $m_2$. 

When Procedure $\GoToC$ is executed by a robot $r$, such robot is required to move toward a point of an arc of $\CT$ denoted as $A'_r$. In particular, $r$ is required to reach the leftmost endpoint (denoted as $a_r$) of $A'_r$ or the middle point of $A'_r$ according whether $a_r$ is a ``forbidden point for $\CT$'' or not. Informally, a point of $\CT$ is forbidden if it may form a regular $n$-gon along with the points occupied by some robots already located on $\CT$. The rationale underlying this definition is that when $r$ reaches $\CT$ all robots in such a circle are non-equivalent; this helps to ensure that no unsolvable configurations are created. Concerning the formal definition of $A'_r$, it depends on $r$ and various other parameters (for a visualization of the most of them, refer to Figure~\ref{fig:T2}). In what follows we formalize all such parameters. To this aim, assume that $\GoToC$ takes as input a set of robots $R_x \subseteq Rob(\Ann\cup C(R))$:

%_______________________ Dettagli su \GoToC()
 \begin{itemize}
 	\item  Let $r\in R_x$ and $h=\halfline (c(R),r)$;
 	\item  Let $r^-$ be the robot on $C(R)$ such that $h^-=\halfline (c(R),r^-)$ overlaps $h$ by the minimal clockwise rotation;
% 	\item Let $r^+$ be a robot in $Ann$, if any, such that $h$ overlaps $h^+=\halfline (c(R),r^+)$ by the minimal clockwise rotation;
 	\item Let $r^+$ be a robot in $\Ann\cup C(R)$ such that $h$ overlaps $h^+=\halfline (c(R),r^+)$ by the minimal clockwise rotation;
 	\item Let $\alpha$ be the size of the smallest angle greater than $\angolo (r^-,c(R),r)$, formed in $c(F)$ between two consecutive targets on $C(F)$;
 	\item Let $h'$ be the half-line obtained by rotating clockwise $h^-$ of $\alpha$ degrees;
 	\item Let $A_r$ be the portion of $\CT$ delimited by $h$ and the closest half-line between $h'$ and $h^+$. Let $a_r$ and $b_r$ the end points of $A_r$, such that $b_r$ follows $a_r$ in the clockwise order;
 	\item A point $p \in \CT$ is said \emph{forbidden} for $\CT$ if it forms an angle of $\frac {2\pi} n \cdot k$ degrees in $c(R)$ with any robot on $\CT$, for $k = 0,1,\ldots, n$ (we recall the reader that $n$ denotes the number of robots);
 	\item Let $A'_r$ be the sub-arc of $A_r$ starting from $a_r$ and ending at the closest point between $b_r$ and the first forbidden point for $\CT$ different from $a_r$ met in the clockwise order along $A_r$, if any.
 \end{itemize}

\begin{algorithm}
\caption{$\GoToC(R_x)$}\label{alg:gotoc}
\begin{algorithmic}[1]
\begin{small}
\IF{$r\in R_x$}
\IF{$a_r$ is not forbidden for $\CT$}
\STATE $r$ straightly moves toward $a_r$ 
\ELSE
\STATE Let $q$ be the middle point of arc $A'_r$;
\STATE $r$ straightly moves toward $q$ until reaching $\CT$ on the closest intersection point of $\CT$ and $[r,q]$.
\ENDIF
\ENDIF
\end{small}
\end{algorithmic} 
\end{algorithm}

\begin{figure}[ht]
\begin{center}
\scalebox{1.0}{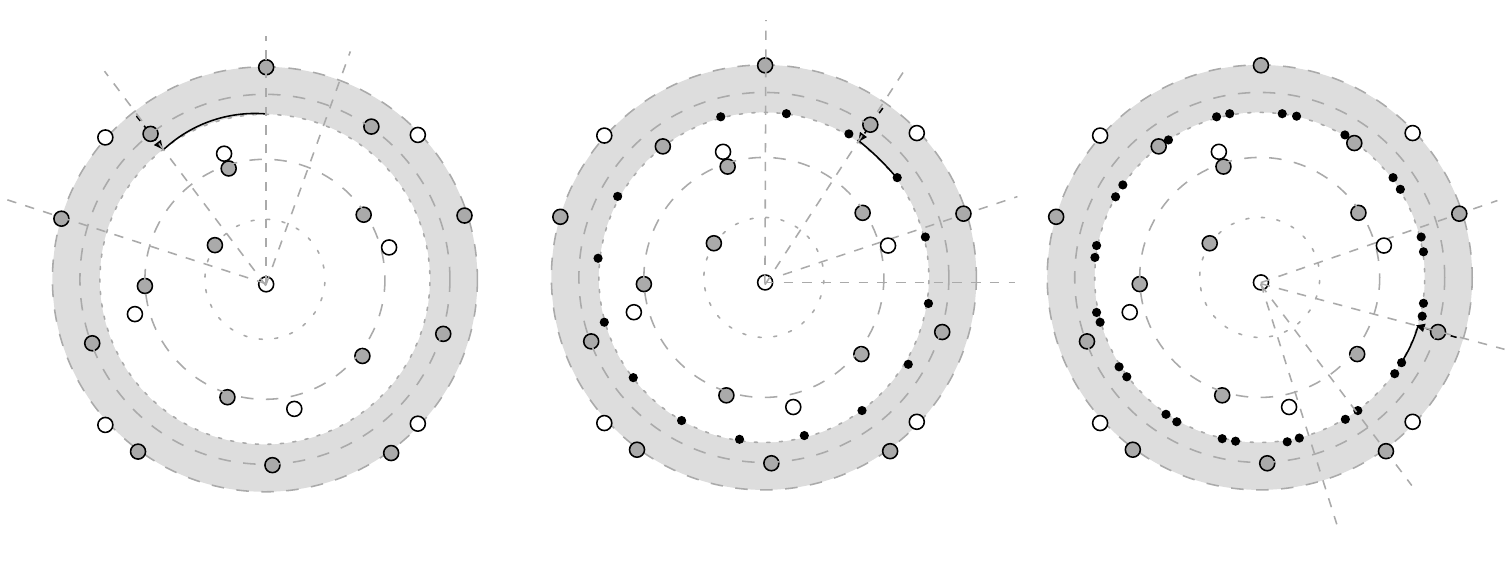 }
\end{center}
\caption{Task $T_2$: Make $\Ann$ (the light-gray corona) empty to ensure stationarity (notice that only the trajectories of the first three moving robots are shown). Small black dots represent forbidden points for $\CT$. 
}
\label{fig:T2}
\end{figure}

By considering again our running example, we have configuration at Figure~\ref{fig:T1}, right side, as input for the current task $T_2$. As done for the analysis of task $T_1$, we now formally show that such configuration belongs to $T_2$. 

As analyzed for task $T_1$ we have the same values for variables $\xw$, $\xg$, $\xp$, $\xduno$, $\xddue$, so the configuration is not in $T_7$, $T_8$, $T_9$, $T_{10}$, and $T_{11}$. Variable $\xa$ is false since $\Ann$ contains robots. Hence the configuration is not in $T_6$, $T_4$, and $T_3$. About $T_5$, we have that $\xf$ is false as there are too many robots on $C(R)$ with respect to $\rho(F)$. Since variable $\xc$ is now false, then the configuration belongs to $T_2$.
  
By applying move $m_2$, all robots in $\Ann$ eventually move toward $\CT$. In particular, since all robots in $\Ann$ reside on a single circle (say $C$) and $\M(C)=\emptyset$ (on $C$ there are no regular $m$-gons such that $m>1$ and $m$ divides $\rho(F)=4$), then $m_2$ calls $\GoToC( R_2 )$ with $R_2$ containing all robots on $C$. The trajectories performed by some robots in $C$ are illustrated in Figure~\ref{fig:T2} (notice that, for sake of presentation, we assume that in such an example the asynchronous scheduler makes active one robot in $C$ at a time - cf. Figure~\ref{fig:T2} where the first three executions of $\GoToC$ are illustrated).  Once all robots in $\Ann$ reach $\CT$, as we will show the obtained configuration (cf. Figure~\ref{fig:T4}, left side) belongs to task $T_4$.

%%%%%%%%%%%%%%%%%%%%%%%%%%%%%%%%%%%%%%%%%%%%%%%%%%%%%%%%%%%%%%%%%%%%
%%%%%%%%%%%%%%%%%%  -------- inizio lemma correttezza GoToC
%%%%%%%%%%%%%%%%%%%%%%%%%%%%%%%%%%%%%%%%%%%%%%%%%%%%%%%%%%%%%%%%%%%%

The next lemma gives important properties of Procedure $\GoToC$ when applied to an initial configuration belonging to $T_2$.

\begin{lemma}\label{lem:GoCorrectness}
 Let $R=R(t_0)$ be an initial configuration at time $t_0$ belonging to $T_2$, and $S(t_0)$ be the set of robots to move according to $m_2$.
There exists a time $t_k> t_0$ where the reached configuration $R'=R(t_k)$ differs from $R$ only for robots in $S(t_0)$ that are all on $\CT$ in $R'$, such that the following properties hold:
 \begin{enumerate}
  \item \label{p:1} 
        $R(t_i)$ belongs to $T_2$ 
        for each $t_0 < t_i < t_k$;
  \item\label{p:2} 
        $\rho(R(t_i))$ divides 
        $\rho(F)$ for each $t_0 < t_i \le t_k$;
  \item\label{p:3} $R'$ is stationary;
  \item\label{p:4} $R(t_i)$, $t_0 \leq t_i \leq t_k$, has no  multiplicities.
 \end{enumerate}
\end{lemma}

\begin{proof}
We now prove the existence of $R'$ and each property at Items~\ref{p:1}--\ref{p:4} in the statement.
\begin{itemize}
\item
% 1 ---------
About the existence of $R'$ and property at Item~\ref{p:1}.  

Let $C_{\uparrow}^{i}(R)$ be the circle in $\Ann$ closest to $\CT$. Then $S(t_0)\subseteq \Rob(C_{\uparrow}^{i}(R))$ according to move $m_2$. The call $\GoToC(S(t_0))$ aims to move all robots in $S(t_0)$ toward $\CT$. Let $\bar R = R(t_i)$, $t_0 < t_i < t_k$, and assume that some robots in $S(t_0)$ are not on $\CT$ in $\bar R$. We now show that $\bar{R}$ is still in $T_2$. 
  
Clearly $\bar{R}$ does not belong to $T_{11}$ as there are robots in $\Ann$. It does not belong to $T_{10}$ because of $\xg$ that only depends on $F$. In order to show it does not belong to $T_9$, it is sufficient to remind the area within a robot $r$ is moving according to Procedure $\GoToC$. In fact, this ensures that $\xp$ remains false because of the limit established by angle $\alpha$. Such a limit guarantees that along all its movement $r$ cannot be in a position corresponding to the projection of a point from $C(F)$ to $\CT$. $\bar{R}$ is not in $T_8$, $T_7$, $T_6$, $T_4$ and $T_3$ because $\xa=\false$. It is not in $T_5$ because from $\pre_2 \wedge \neg \pre_5$ we deduce that $\xf=\false$ in $R$ and since no robots are moved from $C(R)$ then $\xf$ remains false. Then $\bar{R}$ is in $T_2$ because the value of $\xc$ has not changed.

Being $\bar{R}$ in $T_2$, again Procedure $\GoToC$ is applied. Note that, the input provided to the successive calls of Procedure $\GoToC$ is constituted by a subset $S(t_i)\subseteq S(t_0)$. In fact, it involves robots lying on the current circle $C_{\uparrow}^{j}(R)$ in $\Ann$ closest to $\CT$, whose radius is certainly not greater than that of the initial $C_{\uparrow}^{i}(R)$ from where robots in $S(t_0)$ were selected. By applying the arguments above, we can state that by repeatedly applying $\GoToC$, the algorithm will lead all robots in $S(t_i)$ to reach $\CT$. This implies there exists a time where the portion of $\Ann$ delimited by $C_{\uparrow}^{i}(R)$ and $\CT$, and excluding such circles, will not contain robots, eventually.
At this time, either all robots contained in $S(t_0)$ have reached $\CT$ or some of them are still on $C_{\uparrow}^{i}(R)$. In the latter case, move $m_2$ ensures to call $\GoToC$ providing as input only the robots originally contained in $S(t_0)$. 

By reconsidering the above analysis, we conclude that all configurations generated while robots in $S(t_0)$ are moved toward $\CT$ belong to $T_2$. Once all such robots reach $\CT$, say at time $t_k> t_0$, then the requested configuration $R'$ is obtained.

% 2 ---------
\item
About property at Item~\ref{p:2}. 
Consider two different cases for $R=R(t_0)$: $\Rob(\CT)=\emptyset$  and $\Rob(\CT)\not = \emptyset$.

If $\Rob(\CT)=\emptyset$, let us first analyze the case when $\partial C_{\uparrow}^{i}(R)\setminus \M'(C_{\uparrow}^{i}(R)) = \emptyset$. When $m_2$ is applied to configuration $R=R(t_0)$, at most $S(t_0)$ robots will move at the same time. If  more than one robot moves, this is because they are of minimal view and, by Lemma~\ref{lem:divM}, if one of them belongs to an element $M \in \M(C_{\uparrow}^{i}(R))$, then all the other robots belong to the same regular $|M|$-gon. The robots move radially toward $\CT$, as so far there is no forbidden point for $\CT$. 
The robots that trace concurrently the same distance could form a regular $|M'|$-gon, but in this case, $|M'|\leq |M|$ and $|M'|$ divides $\rho(R)$. Then the symmetricity of the whole configuration divides $\rho(R)$,  which in turn divides $\rho(F)$.
Possibly, some robots reach $\CT$ whereas some other are stopped before by the adversary or they do not start moving yet. In such cases, the trajectories of the robots might change in order to reach $\CT$ by avoiding the forbidden points generated by robots arrived on $\CT$.
Then, each configuration obtained while the remaining robots move toward $\CT$ cannot have a symmetricity larger than $\rho(R)$ (this could be obtained only if the robots reach the forbidden points for $\CT$). Moreover the symmetricity of any of these configurations has to divide $\rho(R)$ because, otherwise, there is an automorphism $\varphi$ such that one robot $r$ of the first arrived on $\CT$ should be equivalent to a robot $r'=\varphi(r)$, but this violates the requirement for $r'$ to avoid forbidden points for $\CT$.

As the above property holds for each generated configuration $\bar{R}$, when robots reach $\CT$ by successive calls of Procedure $\GoToC$, then we conclude $\bar{R}$ is such that $\rho(\bar{R})$ divides $\rho(R)$ and then $\rho(F)$. The same considerations hold for $\rho(R')$.
 
Let us now analyze the case when $\Rob(\CT)=\emptyset$ and $\partial C_{\uparrow}^{i}(R)\setminus \M'(C_{\uparrow}^{i}(R)) \neq \emptyset$. Let $\M'(C_{\uparrow}^{i}(R)) \neq \emptyset$. Similarly as above, Procedure $\GoToC$ is called until all the robots in $S(t_0)$ are moved from $C_{\uparrow}^{i}(R)$ to $\CT$. By Lemma~\ref{lem:divM}, the symmetricity of each generated configuration $\bar{R}$ as well as $R'$ divide $|M|$. Then both $\rho(\bar{R})$ and $\rho(R')$ divide $\rho(F)$.
If instead $\M'(C_{\uparrow}^{i}(R)) =\emptyset$, then $|S(t_0)|=1$, the configuration is asymmetric and it is maintained as such by means of Procedure $\GoToC$ because the only moved robot cannot be equivalent to any other until it reaches $\CT$. Then for each generated configuration $\bar{R}$, $\rho(\bar{R})=\rho(R')=1$ that obviously divide $\rho(F)$.

\smallskip
Finally, consider the case when $\Rob(\CT)\not = \emptyset$ in $R(t_0)$. %Robots on $\CT$ possibly form regular $|M|$-gons with $\rho(R)$ that divides $|M|$. 
The analysis is basically the same as above, with the only difference that now there are already some forbidden points for $\CT$ and hence the trajectories of robots in $S(t_0)$ initially are not necessarily radial toward $\CT$.
%If $|M|=1$ the configuration is asymmetric and since no robots of $S(t_0)$ will reach a forbidden point for $\CT$ we have that $\rho(R')=1$ as well as $\rho(\bar{R})=1$, that obviously divide $\rho(F)$. Otherwise, we can apply an analysis similar to the one above, %the same analysis based on Lemma~\ref{lem:divM} 
%as done before
%hence obtaining the same conclusions.

% 3 ---------
\item
About property at Item~\ref{p:3}.
As shown above, starting from $R$, all calls of Procedure $\GoToC$ only involve robots originally contained in $S(t_0)$. Any other robot does not move, that is it is stationary. Once all the robots in $S(t_0)$ reach $\CT$, $R(t_k)=R'$ is obtained which is then stationary.

% 4 ---------
\item
About property at Item~\ref{p:4}.
According to Procedure $\GoToC$, configuration $R'$ has no multiplicities since each robot $r$ moves toward $\CT$ in a region of $\Ann$ confined by: $\CT$, the rays from $c(R)$ passing through $r$ itself, and the next robot $r^+$ in the clockwise direction on $\Ann\cup C(R)$. In this region there are no robots and no other robots enter such a region. Moreover, the destination point on $\CT$ cannot be occupied by a robot, as otherwise by definition it would be a forbidden point for $\CT$. 
\end{itemize}
\end{proof}

%%%%%%%%%%%%%%%%%%%%%%%%%%%%%%%%%%%%%%%%%%%%%%%%%%%%%%%%%%%%%%%%%%%%
%%%%%%%%%%%%%%%%%%  -------- fine lemma correttezza GoToC per m2
%%%%%%%%%%%%%%%%%%%%%%%%%%%%%%%%%%%%%%%%%%%%%%%%%%%%%%%%%%%%%%%%%%%%

% --------------------------------------------------------
\subsection{Task $T_4$} 
In order to solve the sub-problem $\RS$, that is the creation of a common reference system, task $T_4$ is meant to manage the cases in which there are too many robots on $C(R)$ with respect to $\rho(F)$. In particular, task $T_4$ is specialized to manage the cases $\M(C(R)) = \emptyset$. We recall that $\M(C(R))$ denotes the set containing all the maximum cardinality subsets $M \subseteq \partial C(R)$ such that $|M|>1$, robots in $M$ form a regular $|M|$-gon, and $|M|$ divides $\rho(F)$. Since the input configuration $R$ and the pattern to form must guarantee that $\rho(R)$ divides $\rho(F)$, then $\M(C(R)) = \emptyset$ implies that $R$ is asymmetric. This allows the algorithm to remove one robot at a time from $C(R)$ until exactly $m$ robots remain, with $m$ being the minimal prime factor of $\rho(F)$ or $m=3$.

Clearly, the removal of robots must be done very carefully so as to guarantee that $C(R)$ does not change (hence, each time the moving robot must be non-critical). Moreover, if $\rho(F)$ is even and hence only two robots must remain in $C(R)$, then it is possible that $T_4$ must terminate with three robots on $C(R)$ instead on two (it is possible that each of the three remaining robots is critical). In this case, task $T_6$ is required before the removal of the last robot from $C(R)$, that is two antipodal robots must be created on $C(R)$ as otherwise the smallest enclosing circle of the robots would change with respect to the initial one.

For this task, again Procedure $\GoToC$ is used. According to move $m_4$, it is performed by the non-critical robot in $\partial C(R)$ of minimal view. In this way, the moving robot will reach $\CT$ by also ensuring that the new configuration still guarantees that $\rho(R)$ divides $\rho(F)$. It is worth to remark that in case the moving robot is stopped by the adversary before reaching the parking circle, then task $T_2$ is applied again to make $\Ann$ empty (in other words, $T_2$ collaborates with $T_4$ to correctly transfer robots from $C(R)$ to $\CT$). 

\begin{figure}[ht]
\begin{center}
\scalebox{1.0}{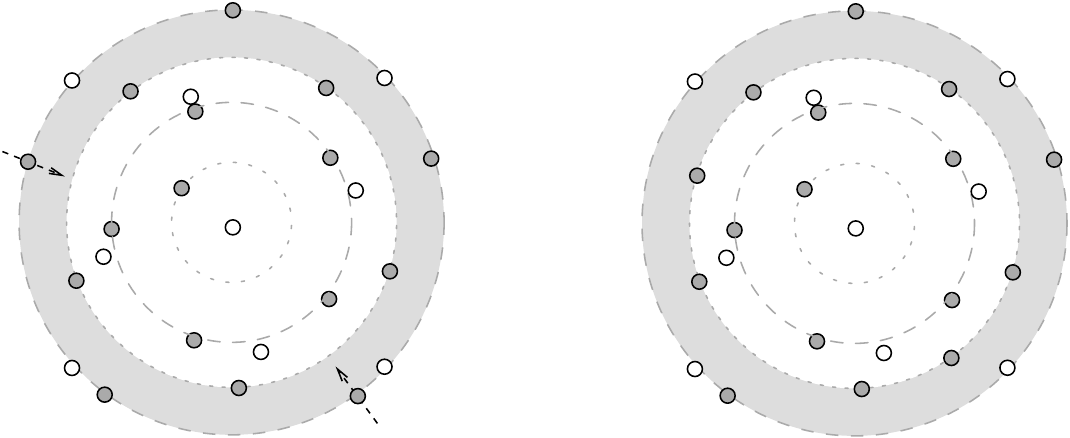 }
\end{center}
\caption{Task $T_4$: Case $\M(C(R)) = \emptyset$, removing robots from $C(R)$ until exactly $m$ robots remain, with $m$ being the minimal prime factor of $\rho(F)$ or $m = 3$. The configuration on the left side is obtained from Figure~\ref{fig:T2} after all robots in $\Ann$ reached the parking circle $\CT$.
}
\label{fig:T4}
\end{figure}

Concerning the running example, Figure~\ref{fig:T4} (left side) shows the configuration belonging to task $T_4$. This membership can be verified as follows. 
As analyzed for tasks $T_1$ and $T_2$ we have the same values for variables $\xw$, $\xg$, $\xp$, $\xduno$, $\xddue$, so the configuration is not in $T_7$, $T_8$, $T_9$, $T_{10}$, and $T_{11}$. Variables $\xt$ and $\xf$ are both false, so the configuration is not in $T_6$ nor in $T_5$. Since the precondition $\pre_4=\prequattro$ holds (in fact, here $\xa=\true$, $\xc=\false$, and $\xm=\true$), then the predicate $P_4$ holds and hence the current configuration belongs to $T_4$. 

Figure~\ref{fig:T4} (right side) shows the stationary configuration obtained after two consecutive applications of task $T_4$. Since this configuration contains three robots on $C(R)$ and $\rho(F)=4$, then it must be processed by $T_6$ in order to guarantee two antipodal robots on $C(R)$ before leaving two robots on $C(R)$.

% --------------------------------------------------------
\subsection{Task $T_6$}
This task is performed when there are exactly three robots on $C(R)$, 3 does not divide $\rho(F)$, and $\rho(F)$ is even. In such a case, one of the three robots, chosen so as to not modify $C(R)$, rotates until it becomes antipodal with respect to one of the other two robots. Once this happens, variable $\xm$ becomes false since a regular $2$-gon is created on $C(R)$.

\begin{figure}[ht]
\begin{center}
\scalebox{1.0}{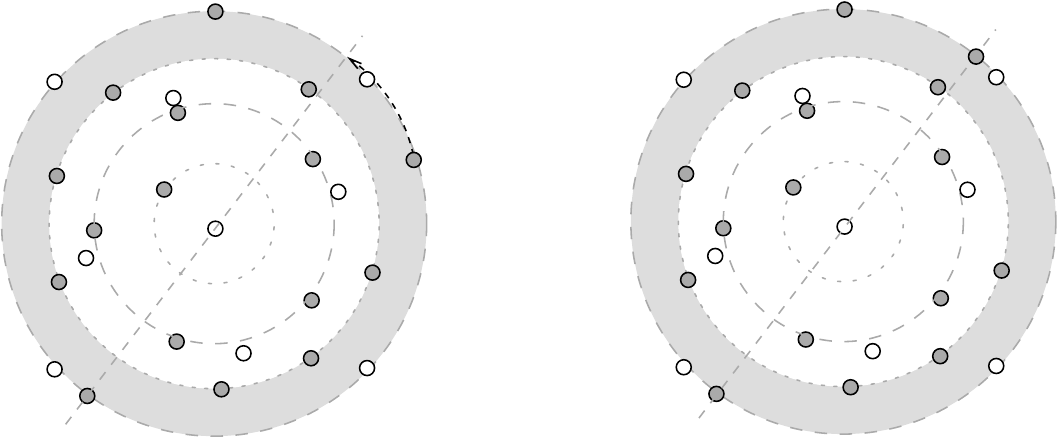 }
\end{center}
\caption{Task $T_6$: Create two antipodal robots on $C(R)$.
}
\label{fig:T6}
\end{figure}

Consider the running example of Figure~\ref{fig:T6} (left side). This configuration belongs to $T_6$. In fact, as analyzed in previous tasks we have the same values for variables $\xw$, $\xg$, $\xp$, $\xduno$, $\xddue$, so the configuration is not in $T_7$, $T_8$, $T_9$, $T_{10}$, and $T_{11}$. Instead, now $\xa=\true$ (i.e., there are no robots in $\Ann$), $\xc=\false$ (i.e., there are no robots in the interior of $\CB$), $\xm=\true$ (i.e., there are no regular 2-gons in $C(R)$), and $\xt=\true$ (i.e., $|\partial C(R)| = 3$ and 2 is a divisor of $\rho(F)$). Hence the predicate defining $T_6$ is true. 

The three robots on $C(R)$ form a triangle with angles $\alpha_1 \ge \alpha_2 \ge \alpha_3$ where $r_1$, $r_2$ and $r_3$ are the three corresponding robots. The move planned for this task (cf. move $m_6$) rotates $r_2$ along $C(R)$ so as to obtain a configuration with two antipodal robots on $C(R)$. Once this happens (and, as usual, it may require multiple LCM cycles), the configuration belongs to $T_3$ as there is a regular 2-gon on $C(R)$, with 2 being a divisor of $\rho(F)$ but with a third robot that must be moved from $C(R)$ toward $\CT$. Such a movement initiated by $T_3$ might be continued via task $T_2$ if the robot does not conclude its movement within one LCM cycle.

% --------------------------------------------------------
\subsection{Task $T_3$}
Together with task $T_4$, this task is meant to manage the cases in which there are too many robots on $C(R)$ with respect to $\rho(F)$. In particular, task $T_3$ is specialized to manage the case in which $\M(C(R)) \neq \emptyset$. %We recall that $\M(C(R))$ denotes the set containing all the maximum cardinality subsets $M \subseteq \partial C(R)$ such that $|M|>1$, robots in $M$ form a regular $|M|$-gon, and $|M|$ divides $\rho(F)$. 

The move planned for this task is $m_3$ and it carefully moves robots from $C(R)$ toward the parking circle $\CT$ by means of Procedure $\GoToC$. According to its specification, we observe that it considers two cases: (1) if $\partial C(R)\setminus \M'(C(R))\neq \emptyset$ then all robots of minimal view in $\partial C(R)\setminus \M'(C(R))$ are moved, otherwise (2) all robots on $C(R)$ of minimal view are moved. %In the latter case all the robots in one regular $m$-gon are moved (because in this case the configuration is symmetric).
%%%%%%%%%%%%%%%%%%%%%%%%%%%%%%%%%%%%%%%%%%%%%%%%%%%%%%%%%%%%%%%
%\linecomment{GAB ALF}{fare riferimento al Lemma~\ref{lem:divM} } 
%\linecomment{ser}{penso che serva solo in fase di correttezza, qui diamo solo una descrizione ad alto livello}
%%%%%%%%%%%%%%%%%%%%%%%%%%%%%%%%%%%%%%%%%%%%%%%%%%%%%%%%%%%%%%%
Notice that it is possible that even though $R$ might be symmetric, its symmetricity is (or becomes) smaller than $\rho(F)$. However, by Lemma~\ref{lem:divM} we are ensured that $\rho(R)$ remains a divisor of $\rho(F)$ as long as $T_3$ is applied. Moreover, even in the possible case where $\rho(R)>1$, due to the \async model not all robots belonging to a same regular $m$-gon (say $M$) are necessarily active, and hence after some LCM cycles some of such robots may be in $\Ann$  while some other may still stay on $C(R)$. Any robot in $\Ann$ is then moved by $T_2$, and once $T_2$ has completely removed robots from $\Ann$, then the remaining robots of $M$ left on $C(R)$ are later processed again by $T_3$ since they result to be in $\partial C(R)\setminus \M'(C(R))$. 

It is worth to remark that, as soon as a robot leaves $C(R)$, variable $\xa$ becomes false, and task $T_2$ might be invoked. %: we will show in the next section that the involved robots are almost-stationary or robust (cf. Section~\ref{sec:methodology:correctness}), hence providing the means to guarantee the desired behavior. 

\begin{figure}[ht]
\begin{center}
\scalebox{1.0}{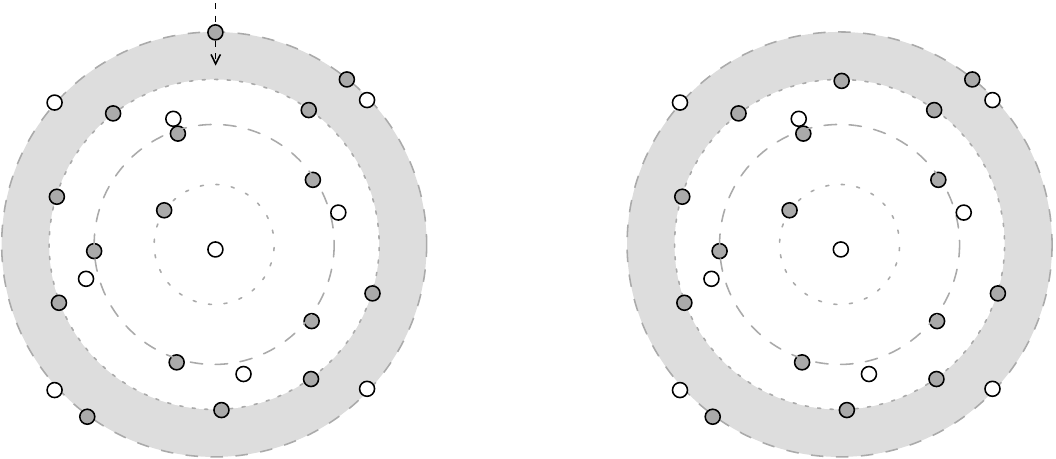 }
\end{center}
\caption{Task $T_3$: Case $\M(C(R)) \neq \emptyset$, removing robots from $C(R)$ until exactly one maximal regular $m$-gon of $\M$ remains.
}
\label{fig:T3}
\end{figure}

Consider the running example of Figure~\ref{fig:T3} (left side). This configuration belong to $T_3$. In fact, as analyzed in previous tasks we have the same values for variables $\xw$, $\xg$, $\xp$, $\xduno$, $\xddue$, so the configuration is not in $T_7$, $T_8$, $T_9$, $T_{10}$, and $T_{11}$. Variable $\xm=\false$ (there is one regular 2-gon in $C(R)$ and $\rho(F)$ is even), so the configuration is not in $T_6$ or $T_4$; variable $\xf=\false$ ($|\partial C(R)| = 3$ and $\rho(F)=4$), hence it does not belong $T_5$. In conclusion, since precondition $\pre_3=\pretre=\true$, then the configuration belongs to $T_3$. 
 
Move $m_3$, possibly interleaved by move $m_2$, will lead to obtain the configuration shown in Figure~\ref{fig:T3} (right side). In this configuration the problem $\RS$ is solved, and hence the subsequent sub-problem $\PPF$ can be addressed by performing the planned task $T_8$.

% --------------------------------------------------------
\subsection{Task $T_8$}
This task is responsible for solving the $\PPF$ sub-problem. In particular, it moves all robots that are inside or on $\CT$ toward the targets computed with respect to the embedding of the \emph{modified pattern $F'$}. As described in Section~\ref{ssec:algorithm:subdivision} (cf. description of $\PPF$), pattern $F'$ differs from $F$ only for those possible targets on $C(F)$ different from the $m$ ones already matched by the resolution of sub-problem $\RS$ (i.e., the embedding of $F$ on $R$ and hence the embedding of $F'$ on $R$ are well-defined, cf. description of $\RS$). Such additional points on $C(F)$, if any, are instead radially projected to $\CT$ in $F'$. In our strategy, task $T_8$ is designed to solve the pattern formation problem with respect to $F'$. 

\begin{figure}[ht]
\begin{center}
\scalebox{1.0}{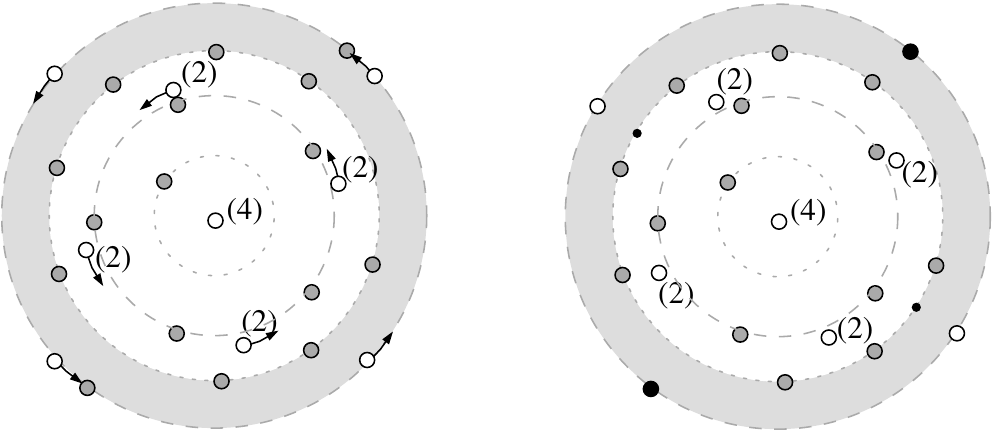 }
\end{center}
\caption{Task $T_8$. Make a partial pattern formation: embedding of $F$ and $F'$. 
The light-gray corona is the $\Ann$; gray circles represent robots; white circles represent points of $F$;
On the left, arrows represent how $F$ must be rotated according to the embedding defined in Section~\ref{ssec:algorithm:subdivision}.
On the right, black circles represent robots matched with points of $F$ after the embedding; finally, the two black dots on $\CT$ represent points of $F'$ obtained as radial projections of unmatched points of $F$ on $C(R)$.}
\label{fig:T8a}
\end{figure}

Concerning the running example, Figure~\ref{fig:T8a} shows how each robot views the embedding of $F'$ in the current configuration. It is worth to note that, during this task, (1) no robots on $C(R)$ move, and (2) no robots are moved out of $\CT$ (i.e., no robot enters in $\Ann$); this implies that the embedding of $F'$ remains the same during the whole task $T_8$. 

To solve $\PPF$, at any time, each robot inside $\CT$ must determine (1) whether it is already on its target or not (i.e., whether it is \emph{matched} or not), (2) if it is not matched, which is its target, and (3) whether it is its turn to move or not. 
To this aim, and to formally define Procedure $\Distmin$ that is used to solve task $T_8$, we need some further definitions and properties (cf. Figure~\ref{fig:sector}).

\begin{figure}[ht]
\begin{center}
\scalebox{0.80}{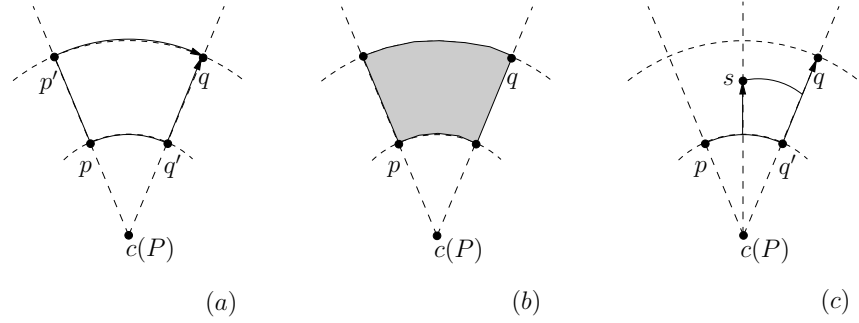 }
\end{center}
\caption{A representation of: $(a)$ the sectorial paths between points $p$ and $q$; $(b)$ the annulus sector $\AS(p,q)$ (the gray region); $(c)$ two shortest paths from $p$ to $q$ (one passing through $s$ and composed by two sectorial paths).
}
\label{fig:sector}
\end{figure}

Let $P$ be a multiset of points and let $p,q\in P$. We denote by $C_p$ and $C_q$ the circles centered in $c(P)$ and with radii $d(c(P),p)$ and $d(c(P),q)$, respectively. Points $p'$ and $q'$ correspond to $C_q\cap \halfline(c(P),p)$ and $C_p\cap \halfline(c(P),q)$, respectively  (cf. Figure~\ref{fig:sector}.$(a)$). Symbol $\AS(p,q)$ is used to denote the \emph{annulus sector} given by the area enclosed by circles $C_p$ and $C_q$, and by segments $[p,p']$ and $[q',q]$, subtending $\angolo(p,c(R),q)$ (cf. Figure~\ref{fig:sector}.$(b)$). Notice that when $\angolo(p,c(R),q) = \pi$, by definition $\AS(p,q)$ corresponds to the annulus sector spanned by $\halfline(c(P),p)$ to overlap $\halfline(c(P),q)$ by means of a clockwise rotation.
We say that $\AS(p,q)$ is \emph{degenerate} when it reduces to a point (i.e., when $p=q$) or to a segment/arc (i.e., when $p$ and $q$ lie on the same ray/circle).

\begin{definition}[Sectorial path and sectorial distance]
Let $P$ be a multiset of points in the plane. Given $p,q\in P$, the \emph{sectorial path}
between $p$ and $q$ is given by either the arc $\arc{pq'}$ composed with the segment $[q',q]$, or the segment $[p,p']$ composed with the arc $\arc{p'q}$  (cf. Figure~\ref{fig:sector}.$(a)$). 
The \emph{sectorial distance} between $p$ and $q$ is denoted by $\dist(p,q)$ and if $\delta(C(P))= 0$ then $\dist(p,q)=0$, else 
\[\dist(p,q) = |d(p,c(P)) - d(q,c(P))| / \delta(C(P)) + \min \{\angolo (p,c(P),q), \angolo (q,c(P),p) \}/ \pi.\]
\end{definition}
Informally, the sectorial distance is a sort of Manhattan distance where moving between two points is constrained by rotating along concentric circles centered at $c(P)$ and moving along rays starting from $c(P)$. It is easy to verify that function $\dist()$ is in fact a distance function.

\begin{property}\label{prop:multi-shortest-paths}
Let $P$ be a multiset of points in the plane, and let $p,q\in P$. For each point $s\in \AS(p,q)$ it follows that $\dist(p,q) = \dist(p,s)+\dist(s,q)$.
\end{property}

According to this property, the sectorial distance implies the existence of infinitely many shortest paths (composed of one or more sectorial paths) connecting two distinct points (cf. Figure~\ref{fig:sector}.$(c)$).

The above notation and definitions will be applied to what was before informally called a ``sector''. The following definition formalizes such a concept. 

\begin{definition}[Sector]
	Let $\ell$ and $\ell’$ be two consecutive (clockwise) robot-rays. A \emph{sector} $S$ is the area confined by $\ell$, $\ell’$, and $\CT$. Concerning the boundary, $\ell$ belongs to $S$, $\ell’$ does not belong to $S$, the portion on $\CT$ delimiting $S$ belongs to $S$, and $c(R)$ does not belong to $S$. $\Sector(R)$ denotes the set containing all the sectors of a configuration $R$. 
\end{definition}

We now exploit the sectorial distance to determine the trajectories used by robots to move toward the targets.

\begin{definition}[Safe trajectory]
Given a configuration $R$ and a sector $S\in \Sector(R)$, a robot $r \in \Rob(S)$ is said to admit a \emph{safe trajectory} toward a target point %$t\in \Sector(R)\cup c(R)$ 
$t\in S\cup c(R)$  if there exists a shortest path between $r$ and $t$ according to $\dist()$ that does not pass through any other robot. 
\end{definition}

The next statements (see Lemma~\ref{lem:degenerate} and Proposition~\ref{prop:asymmetric}) will play a central role for the definition of $\Distmin$. 

\begin{lemma}\label{lem:degenerate}
Given a configuration $R$  and a sector $S\in \Sector(R)$, let $r\in \Rob(S)$ and 
%$t\in \Sector(R)$ 
$t\in S$ be a target point. If $\AS(r,t)$ is not degenerate, then $r$ admits a safe trajectory toward $t$.	
\end{lemma}
\begin{proof}
The claim simply follows from Proposition~\ref{prop:multi-shortest-paths} that implies the existence of infinitely many shortest paths between $r$ and $t$, and by observing that $R$ is finite. 
\end{proof}

\begin{property}\label{prop:asymmetric}
For each sector $S$, the sub-configuration given by $\partial C(R) \cup Rob(S)$ is asymmetric.
\end{property}

\noindent
The above statements can be combined as follows:
the former ensures that when a robot $r$ moves toward a target $t$ and $\AS(r,t)$ is not degenerate, then $r$ admits a safe trajectory toward $t$; the latter says that inside a sector $S$ it is always possible to elect a leader $r\in \Rob(S)$. By combining them we get that inside a sector $S$ we can always elect a robot $r$ to move toward a target $t$, and if $\AS(r,t)$ is not degenerate then $r$ can move along a shortest path without creating collisions. Given a sector $S$, the following additional notation allow us to formalize  such an approach:
\begin{itemize}
\item $R^m(S)=Rob(S)\cap F'$ denotes the matched robots;
\item $F^m(S)= F'\cap R^m(S)$ denotes the matched targets;
\item $R^{\neg m}(S)= Rob(S) \setminus R^m(S)$ denotes the unmatched robots;
\item $F^{\neg m}(S)= (F'\cap S)\setminus F^m(S)$ denotes the unmatched targets;
%\item $\Delta(S) = \min \{ \dist(r,f): r\in R^{\neg m}(S), f\in F^{\neg m}(S)\}$ denotes the minimum distance $\Delta(S)$ according to $\dist()$ between unmatched robots and unmatched targets in $S$;
%\item $R^{\neg m}_{\Delta}(S)= \{ r\in R^{\neg m}(S) : \dist(r, F^{\neg m}(S)) = \Delta(S) \}$ denotes the set of unmatched robots at minimum distance from unmatched targets in $S$;
\item $R^{\safe}(S) = \{r\in R^{\neg m}(S) : \exists$ a safe trajectory from $r$ to $t$, $t\in F^{\neg m}(S)\}$ denotes the subset of $R^{\neg m}(S)$ containing only robots having a safe trajectory toward at least one target in $F^{\neg m}(S)$;
\end{itemize}

\noindent
If in $S$ both $R^{\neg m}(S)\neq \emptyset$ and $F^{\neg m}(S)\neq \emptyset$ then:
\begin{itemize}
\item $r^*(S)$ denotes the unmatched robot in $S$ that has to move toward an unmatched target still in $S$. If $R^{\safe}(S)\neq\emptyset$ then $r^*(S)$ is the robot of minimum view satisfying $\argmin_{r\in R^{\safe}(S)}\{dist(r,t): t\in F^{\neg m}(S)\}$ else $r^*(S)$ is selected from $R^{\neg m}(S)$ according to the minimum view (cf. Proposition~\ref{prop:asymmetric}).
%denotes the robot that has to move toward a target in $F^{\neg m}(S)$. Such a robot is selected from $R^{\safe}(S)$ if $R^{\safe}(S)\neq\emptyset$ according to its minimum distance to , otherwise from the superset $R^{\neg m}(S)$ according to the minimum view (cf. Proposition~\ref{prop:asymmetric}).

%\item the robot $r^*(S)$ that has to move toward an unmatched
%target inside $S$. Such a robot is selected from $R^{\safe}(S)$ if $R^{\safe}(S)\neq\emptyset$, otherwise from the superset $R^{\neg m}_{\Delta}(S)$. In case the set from which $r^*(S)$ is selected contains more than one element, $r^*(S)$ is uniquely identified according to the minimum view (cf. Proposition~\ref{prop:asymmetric}).
\end{itemize}  

\smallskip\noindent
Consider now the case in which there are more robots than targets within a sector $S$. Our approach will move one robot at a time in $S$ (always identified  as $r^*(S)$) toward a target in $F^{\neg m}(S)$ until all targets become matched. At that time, we will get $R^{\neg m}(S)\neq \emptyset$ and $F^{\neg m}(S)= \emptyset$. Then, our strategy will move the remaining robots in $R^{\neg m}(S)$ toward points on the robot-ray belonging to $S'$, where $S'$ is the next sector with respect to $S$ according to the clockwise direction. We then extend the previous notation as follows:

%In particular, each robot $r\in R^{\neg m}(S)$ will be moved toward the point $t$ on the boundary between $S$ and $S'$ (belonging to $S'$) and lying on the same circle $C^\uparrow_i$ where $r$ lies, if this does not cause collisions. Otherwise  trajectory is not simply a rotation along $C^\uparrow_i$ but it is first composed by a radial movement so as $S'$ then become safely reachable by a simple rotation. The strategy for moving robots toward the boundary robot-ray of $S'$ is the same as before, that is based on selecting one robot per time first among those admitting a rotational trajectory toward the boundary, and in case among those at shortest distance from the the boundary. For these reasons, we extend the previous notation as follows:
\begin{itemize}
\item $R^{\safe}(S,S')= \{r\in R^{\neg m}(S) : r$ is lying on a circle $C_\downarrow^i$ and it can rotate along $C_\downarrow^i$ until reaching $S'$ without collisions$\}$ denotes the set containing any robot that can reach $S'$ by means of a simple rotation along the circle $C_\downarrow^i$ where it lies;
\item $r^*(S,S')$ denotes the unmatched robot in $S$ that has to move toward $S'$.  If $R^{\safe}(S,S')\neq\emptyset$ then $r^*(S)= \argmin_{r\in R^{\safe}(S,S')}\{dist(r,t): t\in S'\}$ else $r^*(S)$ is selected from $R^{\neg m}(S)$ according to the minimum view. 
\end{itemize} 
Procedure $\Distmin$ is given in Algorithm~\ref{alg:dist_min}. Its description can be found in the corresponding correctness proof provided in Lemma~\ref{lem:corr-DistMin}. Figure~\ref{fig:T8b} provides a partial illustration of how $\Distmin$ determines the pairs robot-target within one sector of the running example. 

\begin{algorithm}
\caption{$\Distmin$}\label{alg:dist_min}
\label{alg:distmin}
\begin{algorithmic}[1]
\begin{small}
\IF{ $\mult(c(R),R) < \mult(c(F),F)$ \label{l:1}}
	\IF{$d(r,c(R))$ is minimum among all robots in $R$, and $r$ is of minimum view in case of ties}
	\STATE $r$ moves toward $c(R)$~\label{l:3}
	\ENDIF
\ELSE
	\IF{$\exists$ sector $S$ s.t. $R^{\neg m}(S)\neq \emptyset$ and
	    $F^{\neg m}(S) \neq \emptyset$ ~\label{l:5} }
	    \IF { $R^{\safe}(S)\neq\emptyset$ }
	        \STATE $r^*(S)$ moves toward its target $f\in F^{\neg m}(S)$ 
	           along a safe trajectory~\label{l:7}
	    \ELSE
	        \IF { $r^*(S)$ and its target $f$ belong to a circle $C_\downarrow^i(R)$ }
	            \STATE $r^*(S)$ moves radially at half distance from $C_\downarrow^{i-1}(R)$ if this exists or from $c(R)$~\label{l:10}
	        \ELSE
	            \STATE $r^*(S)$ rotates clockwise at half distance from the closest robot-ray or from the closest robot if there is one on the way~\label{l:12}
%\linecomment{Alf}{poichè il boundary destro non fa parte di $S$ allora c'è sempre spazio a destra!}
	        \ENDIF
	    \ENDIF
	\ELSE

  	    \IF{$\exists$ sector $S$ s.t. $R^{\neg m}(S)\neq \emptyset$ ~\label{l:14}}
      		\STATE Let $S'$ be the next sector in clockwise order; 
	        \IF { $R^{\safe}(S,S')\neq\emptyset$ }
	            \STATE $r^*(S,S')$ rotates toward the robot-ray of $S'$~\label{l:17}
	        \ELSE	
	        	\STATE Let $C_\downarrow^i(R)$ be the circle to which $r^*(S,S')$ belongs to       
	            \STATE $r^*(S,S')$ moves radially at half distance from $C_\downarrow^{i-1}(R)$ if this exists or from $c(R)$~\label{l:19}
	        \ENDIF
        	\ELSE
		    \STATE Let $r$ be the robot on $c(R)$:~\label{l:21}
   			$r$ radially moves along the segment connecting
   			     $c(R)$ with the unique point left in $F^{\neg m}(S)$ 
   			     for some sector $S$, until distance 
   			     $\delta (\CB)$;~\label{l:22}
 
    	    \ENDIF
	\ENDIF
\ENDIF
\end{small}
\end{algorithmic} 
\end{algorithm}

\begin{figure}[ht]
\begin{center}
\scalebox{1.0}{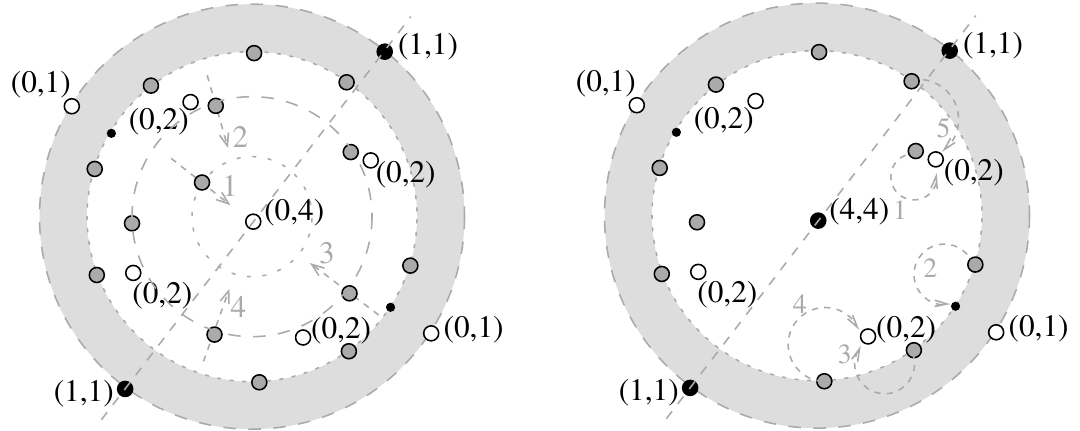 }
\end{center}
\caption{Task $T_8$. Any pair of integers close to points of $F$ represents multiplicities of robots and of targets, respectively. (left) Preliminary phase, the right multiplicity is formed on $c(F)$ (cf. Lines~\ref{l:1}-\ref{l:3} of  Algorithm~\ref{alg:distmin}). The numbers close to the arrows show the order in which robots move.
(right) Order of robots' movements toward targets within one sector. Notice that the gray arrows only show robot-target pairs and not trajectories: we recall that Algorithm~\ref{alg:distmin} uses sectorial paths as robots' trajectories.}
\label{fig:T8b}
\end{figure}

\begin{lemma}\label{lem:corr-DistMin}
Given a configuration $R$ belonging to $T_8\cap (\I\setminus \U(F))$, by repeatedly applying Procedure $\Distmin$ the pattern $F'$ can be formed.
\end{lemma}
\begin{proof}
According to Proposition~\ref{prop:asymmetric}, two robots with the same view cannot belong to a same sector $S\in\Sector(R)$. Hence, all moves allowed by Procedure $\Distmin$ involve at most one robot per sector as ties are always broken by means of the minimum view.

Lines~\ref{l:1}-\ref{l:3} consider the cases when the current multiplicity in the center $c(R)$ is less than that required in $c(F)$.  Notice that $\rho(F)>1$ by hypothesis, and this implies that in $c(F)$ there is a number of points which is multiple of $\rho(F)$. Since $\rho(R)$ divides $\rho(F)$, then the number of robots in each circle $C_{\downarrow}^i(R)$ divides $\rho(F)$, and hence the number of robots to be moved toward the center is always correctly determined by the procedure: this is $\rho(R)$ which divides $\rho(F)$ that in turn divides the number of robots in $c(F)$. 

Lines~\ref{l:5}-\ref{l:22} consider the cases when the multiplicity in the center (if any) is already correctly formed. In particular, lines~\ref{l:5}-\ref{l:12} are executed when there exists a sector $S$ in which there are both unmatched robots and unmatched targets. According to our definitions, the robot $r^*(S)$ elected to move follows a safe trajectory if it exists.  Once this robot starts moving, it will be moved until reaching its target, possibly within multiple LCM cycles. In fact, (1) robots admitting safe trajectories move before robots not admitting safe trajectories, and (2) the moves along safe trajectories assure to decrease the distances to the target; hence, in case of multiple LCM cycles, the moving robot $r^*(S)$, for each sector $S$, will be again chosen to reach its target. 
In case there is not a safe trajectory from $r^*(S)$ to the target, then the robot is slightly deviated (see moves at Lines~\ref{l:10} and~\ref{l:12}) to avoid collisions. Then, by Lemma~\ref{lem:degenerate}, the deviated robots admit safe trajectories and will be chosen again by the algorithm to be moved.%

Once each sector contains only unmatched robots or only unmatched targets, then unmatched robots in any sector $S$ are moved toward a point on the boundary of the next sector $S'$ in clockwise order (cf. Lines~\ref{l:14}-\ref{l:19}). As before, robots moved are first those elected that admit a safe trajectory toward the next (clockwise) sector, and then the remaining ones (which are deviated as before in order to avoid collisions). Notice that, as soon as the moved robot reaches the boundary, it enters into the next sector $S'$. As a consequence, the procedure processes this robot  when it will be elected in $S'$ to be moved either toward an unmatched target in the same sector, or toward the boundary of the successive (clockwise) sector. 

The last line (Line-\ref{l:22}) consider the cases when a robot must be moved from the center $c(R)$ whereas any other robot is matched. This case is processed at the end because, by definition, the center $c(R)$ does not belong to any sector. The robot is moved toward the last unmatched target in $F'$ until reaching the circle $\CB$ (by definition, along the trajectory there are no targets and hence no robots). Regardless whether it is stopped or not by the adversary, once it becomes active again, it will be processed as an unmatched robot by Lines~\ref{l:5}-\ref{l:12}.
\end{proof}

% --------------------------------------------------------
\subsection{Task $T_9$}
This task is devoted to finalize the pattern formation. It is characterized by the precondition $\pre_9=\prenove$, which means: there is a subset of $m\ge 2$ robots on $C(R)$ that form a regular $m$-gon, with $m$ divisor of $\rho(F)$; the unmatched robots with respect to $F$ are only those in $\Ann$ or on $\CT$; $F$ can be obtained by radial movements of the unmatched robots toward $C(R)$. 

Move $m_9$ makes such robots moving radially toward $C(R)$. As described in Section~\ref{ssec:algorithm:subdivision} (cf. description of $\Fin$), while robots move from $\CT$ to $C(R)$, the common reference system might be lost as soon as some robots reache $C(R)$. However, robots can always detect whether the configuration obtained by a radial projection of all robots in $\Ann \cup \CT$ to $C(R)$ produces $F$ or not as both $\Ann$ and $\CT$ can be determined just on the basis of $F$. This is the way to establish the value of variable $\xp$. Trivially, once all robots finish their movements, $\xw$ becomes true, that is $F$ is formed. Figure~\ref{fig:T9} provides an illustration of this task when it is applied to the running example.

\begin{figure}[ht]
\begin{center}
\scalebox{1.0}{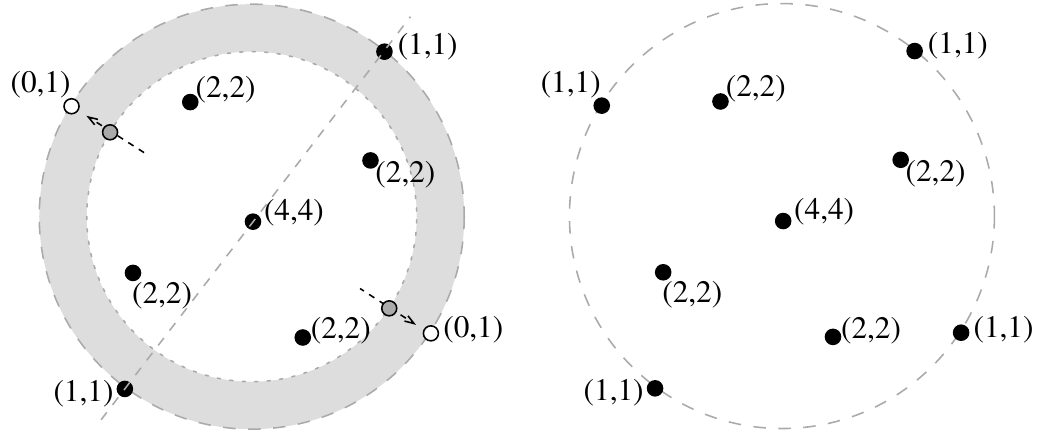 }
\end{center}
\caption{Task $T_9$: finalize the pattern formation.}
\label{fig:T9}
\end{figure}
%

% --------------------------------------------------------
\subsection{Task $T_{11}$}

This is actually not a real task. It is identified by variable $\xw$ which means $F$ is formed, hence robot must not move anymore. It guarantees the obtained configuration does not change anymore.

% --------------------------------------------------------
\subsection{Task $T_5$}
This task is complementary with respect to $T_3$ and $T_4$ as it is invoked when the number $m$ of robots on $C(R)$ is too small with respect to $\rho(F)$, that is $m$ is smaller than the minimal prime factor of $\rho(F)$. In this case, the configuration is necessarily asymmetric and, consequently, one robot per time is moved from $C_\downarrow^2(R)$ toward $C_\downarrow^1(R) = C(R)$ by means of move $m_5$. Robots are moved toward $C(R)$ avoiding forbidden points for $C(R)$. These forbidden points are similar to those introduced in the description of Task $T_2$: a point of $C(R)$ is forbidden if it may form a regular $n$-gon along with the points occupied by some robots already located on $C(R)$. Again, avoiding forbidden points ensures that when a robot reaches $C(R)$ all robots in such a circle are non-equivalent; this helps to ensure that no unsolvable configurations are created.

\begin{figure}[ht]
\begin{center}
\scalebox{0.95}{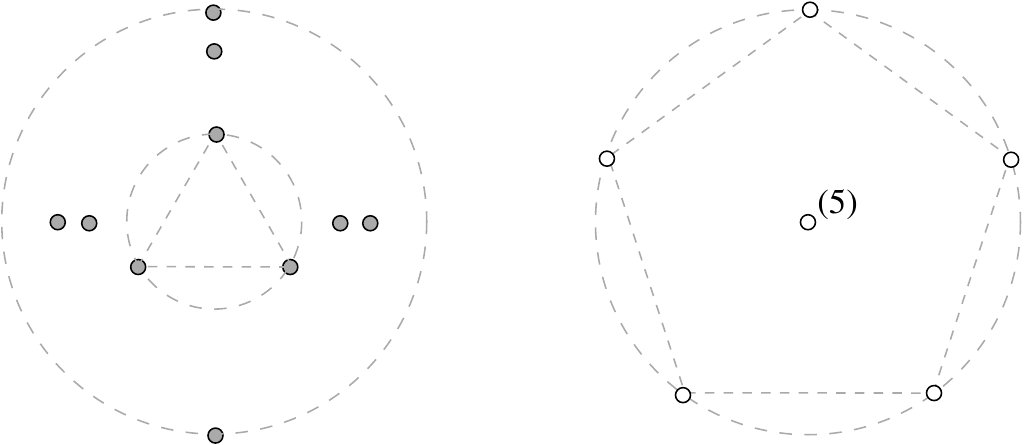 }
\end{center}
\caption{The input for the \pf\ problem that we use as secondary running example. Notice that the initial configuration $R$ is composed of 10 robots and $rho(R)=1$, while the pattern $F$ has symmetricity $\rho(F)=5$ (numbers close to points refer to multiplicities).
}
\label{fig:input-bis}
\end{figure}
\begin{figure}[ht]
\begin{center}
\scalebox{0.8}{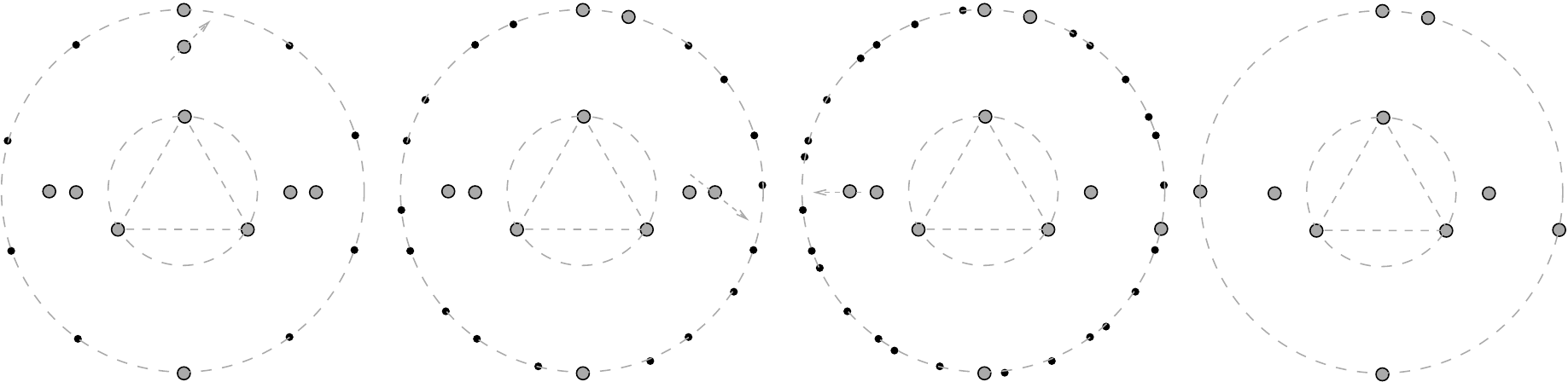 }
\end{center}
\caption{Task $T_5$ applied to the input specified in Figure~\ref{fig:input-bis}: bring robots to $C(R)$ until $|\partial C(R)|$ divides $\rho(F)$. Gray circles represent robots %, white circles represent point of $F$, 
and small black dots represent forbidden points for $C(R)$.
}
\label{fig:T5}
\end{figure}

An example of application of $m_5$ can be seen in Figure~\ref{fig:input-bis}.
 There $\rho(R)=1$ whereas $\rho(F)=5$. Moreover, $|\partial C(R)|=2$ is smaller than the minimal prime factor of $\rho(F)$, which is five. So, $\xf=\true$, whereas $\xc=\false$. The configuration is then in $T_5$ as it can be easily checked: $\xw$, $\xg$, $\xp$ are false, that is the configuration does not belong to $T_{11}$, $T_{10}$, $T_9$, respectively; $\xduno$ and $\xddue$ are true, hence $R$ is not in $T_8$ nor in $T_7$; $\xt=\false$ since 2 is not a divisor of $\rho(F)$, hence $R$ is not in $T_6$. 
 
Once three robots, one per time, are moved to $C(R)$ by means of $m_5$, the configuration in Figure~\ref{fig:T5}, right side, is obtained. It belongs to $T_7$ as a regular 5-gon must be formed on $C(R)$ since $\xduno$ and $\xddue$ are now false as well as $\xu$.

% --------------------------------------------------------
\subsection{Task $T_7$}
This task is meant to create a regular $m$-gon on $C(R)$. It is a sort of generalization of $T_6$ as it is used when $m=|\partial C(R)|$ is the minimal prime factor of $\rho(F)$, that is $\xduno$ and $\xddue$ are both false. By means of move $m_7$ the $m$ robots on $C(R)$ are opportunely rotated so as to obtain a regular $m$-gon. Once this happens, $\xm$ becomes false and $\xu$ becomes true.
Actually, Procedure $\Circle()$ applies the same movements of the algorithm proposed in~\cite{FPS08} where the problem was to uniformly distribute robots along a ring. The only difference is that here the ring is the circumference of $C(R)$, hence to guarantee the correctness of the algorithm we need to guarantee that $C(R)$ does never change.

%______________________________________________
%			Circle Formation
%______________________________________________

\begin{algorithm}
\caption{$\Circle(\alpha)$}\label{alg:circle}
\begin{algorithmic}[1]
\begin{small} 
\STATE Let $r'$, $r$ and $r''$ be three consecutive (clockwise) robots on $C(R)$;
\STATE Let $p$ be the antipodal point of $r'$;
\STATE Let $q$ be the point on $C(R)$ preceding $r''$ wrt the clockwise direction such that $\angolo(q,c(R),r'') = \alpha$;
\IF{$\angolo(r,c(R),r'') > \alpha$}
\STATE $r$ rotates clockwise toward the closest point among $p$ and $q$;
\ENDIF
\end{small}
\end{algorithmic} 
\end{algorithm}

%______________________________________________
%           Correttezza
%______________________________________________

Given a configuration $R$, with $|R|\geq 3$, let $\alpha=2\pi / |\partial C(R)|$,
and let $r'$, $r$ and $r''$ be three consecutive (clockwise) robots on $C(R)$. The following lemma can be stated.

\begin{lemma}\label{lem:rot}
Let $p$ and $q$ be the two points calculated by a robot $r$ when running algorithm $\Circle(\alpha)$. If $r$ has to move, it will reach $q$, within a finite number of LCM cycles.
\end{lemma}

\begin{proof}
If $p$ is not in between $r$ and $q$ then the statement clearly holds as all moving robots follow the clockwise direction, and hence within different LCM cycles the target to reach either is unchanged or it is further (clockwise) than $q$ with respect to the starting position of $r$, that is $r$ reaches (and possibly overpasses) $q$. 

When $p$ is in between $r$ and $q$ then $r$ must stop at $p$, and eventually $r$ reaches $p$.
In this case, since  $|R|\geq 3$, necessarily $\angolo(r',c(R),r) > \alpha$, that is, also robot $r'$ must move.

Consider the points $q'$ and $q''$ on $C(R)$ antipodal to $q$ and $r''$, respectively. When $r'$ moves, it cannot overpass $q''$, however, by construction, once $r$ has reached $p$, then $q'$ is met by $r'$ before reaching $q''$. It follows that as soon as $r'$ reaches $q'$ then $r$ is free to reach $q$. 
\end{proof}

By combining the result of Lemma~\ref{lem:rot} with the correctness proof of the Circle Formation algorithm given in~\cite{FPS08}, the following corollary holds.

\begin{corollary}\label{cor:circle}
%Algorithm $\Circle$ is correct.
Let $R$ be a configuration belonging to $T_7\cap \IA$ with $m$ robots on $C(R)$. By repeatedly applying Algorithm $\Circle$, configuration $R$ is transformed into a configuration $R'$ having a regular $m$-gon on $C(R)$.
\end{corollary}

\begin{proof}
The proof simply follows by observing that algorithm $\Circle$ operates the same movements of those in~\cite{FPS08} but with the further constraint to not changing $C(R)$. However, Lemma~\ref{lem:rot}, proves that eventually each moving robot will reach the destination imposed in~\cite{FPS08}. It means that within multiple (but finite) LCM cycles, each moving robot behaves like in~\cite{FPS08}.
\end{proof}

%______________________________________________

%
\begin{figure}[ht]
\begin{center}
\scalebox{0.95}{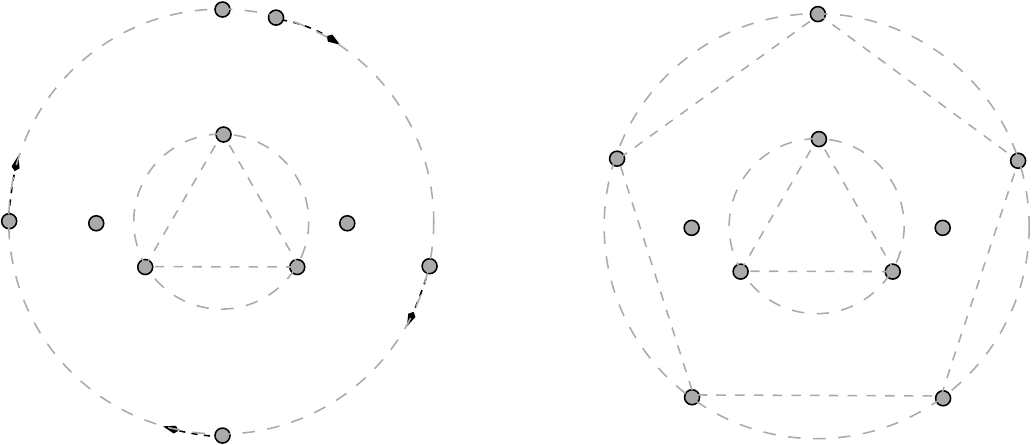 }
\end{center}
\caption{Task $T_7$. Create a regular $m$-gon on $C(R)$.
}
\label{fig:T7}
\end{figure}

Considering the running example of Figure~\ref{fig:T7} (left side), the robots on $C(R)$ are opportunely rotated in the clockwise direction so as to obtain a configuration with a regular pentagon on $C(R)$. In particular, the only robot that will never move in the specific configuration is the top-most one since the angle it forms in $c(R)$ with the clockwise neighbor is smaller than $\frac {2\pi} 5$. All other robots, will rotate eventually. Once configuration in Figure~\ref{fig:T7} (right side) is obtained, it means $\xu=\true$. Predicates $\xg$, $\xa$ and $\xddue$ did not change their values, whereas variables $\xw$ and $\xp$ are clearly false. Hence the configuration cannot belong to $T_9$, $T_{10}$ and $T_{11}$. Since $\neg \xddue \Rightarrow \neg \xduno$, then the configuration belongs to $T_8$.

% --------------------------------------------------------
\subsection{Task $T_{10}$}

This is actually not a real task. It solves \pf by exploiting other algorithms (namely $\Gathering$ from~\cite{CFPS12} and $\Leader$ from~\cite{CDN19}) when $F$ is composed of one point with multiplicity $|R|$, that is $\rho(F) = |R|$, or when $\rho(F) = 1$, respectively. Notice that $P_{10}$ depends only on $F$ and not on the current configuration. This implies that once one algorithm among $\Gathering$ and $\Leader$ starts, it will be invoked to process the configuration until the pattern is formed.

% ===================================================
% Correctness
% ===================================================
\section{The algorithm for \pf: correctness}\label{sec:correctness}

In this section we prove the correctness of the provided algorithm. According to the proposed methodology (cf. Section~\ref{sec:methodology} and in particular to Claim~\ref{claim:correctness}), it is done by proving that each property in Table~\ref{tab:properties} holds.

% ------------------------------ inizio tabella mosse
\begin{table*}[ht]
%\small
\caption{Properties underlying the correctness}
\label{tab:properties}
\bgroup
\def\arraystretch{1.3}%  1 is the default, change whatever you need
\setlength{\tabcolsep}{5pt}
\begin{center}
  \begin{tabular}{  l l  p{0.85 \textwidth}  }
    
    $\h_1$  & = &  \textit{for each configuration in $\IA$ at least one 
                        predicate $P_i$ is true and
                        for each $i\neq j$, $T_i \cap T_j= \emptyset$;} \\ 
      
    $\h_2$  & = &  \textit{configurations in $\Unew(F)$ are not generated 
                        by $\A$, i.e. $\IA\cap \Unew(F)=\emptyset$ - 
                        this means that given $R$ and $F$ as input,  
                        each generated configuration $R(t)$, $t>0$, 
                        must ensure that $\rho(R(t))$ divides $\rho(F)$;} \\ 

    $\h_3$  & = &  \textit{for each class $T_i$, the classes reachable from $T_i$ by means of a transition are exactly those represented in the transition graph $G$}  (i.e., the transition graph is correct);\\ 

    $\h_{3'}$  & = &  \textit{each transition not leading to $T_F$ is 
                           stationary, almost-stationary, or robust, 
                           while each transition leading to $T_F$ is
                           stationary;} \\                       

    $\h_{3''}$  & = &  \textit{the algorithm is collision-free;}
                       \\ 

    $\h_{4}$  & = &  \textit{possible cycles in the transition graph $G$ (including self-loops but excluding the self-loop in $T_F$) must be performed a finite number of times.}                 

\end{tabular}
\end{center}
\egroup
\end{table*}

\smallskip\noindent
Concerning property $\h_1$,  since the tasks' predicates $P_1,P_2,\ldots,P_{11}$ used by the algorithm have been defined as suggested by Equation~\ref{eq:predicates}, it holds according to Remark~\ref{rem:pre-i}. 

Since properties $\h_2$, $\h_3$, $\h_{3'}$, $\h_{3''}$ and $\h_4$ (the last limited to self-loops only) must be proved for each transition/move, then in the following we provide a specific lemma for each task. %(except the lemma for $T_{10}$, which refers to previous algorithms). 
It is worth to point out that, according to Remark~\ref{rem:simplified-transitions}, if one of such lemmata analyzes a task - say $T_i$ - and we have already proved that all transitions toward $T_i$ are stationary or almost-stationary or robust, then during the analysis of $T_i$ we can basically ignore possible pending moves. % and moving robots due to possible previous moves. 

A final theorem (cf. Theorem~\ref{th:correctness}) will make use of all these lemmata and will also prove the remaining part of property $\h_4$ concerning cycles that are not self-loops. 
As last remark, we remind that properties $\h_{3'}$ and $\h_{3''}$ are desirable but not necessary to prove the correctness of the algorithm. As we are going to see, in a few cases we cannot guarantee them.

\begin{lemma}\label{lem:corr-T10}
Let $R$ be a stationary configuration in $T_{10}$. From $R$ the algorithm eventually leads to a stationary 
configuration belonging to $T_{11}$.
\end{lemma}
\begin{proof}
Since $\xg$ holds, we have two cases: either $\rho(F)=1$ or $F$ contains only one element with multiplicity $|R|$.
In the first case move $m_{10}$ consists in calling the $\Leader()$ algorithm given in~\cite{CDN19}. In the second case, move
$m_{10}$ consists in applying the algorithm $\Gathering()$ given in~\cite{CFPS12}. Since the predicate only depends on $F$, its value never changes then one of the two algorithms can be applied until forming pattern $F$. Concerning the correctness of the algorithms we refer the reader to the proofs given in~\cite{CFPS12} and~\cite{CDN19}, respectively.
\end{proof}

\begin{remark}
As $\xg$ only depends on $F$ and not on the current configuration, from now on we can always consider variable $\xg$ as false since the movements of robots cannot change its value. It also follows that no transitions can lead to $T_{10}$ apart for self-loops.
\end{remark}

\begin{lemma}\label{lem:corr-T9}
Let $R$ be a stationary configuration in $T_9$. From $R$ the algorithm eventually leads to a stationary 
configuration belonging to $T_{11}$.
\end{lemma}
\begin{proof}
Move $m_9$ aims to finalize the pattern formation by performing only radial movements of robots from $\CT \cup \Ann$ to $C(R)$.
 \begin{itemize}

\item[$\h_2$:] During this task, since move $m_9$ does not remove any robot from $C(R)$, then $C(R)$ does not change and $\neg \xm$ remains  true. Moreover, since the movement is radial and by the fact that the computation of $\CT$ depends only on $F$, $\xp$ remains true during all the movements. Similarly, $\xw$ remains false until the last robot reaches $C(R)$. This means that it is always possible to solve PF when $\neg \xm \wedge \xp$ holds. It is enough to radially move all robots from $\CT \cup \Ann$ to $C(R)$ (which is exactly what move $m_9$ does). Hence, independently on the activation of the robots, the incurred configurations until $F$ is formed are all solvable, that is none of them belongs to $\U(F)$.

\item[$\h_3$:] as observed, during the move $\neg \xm$, $\xp$ and $\neg \xw$ remain true, then no other tasks can start. If robots are stopped during their movement by the adversary, the configuration remains in $T_{9}$. This defines a self-loop in $T_{9}$. % which is almost-stationary as there might be robots moving but  the same trajectories will be assigned to all robots in $\CT \cap \Ann$. 
If all robots involved by move $m_9$ reach their target on $C(R)$ then $\xw$ becomes true and the configuration is in $T_{11}$. 

\item[$\h_{3'}$:] If the configuration remains in $T_9$ after applying move $m_9$, the set of robots involved by the move as well as their trajectories do not change, hence the self-loop of $T_9$ is almost-stationary. Once all the robots in $\CT \cup \Ann$ reach $C(R)$ (that is $F$ is formed and the configuration is in $T_{11}$) the configuration is stationary.

\item[$\h_{3''}$:] actually two robots on the same ray can potentially collide, but this is not a problem as at their destination there must be a multiplicity, as $\xp$ holds.\footnote{We remind that property $\h_{3''}$ - as well as $\h_{3'}$ - are desirable but not necessary. As we are going to prove, the current case is actually the only one where property $\h_{3''}$ might be violated.} 

\item[$\h_4$:] when the self-loop is traversed, the overall distance of the robots involved by move $m_9$ to $C(R)$ is decreased. Then, eventually, it becomes zero and all such robots will be on $C(R)$.
\end{itemize}
\end{proof}

%\linecomment{Alf}{In realtà da $T_9$ si potrebbe anche togliere il predicato $\neg \xm$ dato che $\xp$ non necessità dell'embedding. Da valutare se questa semplificazione semplifica anche la discussione relativa alla metodologia o meno.}
%\linecomment{Alf}{ATT.NE: Il predicato $\xm$ è fondamentale in $T_9$ per non saltare da $T_7$ in $T_{9}$ nel during.}

\begin{lemma}\label{lem:corr-T8}
Let $R$ be a stationary configuration in $T_8$. From $R$ the algorithm eventually leads to a solvable and stationary 
configuration belonging to $T_9$ or $T_{11}$.
\end{lemma}
\begin{proof}
The aim of the task is to form pattern $F'$ rather than $F$. This is done so as the embedding of $F$ on $R$ is maintained thanks to the $k$-gon on $C(R)$ (cf. description of sub-problems $\RS$ and $\PPF$ of Section~\ref{ssec:algorithm:subdivision}). Note that $\neg \xduno \wedge \xu \Rightarrow \neg \xm$, hence $\xp$ must be false as otherwise the configuration would be in $T_9$.

 \begin{itemize}
\item[$\h_2$:] during the movements, as robots in $C(R)$ remain unchanged (and so $C(R)$ itself), $\rho(R)$ can be at most $|\partial C(R)|$ or a divisor of it. Being $|\partial C(R)|$ a divisor of $\rho(F)$ (since $\neg \xm$ holds), then no unsolvable configurations with respect to the symmetricity (cf. Theorem~\ref{th:rho}) can be generated. Moreover, by Lemma~\ref{lem:corr-DistMin}, if move $m_8$ leads to create a multiplicity, this is on a point corresponding to a multiplicity in $F'$, and also its size would not be greater than that specified by $F'$. Hence, no unsolvable configurations are created on this respect as well.

\item[$\h_3$:] during the movements, $m_8$ does not change the values of $\xw$, $\xm$, and $\xduno$, that are $\false$, and those of $\xa$ and  $\xu$, that are $\true$, as no robots are moved neither toward nor from $C(R)\cup \Ann$. Moreover, by definition $\xp$ remains false until the last robot reaches its destination, that is once $F'$ is formed. The configuration is then always in $T_8$ until $F'$ is formed. %, and the self-loops are almost-stationary according to the correctness of $\Distmin$, cf. Lemma~\ref{lem:corr-DistMin}. 
As soon as the last robot reaches its destination, the configuration satisfies $\xp$. Hence, if $F'$ is different from $F$ (in case there are robots on $\CT$), then the configuration is in $T_{9}$, otherwise the configuration is in $T_{11}$.
\item[$\h_{3'}$:] except for the robots moved by $\Distmin$, no other robot is moved (the only possible ones are those on $\partial C(R)$, not affected by $m_8$), then, when $\xp$ holds the configuration is stationary. Whereas, if the configuration remains in $T_8$ after applying move $m_8$ and it is non-stationary, the set of robots involved by the move does not change but their trajectories could.  This may happen when robots deviate to avoid collisions. Hence, the self-loop in $T_8$ is robust.
\item[$\h_{3''}$:] by Lemma~\ref{lem:corr-DistMin}, procedure $\Distmin$ avoids collisions.
\item[$\h_4$:] if a moving robot is stopped by the adversary during its movement, the configuration remains in $T_{8}$ and the robot will be moved again.  By Lemma~\ref{lem:corr-DistMin}, the total distance of the robots from their target decreases. Hence, the self-loop of $T_8$ can be traversed only a finite number of times.
%By Lemma~\ref{lem:corr-DistMin}, for each sector $S$ all the robots in $R^{\neg m}(S)$ will reach their destination, eventually.
\end{itemize}
\end{proof}

The next lemmata refer to the \RS subproblem, that is to tasks $T_1$, $T_2$, $\ldots$, $T_7$. All those tasks operates on configurations in $\I \setminus \U(F)$, that is solvable configurations without multiplicities, and as we are going to show each of them generates a configuration in $\I \setminus \U(F)$. Non-initial configurations are instead managed only by tasks $T_8$, $T_9$, $T_{10}$, $T_{11}$ and, as shown in the above lemmata, they never generate configurations in $T_1$, $T_2$, $\ldots$, $T_7$.

\begin{lemma}\label{lem:corr-T7}
Let $R$ be a stationary configuration in $T_7\cap (\I \setminus \U(F))$. From $R$ the algorithm eventually leads to a 
stationary configuration in $\I\setminus \U(F)$ belonging to $T_8$, $T_9$ or $T_{11}$.
\end{lemma}
\begin{proof}
Let  $k=|\partial C(R)|$ be the minimal prime factor of $\rho(F)$. Then $\neg \xddue$ holds and this implies that $\neg \xduno$ holds too.
The $k$ robots on $C(R)$ are rotated by $m_7$ which applies Procedure $\Circle$ so as to obtain a regular $k$-gon without affecting $C(R)$. Once this happens, $\xm$ becomes false and $\xu$ becomes true. 

 \begin{itemize}

	\item[$\h_2$:] as $k=|\partial C(R)|$ is the minimal prime factor of $\rho(F)$, then $k$ is prime. This implies either $\rho(R)=1$ or $\rho(R)=k$. This last possibility can happen only at the end of this task when $\xu$ becomes true, whereas $\rho(R)=1$ for each generated configuration $R$ during the task. Moreover, as Procedure $\Circle$ guarantees to not create multiplicities, then no unsolvable configurations can be generated. 
	\item[$\h_3$:]  
	the move only involves robots in $C(R)$ along $C(R)$, hence $\xa$, that is $\true$, and $\xddue$, that is $\false$ do not change their values.   Variable $\xw$ can become true only once the $k$-gon is formed. Similarly $\neg \xm \wedge \xp$ and $\xu$ remain false as long as the $k$-gon is not formed. Hence, if robots are stopped during their movements, the configuration remains in $T_7$. % and the defined self-loop is almost-stationary as the same trajectories will be re-assigned to the same moving robots. 	
	Once the $k$-gon is formed then $\xm$ becomes $\false$ and $\xu$ becomes $\true$. Since $\neg \xduno$ holds, this implies that the configuration can be in $T_8$ (not in $T_1$, $T_2$, $T_3$, $T_4$, $T_5$, $T_6$), in $T_9$, or in $T_{11}$ according to possible changes of the values of $\xp$ and $\xw$.
	\item[$\h_{3'}$:] at the end of the task the configuration is clearly stationary as the only robots allowed to move are those on $C(R)$ and they do not move once $\xu$ holds. If the configuration remains in $T_7$ after applying move $m_7$, the %set of robots involved by the move does not change. The 
	trajectory of a moving robot might be prolonged but always along the circumference of $C(R)$. Hence, the self-loop in $T_7$ is almost-stationary.
	\item[$\h_{3''}$:] in Procedure $\Circle$  no collisions are possible because the target of a move is always between the moving robot and the next (clockwise) robot on $C(R)$. 
	\item[$\h_4$:] the correctness of Procedure $\Circle$ provided in Corollary~\ref{cor:circle} guarantees the property.
\end{itemize}
\end{proof}

\begin{lemma}\label{lem:corr-T6}
Let $R$ be a stationary configuration in $T_6\cap (\I \setminus \U(F))$. From $R$ the algorithm eventually leads to a stationary configuration in $\I\setminus \U(F)$ belonging to $T_3$ or $T_9$.
\end{lemma}
\begin{proof}
There are exactly three robots on $C(R)$ (as $\xt=\true$) and $\xw=\false$. Note that $\rho(R)=1$, otherwise, if $\rho(R)=3$ (and then 3 is a divisor of $\rho(F)$) the configuration would not be in $T_6$ (it would be in $T_8$, because in this case $\neg \xduno \wedge \xu$ holds). Moreover $\rho(F)$ must be even as $\xt$ holds. By referring to the description of move $m_6$ note that $\alpha_1\neq 90^\circ$ as otherwise $r_2$ and $r_3$ are antipodal, against $\xm$. Moreover, $\alpha_1<90^\circ$ as otherwise the three robots would lie in half $C(R)$ hence defining a different smallest enclosing circle. Being $\rho(R)=1$, the configuration is asymmetric and hence robot $r_2$ can always be selected and moved toward its target without modifying $C(R)$. %Once $\alpha_1=90^\circ$, two robots will be antipodal and $\xm$ becomes false.

The configuration can start with an equilateral triangle on $C(R)$ (when three is not a divisor of $\rho(F)$), but as soon as $r_2$ moves, $\xu$ is false and remains false until the end of the task. %The only other variable that could change its value is $\xp$.

 \begin{itemize}
\item[$\h_2$:] since during this task $\rho(R)=1$ and no multiplicities are created, no unsolvable configurations can be generated.
\item[$\h_3$:] during the movement (i.e., before reaching the target), the variables involved in $\pre_6$ do not change their values. Hence the configuration cannot be in $T_9$ because of $\xm$. It cannot be in $T_8$ because of $\xu$. It cannot be in $T_7$ because of $\xt \Rightarrow \xddue$. Then the configuration remains in $T_6$ until the moving robot reaches the target. %, hence self-loops are stationary. 
At that point, $\xm$ becomes false. If $\xp$ is also true then the configuration is in $T_9$. By the same considerations as above, the obtained configuration cannot be in $T_8$ nor in $T_7$. It is not in $T_6$ nor in $T_4$ because of $\xm$.  It is not in $T_5$ because of $\xf$. Hence, it is in $T_3$ since $\pre_3$ holds.
\item[$\h_{3'}$:] the transitions to the tasks following $T_6$ are obviously stationary being $r_2$ the only moving robot. Whereas the self-loop is almost-stationary as the same robot along the same trajectory is moved at any time. 
\item[$\h_{3''}$:] by the definition of move $m_6$ no collision can be generated by $r_2$.
\item[$\h_4$:] the possible self-loops of this task will end as the total distance of the robot from its target decreases. 
\end{itemize}
\end{proof}

\begin{lemma}\label{lem:corr-T5}
Let $R$ be a stationary configuration in $T_5\cap (\I \setminus \U(F))$. From $R$ the algorithm eventually leads to a %n initial 
stationary configuration in $\I\setminus \U(F)$ belonging to $T_2$ or $T_7$.
\end{lemma}
\begin{proof}
At the beginning the configuration is necessarily asymmetric, that is $\rho(R)=1$, because the number of robots on $C(R)$ is less than the minimal prime factor of $\rho(F)$, being $\xf=\true$. Hence one robot per time is moved from $C_\downarrow^1(R)$ toward $C(R)$ by means of move $m_5$. In general, the movements are radial toward $C(R)$. Deviations are applied if the move may cause a collision on $C(R)$ or may potentially make the configuration symmetric. As alternative target we may consider the closest middle point in the clockwise direction between two consecutive forbidden points. In any case, $C(R)$ remains unchanged.
 \begin{itemize}
\item[$\h_2$:] since during this task $\rho(R)=1$ is guaranteed by avoiding forbidden points for $C(R)$, hence avoiding also to create multiplicities, no unsolvable configurations can be generated.
\item[$\h_3$:] during the movement of the robot $\xc= \false$, whereas both $\xf$ and $\xm$ are $\true$; variable $\xw=\false$ and variables both $\xduno$ and $\xddue$ are $\true$. Moreover $\xt=\false$ since 2 is not a divisor of $\rho(F)$. 
Then the configuration cannot be in any task from $T_6$ to $T_{11}$, so any configuration generated during the movement remains in $T_5$. %and self-loops are stationary.
Once the last robot reaches $C(R)$, variable $\xf$ becomes false.
The obtained configuration cannot belong to $T_{11}$ because of variables $\xw$, It cannot belong to $T_9$ and $T_8$ because of $\xm$ and $\xu$, respectively, as moving robots avoided forbidden points for $C(R)$. If $\xa=\true$, it belongs to $T_7$ since both $\xddue$ and $\xu$ are false, otherwise it belongs to $T_2$ since it cannot belong to $T_3,T_4$, and $T_6$ being $\xa=\false$.
\item[$\h_{3'}$:] the transitions to the tasks following $T_5$ are obviously stationary because there is only one moving robot per time. Whereas the self-loop is robust as the same robot will be moved but its target may change because of deviations to avoid forbidden points for $C(R)$. 
\item[$\h_{3''}$:] by the definition of move $m_5$ there is no robot between the moving robot and its target, then no collision can be generated.
\item[$\h_4$:] the possible self-loops of this task will end as the total distance of the robots from $C(R)$ decreases. 
\end{itemize}
\end{proof}

\begin{lemma}\label{lem:corr-T4}
Let $R$ be a stationary configuration in $T_4\cap (\I \setminus \U(F))$. From $R$ the algorithm eventually leads to a %n initial 
stationary configuration in $\I\setminus \U(F)$ belonging to $T_6$, $T_7$, or to a robust configuration in $\I\setminus \U(F)$ belonging to $T_2$.
\end{lemma}
\begin{proof}
In this task $\pre_4=\prequattro$ holds. Since $\rho(R)$ divides $\rho(F)$ by hypothesis, $\xc=\false$ and $\xm=\true$ imply that the current configuration $R$ is asymmetric. Moreover, according to the way predicates are defined, all preconditions concerning tasks $T_5$, $T_6$, $\ldots$, $T_{11}$ are false. In particular, this implies the following properties: 
\begin{itemize}
\item
being $\xc=\false$, from $\pre_5=\false$ we derive $\xf=\false$: this means that on $C(R)$ there is a number of robots greater than or equal to the minimal prime factor of $\rho(F)$.
\item
being $\xa=\true$, from $\pre_7=\false$ and $\pre_8=\false$ we derive that at least one variable among $\xduno$ and $\xddue$ must be true. This means that on $C(R)$ there is a number of robots which is not equal to the minimal prime factor of $\rho(F)$.
\end{itemize}
By combining the previous properties, we know that on $C(R)$ there is a number of robots greater than the minimal prime factor of $\rho(F)$.
%-----

%According to move $m_4$, one robot per time is moved since the configuration is certainly asymmetric as otherwise $\rho(R)$ could not divide $\rho(F)$ being $\xc$ false and $\xm$ true. 

According to move $m_4$,  the algorithm removes one robot at a time from $C(R)$ (without affecting $C(R)$ by opportunely removing non-critical robots) until exactly $p$ robots remain, where $p$ is the minimal prime factor of $\rho(F)$.
 \begin{itemize}
 \item[$\h_2$:] according to Procedure $\GoToC$, the robot $r$ on $C(R)$ of minimal view is straightly moved toward a suitable point on $\CT$. By similar arguments applied in the proof of Lemma~\ref{lem:GoCorrectness}, such a movement maintains the configuration asymmetric, that is its symmetricity equals one. Moreover, no multiplicities are created and hence no unsolvable configurations are generated. % and it is not in $\Unew(F)$. 
\item[$\h_3$:] as soon as $r$ starts moving, $\xa$ becomes false. We can distinguish two cases: either $r$ reaches its target on $\CT$ or it stops before. 

When $r$ reaches its target on $\CT$, each variable referring to $C(R)$ can be potentially influenced. Among those, certainly $\xduno$ and $\xddue$ can change; $\xf$ cannot change and hence it remains false; $\xt$ and $\xu$ can change; $\xm=\true$ and does not change; $\xw$ cannot change. Consequently, no configurations in $T_{11}$ nor in $T_9$ can be generated because of $\xm$. 
Concerning $T_8$, notice that the following implication holds $\neg \xduno \wedge \xu \Rightarrow \neg \xm$.
Hence, since $\xm= \true$ in $R$ and it does not change its value, then no configurations in $T_8$ can be generated ($P_8$ requires $\neg \xduno \wedge\xu= \true$).
If $\xddue$ becomes false, then task $T_7$ must be applied so as to evenly distribute robots on $C(R)$, hence making variable $\xu$ true. This is due to the fact that $\xm$ is false along the whole task. If $\xt$ becomes true, then task $T_6$ must be applied as 3 would not be the minimal prime factor of $\rho(F)$ and $\xm=\true$, that is there are no antipodal robots on $C(R)$. $T_5$ cannot be reached as $\xf$ cannot change and hence it remains false. If nothing changes, still task $T_4$ is applied.

When $r$ does not reach its target on $\CT$ (i.e., it is stopped by the adversary inside $\Ann$), $\xa$ becomes false. It can be easily observed that in this case only task $T_2$ can be reached. 
%Notice that any obtained configuration is stationary (only one robot is moving).

\item[$\h_{3'}$:] the transitions to tasks $T_6$ and $T_7$ as well as the self-loop are obviously stationary because there is only one moving robot per time which has to reach its target. Whereas the transition to $T_2$ is robust as the same robot will be moved by $m_2$ but its target may change because of deviations to avoid forbidden points for $\CT$. 
%the reached configuration is clearly stationary if $r$ reaches $\CT$. However, while $r$ moves inside $\Ann$, the configuration is in $T_2$. If $r$ stops inside $\Ann$, $T_2$ will keep on $r$ moving according to Procedure $\GoToC$. Hence the transition to task $T_2$ is almost-stationary.
\item[$\h_{3''}$:] collisions cannot occur according to Procedure $\GoToC$.
\item[$\h_4$:] the repeated application of $m_4$ eventually ends as the number of robots in $\partial C(R)$ decreases opportunely.
\end{itemize}
\end{proof}

\begin{lemma}\label{lem:corr-T3}
Let $R$ be a configuration in $T_3\cap (\I \setminus \U(F))$. From $R$ the algorithm eventually leads to a %n initial 
stationary configuration in $\I\setminus \U(F)$ belonging to  $T_8$ or to a configuration in $\I\setminus \U(F)$ belonging to $T_2$.
\end{lemma}
\begin{proof}
In this task $\pretre$ holds and all preconditions concerning tasks $T_4$, $T_5$, $\ldots$, $T_{11}$ are false. This means that from $\pre_3 \wedge \neg \pre_4$ it follows $\xm=\false$. That is, on $C(R)$ there exists a maximal set of $k$ robots regularly disposed, such that $k$ divides $\rho(F)$.
On $C(R)$ there must be more than $k$ robots as otherwise being $\xm$ true, $\neg \xduno \wedge \xu$ would be true as well and the configuration is instead in $T_8$.
The aim of the move is to keep on $C(R)$ only $k$ robots forming a regular $k$-gon (hence $C(R)$ is unchanged) and this is realized by means of Procedure $\GoToC$ that moves robots from $C(R)$ to $\CT$. 
According to move $m_3$, at most $\rho(R)$ robots per time can move.
 \begin{itemize}
 \item[$\h_2$:] according to $m_3$, the robots on $C(R)$ that should move are those of minimum view chosen among the set $\partial C(R)\setminus \M'$ if this is not empty, otherwise all robots on $C(R)$ of minimal view are chosen. The selected robots are straightly moved toward suitable points on $\CT$. As the move is basically the same applied in $T_2$ but involving robots from $C(R)$, similar arguments of the proof of Lemma~\ref{lem:GoCorrectness} guarantee to maintain the symmetricity of the configuration equal to a divisor of $k$ along all the movement. Since $k$ divides $\rho(F)$ and since no multiplicities are created, then no configuration in $\Unew(F)$ can be generated.   
\item[$\h_3$:] as soon as robots from $C(R)$ start moving, $\xa$ becomes false. We can distinguish three cases: 1) all the active robots involved by move $m_3$ reach their targets on $\CT$; 2) some of them do not reach their target but all of them start moving; 3) some of them have performed the \look phase but did not start moving yet. 

In case 1, if no variable changes its value, still task $T_3$ is applied. 
Otherwise, being the targets of the moving robots on $\CT$, then $\xw=\false$. Moreover, similarly to what is shown in the proof of Lemma~\ref{lem:GoCorrectness}, $\xp$ remains false as well because of the limit imposed by angle $\alpha$ established when calling $\GoToC$. Differently from $m_2$ now robots start moving from $C(R)$ which potentially may affect the definition of angle $\alpha$. However, since in $m_3$ only robots with the same minimum view can move concurrently, then they could not have been consecutive on $C(R)$ when $T_3$ started. This would in fact imply that all robots on $C(R)$ were equivalent, i.e. $\xu$ was true. Since $\xm$ was false, then also $\xduno$ would have been true, but then the configuration was in $T_8$ rather than in $T_3$.

Hence, the configuration is not in $T_9$ nor in $T_{11}$. 
If $\xduno$ becomes false, then the configuration might belong to $T_8$ if $\xu$ is true. Whereas if $\xu$ is false, it does not belong to $T_8$ nor to $T_7$ because $\xduno \Rightarrow  \xddue$. 
Variable $\xm$ cannot change its value and it is false, that is the configuration cannot belong to $T_6$ nor to $T_4$. It does not belong to $T_5$ because of $\xf$.  

In case 2, some robots are still inside $\Ann$, hence $\xa$ becomes false and task $T_2$ is invoked. 

In case 3, some robots might be still inside $\Ann$ in which case  task $T_2$ is invoked. Whereas if $\Ann$ is empty then task $T_3$ is still applied because more than $k$ robots are on $C(R)$, that is $\neg \xduno \wedge \xu$ is false.

\item[$\h_{3'}$:] the reached configuration is stationary if all robots reach $\CT$ (i.e. the configuration belongs to $T_8$). Otherwise there might be robots on $C(R)$ or in $\Ann$ concerning pending moves that will reach a suitable target on $\CT$, possibly computed from a different task and/or from a different configuration. By Lemma~\ref{lem:GoCorrectness}, we have that the transition to $T_2$ or even the self-loop are unclassified. This is due to the fact that when such transitions occur, there might be robots that have decided to move while they wouldn't have moved from the current configuration, or they would have moved with respect to a different trajectory. 
\item[$\h_{3''}$:] collisions cannot occur according to Procedure $\GoToC$.
\item[$\h_4$:] the repeated application of $m_3$ eventually ends as the number of robots in $\partial C(R)$ decreases until leaving a single regular $k$-gon.
\end{itemize}
\end{proof}

\begin{lemma}\label{lem:corr-T2}
Let $R$ be a configuration in $T_2\cap (\I \setminus \U(F))$. From $R$ the algorithm eventually leads to a %n initial 
stationary configuration in $\I\setminus \U(F)$ belonging to $T_4$, $T_6$, $T_7$, $T_8$, or to a configuration in $\I\setminus \U(F)$ belonging to $T_3$.
\end{lemma}

\begin{proof}

In this task $\pre_2=\predue$ holds and, consequently, all preconditions concerning tasks $T_3$, $T_4$, $\ldots$, $T_{11}$ are false. We recall that task $T_2$ is responsible for the correct removal of the robots from $\Ann$ toward $\CT$ in a configuration $R$. Hence $C(R)$ cannot change. Notice that in $R$, and during all the movements of all robots in $\Ann$, variable $\xa=\false$.  Moreover, there might be a number of robots equal to $\rho(R)$ that can move concurrently according to $m_2$ (this may occur when the processed configuration is symmetric). In particular, all robots in $\Ann$ closest to $c(R)$ and of minimal view move according to the trajectory computed by Procedure $\GoToC$. Note that at beginning of task $T_2$ the configuration could be non-stationary if the previous performed task is $T_3$.

\begin{itemize}
\item[$\h_2$:] if the configuration  $R$ is stationary, by Lemma~\ref{lem:GoCorrectness} no configuration in  $\Unew(F)$ is generated. If the configuration $R$ is non-stationary then the transition that led to $R$ was robust as generated from task $T_3$ by calling the same Procedure $\GoToC$. By similar arguments provided in the proof of Lemma~\ref{lem:GoCorrectness}, it is possible to show that unsolvable configurations cannot be generated.

\item[$\h_3$:] when all the moving robots reach their target, the configuration can be in $T_2$ again if there were more robots in $\Ann$ than the moved ones (e.g., when there are circles $C_{\downarrow}^i$ with different index $i$ inside $\Ann$). The configuration remains in $T_2$ as long as $\Ann \neq \emptyset$. Once this occurs, all the robots from $\Ann$ have reached $\CT$, and the resulting configuration $R'$ cannot be in $T_1$ as $\xc=\false$, in $T_9$ as $\xp$ remains false by the computed targets of $\GoToC$, in $T_{11}$ as $\xw$ remains false. In contrast, $R'$ could be in any class $T_3$, $T_4$, $T_6$, $T_7$, $T_8$, depending on the status of the variables. 

\item[$\h_{3'}$:] the transition to $T_3$ might be unclassified if $R$ was originally generated from $T_3$ itself by means of an unclassified transition. Otherwise, and for any other task different from $T_2$, the obtained configuration is stationary as variable $\xa$ changes its value only when all the robots in $\Ann$ reach $\CT$. The self-loop is unclassified as the set of robots involved by $m_2$ might change as well as their trajectories. However, Lemma~\ref{lem:GoCorrectness} ensures to make $\Ann$ empty eventually.  

\item[$\h_{3''}$:] 
%collisions are not possible as all the trajectories computed by $\GoToC$ are disjoint.
Lemma~\ref{lem:GoCorrectness} guarantees that any configuration obtained while performing task $T_2$ has no multiplicities. This implies that move $m_2$ is collision-free.

\item[$\h_4$:] if a robot does not reach its target because of the adversary, then the configuration remains in $T_2$, since no variable changes its value and $\Ann$ is not empty ($\xa$ remains false). However the moving robot decreases its distance to $\CT$, so task $T_2$ can be performed a finite number of times.
\end{itemize}
\end{proof}

\begin{lemma}\label{lem:corr-T1}
Let $R$ be a stationary configuration in $T_1\cap  (\I \setminus \U(F))$. From $R$ the algorithm eventually leads to a %n initial 
stationary configuration in $\I\setminus \U(F)$ belonging to $T_2$, $T_3$, $T_4$, $T_5$ or $T_6$.
\end{lemma}
\begin{proof}
In this task  $\xc=\true$, which means there is exactly one robot $r$ inside $\CB$ that must be moved. Robot $r$ is moved toward any point at distance $\delta(\CB)$ from $c(R)$. Hence $C(R)$ cannot change. If the robot does not reach its target, move $m_1$ is repeatedly applied to $r$ until a point on $\CB$ is reached by the robot. Then,  if $r$ does not occupy $c(r)$ its trajectory is radial. % Otherwise, we can consider the straight line joining $c(R)$ with the robot on $C(R)$ of minimum view. 

 \begin{itemize}
 \item[$\h_2$:] since $\xc=\true$, a single robot is in $\Int(\CB)$ and then the configuration admits symmetricity equal to one along all the movement of $r$, that is no unsolvable configurations are generated. 
\item[$\h_3$:] when $r$ reaches its target (possibly after applying move $m_1$ many times) all the variables remain unchanged except $\xc$ that becomes false. 
In particular, $\xw=\false$ as the moving robot remains confined on $\CB$, that is it has not reached a possible target point of $F$, regardless the embedding; $\neg \xm \wedge \xp$ remains false as neither robots on $C(R)$ nor robots in $\Ann$ moved and $r$ has not reached a possible target point of $F$; $\xa$ remains unchanged as robots in $\Ann$ are not moved; $\xduno$, $\xddue$ and $\xu$ remain unchanged as robots on $C(R)$ are not moved. We can then conclude that the final configuration can be only in $T_2$, $T_3$, $T_4$,  $T_5$, or $T_6$. % where $\neg \xc$ holds. 
\item[$\h_{3'}$:] the reached configuration is stationary as the only moving robot is $r$ and no other robot moves as all the variables remain unchanged during the movement. The self-loop is instead almost-stationary as the moving robot will be moved along the same trajectory until reaching $\CB$.
\item[$\h_{3''}$:] collisions cannot occur being $r$ the only robot inside $\CB$.
\item[$\h_4$:] the repeated application of $m_1$ eventually ends as the distance of $r$ to its target reduces.
\end{itemize}
\end{proof}
We are now ready to state the correctness of the algorithm.

\begin{theorem}[Correctness]\label{th:correctness}
Let $R$ be an initial configuration of \async robots with chirality, and $F$ be  any pattern (possibly with multiplicities) with $|F|=|R|$. Then, there exists an algorithm able to solve the Pattern Formation problem if and only if $\rho(R)$ divides $\rho(F)$.
\end{theorem}
\begin{proof}
($\Longrightarrow$) This is the case in which $\rho(R)$ does not divide $\rho(F)$. By Theorem~\ref{th:rho}, $F$ is not formable from $R$.

($\Longleftarrow$) 
According to Claim~\ref{claim:correctness}, it is sufficient to show that the provided algorithm fulfills all properties $\h_1,\ldots,\h_4$.
Concerning property $\h_1$, we have already pointed out at the beginning of this section that the tasks' predicates $P_1,P_2,\ldots,P_{11}$ used by the algorithm have been defined as suggested by Equation~\ref{eq:predicates}; then, according to Remark~\ref{rem:pre-i}, $\h_1$ holds. 
By Lemmata~\ref{lem:corr-T10}-\ref{lem:corr-T1} we have that both $\h_2$ (i.e., no unsolvable configurations are created) and $\h_3$ (i.e., the transition graph is exactly that represented in Figure~\ref{fig:transitions2}) are true. In order to conclude the proof, we need to prove property $\h_4$.
By Lemmata~\ref{lem:corr-T10}-\ref{lem:corr-T1} we have that self-loops are executed a finite number of times. According to the methodology proposed in Section~\ref{ssec:cycles}, we can focus on the simple cycles contained in the transition graph shown in Figure~\ref{fig:transitions2},
they are: $(T_2,T_{3})$, $(T_2,T_{4})$, $(T_2,T_6,T_{3})$, $(T_2,T_{4},T_6,T_{3})$. 

Considering node $T_2$, which belongs to all such simple cycles, we now show it can be entered a limited number of times. In particular, concerning the nodes involved in the simple cycles, $T_2$ can be reached from $T_3$ and $T_4$ by means of moves $m_3$ and $m_4$, respectively. Actually, both moves decrease $\partial C(R)$ of at least one robot. Since none of the involved tasks in the cycles increases $\partial C(R)$, then any cycle involving $T_2$ can occur a finite number of times.
\end{proof}

% ===================================================
% Conclusion
% ===================================================
\section{Conclusion}\label{sec:concl}
We have introduced a new methodology to tackle with distributed computing by mobile robots. The aim is to simplify both the design of the resolution algorithms and the writing of the required correctness proofs. 
In order to better explain the potentials of the methodology, we have considered the PF problem approached in~\cite{FYOKY15} as case study. On the one hand the resolution of PF along with the proposal of a new strategy allow to appreciate all facets arising by the new methodology. On the other hand, this work finally characterizes when PF can be solved by means of \async robots empowered with chirality.

Our new methodology opens a wide range of research for both proposing new resolution algorithms for tasks in distributed computing by mobile robots, and double-checking the correctness of existing ones by re-formulating/re-designing them, accordingly.
It can easily happen, in fact, that some very special cases that may occur while running an algorithm are instead neglected in the analysis of the correctness, mainly due to the intrinsic difficult to deal with asynchronous robots. Such scenarios are much more easy to be detected if the evolutions performed by the algorithms are related and constrained to formal logic predicates.

% ==================================================================

% =================================================
% Biblio
% =================================================
\bibliographystyle{splncs04}
\bibliography{../../global_references}

\begin{thebibliography}{10}
\providecommand{\url}[1]{\texttt{#1}}
\providecommand{\urlprefix}{URL }
\providecommand{\doi}[1]{https://doi.org/#1}

\bibitem{BDP17}
Bournat, M., Dubois, S., Petit, F.: Computability of perpetual exploration in
  highly dynamic rings. In: Lee, K., Liu, L. (eds.) 37th {IEEE} International
  Conference on Distributed Computing Systems, {ICDCS} 2017, Atlanta, GA, USA,
  June 5-8, 2017. pp. 794--804. {IEEE} Computer Society (2017).
  \doi{10.1109/ICDCS.2017.80}

\bibitem{BT16b}
Bramas, Q., Tixeuil, S.: {Brief Announcement}: {P}robabilistic asynchronous
  arbitrary pattern formation. In: Proc. 18th Int.'l Symp. on Stabilization,
  Safety, and Security of Distributed Systems (SSS). LNCS, vol. 10083, pp.
  88--93 (2016)

\bibitem{BT16}
Bramas, Q., Tixeuil, S.: {P}robabilistic asynchronous arbitrary pattern
  formation. CoRR  \textbf{abs/1508.03714} (2016),
  \url{https://arxiv.org/abs/1508.03714}

\bibitem{CDGJNRS19}
Cicerone, S., {Di Stefano}, G., Gasieniec, L., Jurdzinski, T., Navarra, A.,
  Radzik, T., Stachowiak, G.: Fair hitting sequence problem: Scheduling
  activities with varied frequency requirements. In: Algorithms and Complexity
  - 11th International Conference, {CIAC}. LNCS, vol. 11485, pp. 174--186.
  Springer (2019). \doi{10.1007/978-3-030-17402-6\_15}

\bibitem{CDN16}
Cicerone, S., {Di Stefano}, G., Navarra, A.: Asynchronous embedded pattern
  formation without orientation. In: Proc. 30th Int.'l Symp. on Distributed
  Computing (DISC). LNCS, vol.~9888, pp. 85--98. Springer (2016)

\bibitem{CDN18a}
Cicerone, S., {Di Stefano}, G., Navarra, A.: {``Semi-Asynchronous'': a new
  scheduler for robot based computing systems}. In: Proc. 38th {IEEE} Int.'l
  Conf. on Distributed Computing Systems, (ICDCS). pp. 176--187. {IEEE} (2018)

\bibitem{CDN19}
Cicerone, S., {Di Stefano}, G., Navarra, A.: Asynchronous arbitrary pattern
  formation: the effects of a rigorous approach. Distributed Computing
  \textbf{32}(2),  91--132 (2019)

\bibitem{CDN18c}
Cicerone, S., {Di Stefano}, G., Navarra, A.: Embedded pattern formation by
  asynchronous robots without chirality. Distributed Computing  \textbf{32}(4),
   291--315 (2019)

\bibitem{CFPS12}
Cieliebak, M., Flocchini, P., Prencipe, G., Santoro, N.: Distributed computing
  by mobile robots: Gathering. SIAM J. on Computing  \textbf{41}(4),  829--879
  (2012)

\bibitem{CieliebakP02}
Cieliebak, M., Prencipe, G.: Gathering autonomous mobile robots. In:
  Proceedings of the 9th International Colloquium on Structural Information and
  Communication Complexity (SIROCCO). vol.~13, pp. 57--72. Carleton Scientific
  (2002)

\bibitem{CGKKKT17}
Czyzowicz, J., Gasieniec, L., Kosowski, A., Kranakis, E., Krizanc, D., Taleb,
  N.: When patrolmen become corrupted: Monitoring a graph using faulty mobile
  robots. Algorithmica  \textbf{79}(3),  925--940 (2017).
  \doi{10.1007/s00453-016-0233-9}

\bibitem{DDN14}
D'Angelo, G., {Di Stefano}, G., Navarra, A.: Gathering on rings under the
  look-compute-move model. Distributed Computing  \textbf{27}(4),  255--285
  (2014)

\bibitem{DFPSY16}
Das, S., Flocchini, P., Prencipe, G., Santoro, N., Yamashita, M.: Autonomous
  mobile robots with lights. Theor. Comput. Sci.  \textbf{609},  171--184
  (2016)

\bibitem{DDFN18}
{D'Emidio}, M., {Di Stefano}, G., Frigioni, D., Navarra, A.: Characterizing the
  computational power of mobile robots on graphs and implications for the
  euclidean plane. Inf. Comput.  \textbf{263},  57--74 (2018)

\bibitem{DPVarx09}
Dieudonn{\'{e}}, Y., Petit, F., Villain, V.: Leader election problem versus
  pattern formation problem. CoRR  \textbf{abs/0902.2851} (2009),
  \url{http://arxiv.org/abs/0902.2851}

\bibitem{DPV10}
Dieudonn{\'{e}}, Y., Petit, F., Villain, V.: Leader election problem versus
  pattern formation problem. In: Proc. 24th Int.'l Symp. on Distributed
  Computing (DISC). LNCS, vol.~6343, pp. 267--281. Springer (2010)

\bibitem{DBO17}
Doan, H.T.T., Bonnet, F., Ogata, K.: Model checking of robot gathering. In:
  21st Int.'l Conf. on Principles of Distributed Systems {(OPODIS)} 2017.
  LIPIcs, vol.~95, pp. 12:1--12:16. Schloss Dagstuhl - Leibniz-Zentrum fuer
  Informatik (2018)

\bibitem{FPS08}
Flocchini, P., Prencipe, G., Santoro, N.: Self-deployment of mobile sensors on
  a ring. Theor. Comput. Sci.  \textbf{402}(1),  67--80 (2008)

\bibitem{FPSW08}
Flocchini, P., Prencipe, G., Santoro, N., Widmayer, P.: Arbitrary pattern
  formation by asynchronous, anonymous, oblivious robots. Theor. Comput. Sci.
  \textbf{407}(1-3),  412--447 (2008)

\bibitem{FPS12}
Flocchini, P., Prencipe, G., {Santoro (Eds.)}, N.: Distributed Computing by
  Oblivious Mobile Robots. Synthesis Lectures on Distributed Computing Theory,
  Morgan {\&} Claypool Publishers (2012)

\bibitem{FPS19}
Flocchini, P., Prencipe, G., {Santoro (Eds.)}, N.: Distributed Computing by
  Mobile Entities, Current Research in Moving and Computing, LNCS, vol. 11340.
  Springer (2019). \doi{10.1007/978-3-030-11072-7}

\bibitem{FYOKY15}
Fujinaga, N., Yamauchi, Y., Ono, H., Kijima, S., Yamashita, M.: Pattern
  formation by oblivious asynchronous mobile robots. {SIAM} J. Computing
  \textbf{44}(3),  740--785 (2015)

\bibitem{FYOKY17}
Fujinaga, N., Yamauchi, Y., Ono, H., Kijima, S., Yamashita, M.: Erratum:
  Pattern formation by oblivious asynchronous mobile robots (2017),
  \url{http://tcs.inf.kyushu-u.ac.jp/~yamauchi/manuscripts/E-FYOKY15.pdf}

\bibitem{GKMNZ08}
Gasieniec, L., Klasing, R., Martin, R.A., Navarra, A., Zhang, X.: Fast periodic
  graph exploration with constant memory. J. Comput. Syst. Sci.
  \textbf{74}(5),  808--822 (2008). \doi{10.1016/j.jcss.2007.09.004}

\bibitem{KK15}
Kawamura, A., Kobayashi, Y.: Fence patrolling by mobile agents with distinct
  speeds. Distributed Computing  \textbf{28}(2),  147--154 (2015).
  \doi{10.1007/s00446-014-0226-3}

\bibitem{M83}
Megiddo, N.: Linear-time algorithms for linear programming in
  {R}\({}^{\mbox{3}}\) and related problems. {SIAM} J. Comput.  \textbf{12}(4),
   759--776 (1983)

\bibitem{PMRM19}
Pattanayak, D., Mondal, K., Ramesh, H., Mandal, P.S.: Gathering of mobile
  robots with weak multiplicity detection in presence of crash-faults. J.
  Parallel Distrib. Comput.  \textbf{123},  145--155 (2019)

\bibitem{SY99}
Suzuki, I., Yamashita, M.: Distributed anonymous mobile robots: Formation of
  geometric patterns. {SIAM} J. Comput.  \textbf{28}(4),  1347--1363 (1999)

\bibitem{Welz91}
Welzl, E.: Smallest enclosing disks (balls and ellipsoids). In: Results and New
  Trends in Computer Science. pp. 359--370. Springer-Verlag (1991)

\bibitem{YUKY17}
Yamauchi, Y., Uehara, T., Kijima, S., Yamashita, M.: Plane formation by
  synchronous mobile robots in the three-dimensional euclidean space. J. {ACM}
  \textbf{64}(3),  16:1--16:43 (2017)

\end{thebibliography}

\end{document}